\documentclass[12pt]{article}
\usepackage[utf8]{inputenc}
\usepackage{color}
\usepackage{fullpage}
\usepackage{amsmath}
\usepackage{amssymb, amsfonts}
\usepackage{amsthm} 
\newtheorem{assumption}{Assumption}

\def\tr{\mathrm{tr}}
\def\diag{\mathrm{diag}}

\def\bm{\boldsymbol}
\usepackage{bm}
\usepackage{mathrsfs}
\usepackage{xr}
\usepackage{tablefootnote}
\usepackage{tabularx}
\usepackage{amsfonts}
\usepackage{graphicx,epstopdf}
\usepackage{float}
\usepackage{color}
\usepackage{booktabs}
\usepackage{natbib}

\usepackage{geometry}
\newtheorem{theorem}{Theorem}

\newtheorem{lemma}{Lemma}
\def\D{{\mathcal D}}

\numberwithin{equation}{section}
\newtheorem{coro}{Corollary}

\allowdisplaybreaks[4]

\def\X{{\bf X}}

\def\diag{\mbox{diag}}

\def\D{{\mathcal D}}

\def\I{{\bf{I}}}

\def\tilde{\widetilde}

\def\tr{{\rm tr}}
\def\bms{{\bf \Sigma}}

\def\diag{\hbox{diag}}

\def\D{{\bf D}}

\def\I{{\bf I}}

\def\R{{\bf R}}

\def\trans{^\top}

\def\0{{\bf 0}}
\def\1{{\bf 1}}
\def\vec{\mathrm{vec}}

\def\bth{{\bm\theta}}
\allowdisplaybreaks[4]

\title{Spatial-Sign based Maxsum Test for High Dimensional Location Parameters}
\author{Jixuan Liu, Long Feng, Ping Zhao and Zhaojun Wang\\
School of Statistics and Data Science, KLMDASR, LEBPS, and LPMC,\\ Nankai University}
\date{\today}

\begin{document}

\maketitle
\begin{abstract}
In this study, we explore a robust testing procedure for the high-dimensional location parameters testing problem. Initially, we introduce a spatial-sign based max-type test statistic, which exhibits excellent performance for sparse alternatives. Subsequently, we demonstrate the asymptotic independence between this max-type test statistic and the spatial-sign based sum-type test statistic (Feng and Sun, 2016). Building on this, we propose a spatial-sign based max-sum type testing procedure, which shows remarkable performance under varying signal sparsity. Our simulation studies underscore the superior performance of the procedures we propose.


\end{abstract}
\section{Introduction}
The testing of location parameters is a crucial and extensively researched area in multivariate statistics with a fixed dimension. The conventional Hotelling’s $T^2$ test is commonly applied, but it fails in high-dimensional scenarios where the variable’s dimension $p$ exceeds the sample sizes $n$. Consequently, numerous efforts have been made to develop a high-dimensional mean test procedure. One straightforward approach is to substitute the Mahalanobis distance with the Euclidean distance. For the two-sample location problem, \cite{bai1996effect} employed the $L_2$-norm of the difference between two sample means. \cite{chen2010two} eliminated some redundant terms in \cite{bai1996effect}'s test statistics and made no assumptions about the relationship between the dimension and sample sizes. \cite{srivastava2009test}, \cite{park2013test}, and \cite{feng2015two} suggested some scalar-invariant test statistics that replace the sample covariance matrix in Hotelling’s $T^2$ test statistics with its diagonal matrix. All these methods are built on the assumption of normal distribution or diverging factor models, which perform poorly for heavy-tailed distributions. For instance, the well-known multivariate t-distribution does not meet the above assumption. Therefore, numerous studies have also considered robust high-dimensional test procedures.

In traditional multivariate analysis, numerous methods have been developed to extend classic univariate rank and signed rank techniques to a multivariate context. A significant method is based on spatial signs and ranks, utilizing the so-called Oja median \citep{oja2010multivariate}. \cite{wang2015high} proposed a high-dimensional spatial-sign test that replaces the scatter matrix with the identity matrix for a one-sample location problem. Similarly, \cite{feng2016} proposed a high-dimensional spatial sign test that replaces the scatter matrix with its diagonal matrix, which has a scalar-invariant property. Furthermore, \cite{feng2016multivariate} considered the high-dimensional two-sample location problem based on the spatial-sign method. Feng, Liu, and Ma (2021) devised an inverse norm sign test that considers not only the direction of the observations but also the modulus of the observation. \cite{huang2023high} extended the inverse norm sign test for a high-dimensional two-sample location problem. \cite{feng2020high} demonstrated that the spatial-rank method also performs well for a high-dimensional two-sample problem. All these methods are constructed using the $L_2$-norm of the spatial median, which performs well under dense alternatives, meaning many variables have non-zero means. However, it is well-known that these sum-type test procedures perform poorly for sparse alternatives, where only a few variables have non-zero means.

In high-dimensional settings, numerous max-type test procedures have been introduced to detect sparse alternatives. \cite{CLX14} proposed a test statistic for the high-dimensional two-sample mean problem, which is based on the maximum difference between the means of two samples’ variables under the Gaussian or sub-Gaussian assumption. For heavy-tailed distributions, \cite{cheng2023} established a Gaussian approximation for the sample spatial median over the class of hyperrectangles and constructed a max-type test procedure using a multiplier bootstrap algorithm. However, their proposed test statistic is not scalar-invariant, and they did not provide the limit null distribution of their test statistic. The multiple bootstrap algorithm is also time-consuming. In this paper, we first introduce a novel Bahadur representation of the scaled sample spatial median and then construct a new max-type test statistic. We demonstrate that the limit null distribution of the proposed test statistic is still a Type I Gumbel distribution. We also establish the consistency of the proposed max-type test procedure. Simulation studies further illustrate its superiority over existing methods under sparse alternatives and heavy-tailed distributions.

In practical scenarios, it’s often unknown whether the alternative is dense or sparse. This has led to numerous studies proposing an adaptive strategy that combines the sum-type test and max-type test. For high-dimensional mean problems, \cite{xu2016adaptive} integrated different $L_r$-norms of the sample means. \cite{he2021} introduced a family of U-statistics as an unbiased estimator of the $L_r$-norm of the mean vectors, { covariance
matrices and regression coefficients}, demonstrating that U-statistics of different finite orders are asymptotically independent, normally distributed, and also independent from the maximum-type test statistic. \cite{feng2022asymptotic} relaxed the covariance matrix assumption to establish independence between the sum-type test statistic and the max-type test statistic. There are also many other studies showing the asymptotic independence between the sum-type test statistics and the max-type test statistic for other high-dimensional problems. For instance, \cite{wu2019adaptive} and \cite{wu2020regularization} examined the coefficient test in high-dimensional generalized linear models. \citet{feng2022a} looked at the cross-sectional independence test in high-dimensional panel data models. \citet{yu2022jasa} focused on testing the high-dimensional covariance matrix. \cite{feng2022testingwhite} considered the high-dimensional white noise test, while \citet{wang2023} looked at high-dimensional change point inference. \cite{ma2023adaptive} and \cite{Yu2023PE} considered testing the alpha of high-dimensional time-varying and linear factor pricing models, respectively.

However, all these methods assume a normal or other light-tailed distributions. There's a gap in the literature when it comes to considering the asymptotic independence between the sum-type test statistic and the max-type test statistic under heavy-tailed distributions.  In this paper, we first establish the asymptotic independence between \cite{feng2016}'s sum-type test statistic and a newly proposed spatial sign-based max-type test statistic for high dimensional one sample location parameter problem. We then propose a Cauchy combination test procedure \citep{liu2020} to handle general alternatives. Both simulation studies and theoretical results demonstrate the advantages of our newly proposed methods.

This paper are organized as follow. Section 2 introduce Bahadur representation of the scaled spatial median and establish the max-type test statistic. In section 3, we prove the asymptotic independence between the sum-type test statistic and the new proposed max-type test statistic and construct the Cauchy combination test procedure. Section 4 show some simulation studies. We give a real data application in Section 5. Some discussion are stated in Section 6. All the technical details are in the Appendix.

\textbf{Notations:} For $d$-dimensional $\boldsymbol x$, we use the notation $\Vert \boldsymbol x\Vert$ and $\Vert \boldsymbol x\Vert_\infty$ to denote its Euclidean norm and maximum-norm respectively. Denote $a_n\lesssim b_n$ if there exists constant $C$, $a_n\leq C b_n$ and $a_n \asymp b_n$ if  both $a_n\lesssim b_n$ and $b_n\lesssim a_n$ hold. Let $\psi_\alpha(x)=\exp \left(x^\alpha\right)-1$ be a function defined on $[0, \infty)$ for $\alpha>0$. Then the Orlicz norm $\|\cdot\|_{\psi_\alpha}$ of a $\boldsymbol X$ is defined as $\|\boldsymbol X\|_{\psi_\alpha}=\inf \left\{t>0, \mathbb{E}\left\{\psi_\alpha(|\boldsymbol X| / t)\right\} \leqslant 1\right\}$. Let $\operatorname{tr}(\cdot)$ be a trace for matrix, $\lambda_{min}(\cdot)$ and $\lambda_{max}(\cdot )$ be the minimum and maximum eigenvalue for symmetric martix. $\mathbf I_p$ represents a p-dimensional identity matrix, and $ \operatorname{diag}\{v_1,v_2,\cdots,v_p\}$ represents the diagonal matrix with entries $\boldsymbol v=(v_1,v_2,\cdots,v_p)$. For $a, b \in \mathbb{R}$, we write $a \wedge b=\min \{a, b\}$.

\section{Max-type test}

Let $\boldsymbol X_1, \ldots, \boldsymbol X_n$ be a sequence of independent and identically distributed (i.i.d.) $p$-dimensional random vectors from a population $X$ with cumulative distribution function $F_X$ in $\mathbb{R}^p$. We consider the following model:
\begin{align}\label{modelx}
\boldsymbol X_i=\boldsymbol\theta+v_i\mathbf\Gamma \boldsymbol W_i,
\end{align}
where $\boldsymbol\theta$ is location parameter, $\boldsymbol W_i$ is a
p-dimensional random vector with independent components, $E(\boldsymbol W_i)=0$, $\mathbf \Sigma=\mathbf \Gamma\mathbf\Gamma^\top$, $v_i$ is a nonnegative univariate random variable and is  independent with the spatial sign of $\boldsymbol W_i$. The distribution of $\boldsymbol X$ depends on $\mathbf\Gamma$ through the shape matrix. Model (\ref{modelx}) encompasses a wide range of frequently utilized multivariate models and distribution families, such as the independent components model \citep{nordhausen2009signed,ilmonen2011semiparametrically,yao2015sample} and the family of elliptical distributions \citep{hallin2006semiparametrically,oja2010multivariate,fang2018symmetric}.

In this paper, we focus on the following one sample testing problem
\begin{align}
H_0: \boldsymbol{\theta}=\mathbf{0} \text { versus } \quad H_1: \boldsymbol{\theta} \neq \mathbf{0}.
\end{align}
The spatial sign function is defined as $U(\mathbf{x})=\|\mathbf{x}\|^{-1} \mathbf{x} I(\mathbf{x} \neq \mathbf{0})$. In traditional fixed $p$ circumstance, the following so-called "inner centering and inner standardization" sign-based procedure is usually used (cf., Chapter 6 of \cite{oja2010multivariate})
$$
Q_n^2=n p \overline{\boldsymbol{U}}^T \overline{\boldsymbol{U}},
$$
where $\overline{\boldsymbol{U}}=\frac{1}{n} \sum_{i=1}^n \hat{\boldsymbol{U}}_i, \hat{\boldsymbol{U}}_i=U\left(\mathbf{S}^{-1 / 2} \boldsymbol{X}_{i}\right), \mathbf{S}^{-1 / 2}$ are Tyler's scatter matrix (cf., Section 6.1.3 of \cite{oja2010multivariate}). $Q_n^2$ is affine-invariant and can be regarded as a nonparametric counterpart of Hotelling's $T^2$ test statistic by using the spatial-signs instead of the original observations $\boldsymbol{X}_{i}$ 's. However, when $p>n, Q_n^2$ is not defined as the matrix $\mathbf{S}^{-1 / 2}$ is not available in high-dimensional settings.

In high-dimensional settings, \cite{wang2015high} proposed a method where Tyler's scatter matrix is replaced by the identity matrix. This led to the following test statistic:
\begin{align*}
T_{WPL}=\sum_{i<j}U^T(\bm X_i)U(\bm X_j).
\end{align*}
Building on this, \cite{feng2016} extended the method and introduced a scalar-invariant spatial-sign based test procedure, which will be detailed in subsection \ref{sumtypesec}.
Both methods utilize sum-type test statistics, which perform well under dense alternatives where many elements of $\bm \theta$ are nonzeros. However, their power decreases under sparse alternatives where only a few elements of $\bm \theta$ are nonzeros.

It is well-known that max-type tests have good performance under sparse alternatives \citep{CLX14}. Therefore, \cite{cheng2023} first provided the Bahadur representation of the classic spatial median $\tilde{\bm \theta}$, defined as:
\begin{align}
\tilde{\bm \theta}=\arg\min_{\bm \theta}\sum_{i=1}^n||\X_i-\bm \theta||.
\end{align}
They then proposed a max-type test procedure based on Gaussian approximation.
While this approach is robust and effective in high-dimensional settings, it loses scalar information of different variables and is not scalar-invariant. In real-world scenarios, different components may have entirely different physical or biological readings, and their scales would not be identical. Moreover, due to the unequal scale of $\tilde{\bm \theta}$, it is not possible to derive the limited null distribution of $||\tilde{\bm \theta}||_{\infty}$ even under weak correlation assumption.
In this paper, we first provide the Bahadur representation and Gaussian approximation of the location estimator proposed in \cite{feng2016multivariate}. We then propose a new max-type test statistic and establish its limit null distribution under some mild conditions.

\subsection{Bahadur representation and Gaussian approximation }

Motivated by \cite{feng2016multivariate}, we suggest to find a pair of diagonal matrix $\mathbf D$ and vector $\boldsymbol{\theta}$ for each sample that simultaneously satisfy
\begin{equation}\label{eq:HR}
\frac{1}{n} \sum_{i=1}^n U\left(\boldsymbol{\epsilon}_i\right)=0 \text { and } \frac{p}{n} \operatorname{diag}\left\{\sum_{i=1}^n U\left(\boldsymbol{\epsilon}_i\right) U\left(\boldsymbol{\epsilon}_i\right)^\top\right\}=\mathbf{I}_{p},
\end{equation}
where $\boldsymbol{\epsilon}_i=\mathbf{D}^{-1 / 2}\left(\boldsymbol{X}_i-\boldsymbol{\theta}\right)$. ($\mathbf D$, $\left.\boldsymbol{\theta}\right)$ can be viewed as a simplified version of HettmanspergerRandles (HR) estimator , which is proposed
 in \cite{hettmansperger2002practical},  without considering the off-diagonal elements of $\mathbf{S}$. We can adapt the recursive algorithm of \cite{feng2016multivariate} to solve Equation \ref{eq:HR}. That is, repeat the following three steps until convergence:\par
(i) $\boldsymbol{\epsilon}_i \leftarrow \mathbf{D}^{-1 / 2}\left(\mathbf{X}_i-\boldsymbol{\theta}\right), \quad j=1, \cdots, n$;\par
(ii) $\boldsymbol{\theta} \leftarrow \boldsymbol{\theta}+\frac{\mathbf{D}^{1 / 2} \sum_{j=1}^n U\left(\boldsymbol{\epsilon}_i\right)}{\sum_{j=1}^n\left\|\epsilon_i\right\|^{-1}}$;\par
(iii) $\mathbf{D} \leftarrow p \mathbf{D}^{1 / 2} \operatorname{diag}\left\{n^{-1} \sum_{j=1}^n U\left(\boldsymbol{\epsilon}_i\right) U\left(\boldsymbol{\epsilon}_i\right)^\top\right\} \mathbf{D}^{1 / 2}$.\par
The resulting estimators of location and diagonal matrix are denoted as $\hat{\boldsymbol{\theta}}$ and $\hat{\mathbf{D}}$. The sample mean and sample variances can be utilized as initial estimators. Regrettably, no evidence has been found to confirm the convergence of the aforementioned algorithm, even in low-dimensional scenarios, despite its consistent practical effectiveness. The existence or uniqueness of the HR estimator mentioned above also lacks proof. This topic certainly warrants further investigation.

In this section, we investigate some theoretical properties based on maximum-norm  about $\hat{\boldsymbol\theta}$. Similar to the proof of Lemma 1 and Theorem 1 in \cite{cheng2023}, we give the Bahadur representation of $\hat{\bf D}^{-1/2}(\hat{\boldsymbol\theta}-\boldsymbol\theta)$ and the Gaussian approximation result for $\hat{\bf D}^{-1/2}(\hat{\boldsymbol\theta}-\boldsymbol\theta)$ over hyperrectangles. Based on Gaussian approximation, we can easily derive the limiting distribution of $\hat{\boldsymbol\theta}$ based on the maximum-norm.

For $i=1,2,\cdots,n$, we denote
$\boldsymbol U_i=U(\mathbf D^{-1/2}(\boldsymbol X_i-\boldsymbol \theta))$ and $R_i=\Vert \mathbf D^{-1/2}(\boldsymbol X_i-\boldsymbol \theta)\Vert$ as the scale-invariant spatial-sign and radius of $\boldsymbol X_i-\boldsymbol \theta$, respectively. The moments of $R_i$ is defined as   $\zeta_k=\mathbb{E}\left(R_i^{-k}\right) $ for  $k=1,2,3,4$.  Denote $\boldsymbol W_i=\left(W_{i, 1}, \ldots, W_{i, p}\right)^{\top}$, the assumption is as follows.
\begin{assumption}\label{max1}
    $W_{i, 1}, \ldots, W_{i, p}$ are i.i.d. symmetric random variables with $\mathbb{E}\left(W_{i, j}\right)=0, \mathbb{E}\left(W_{i, j}^2\right)=$ 1 , and $\left\|W_{i, j}\right\|_{\psi_\alpha} \leqslant c_0$ with some constant $c_0>0$ and $1 \leqslant \alpha \leqslant 2$.
\end{assumption}
\begin{assumption}\label{max2}
    The moments $\zeta_k=\mathbb{E}\left(R_i^{-k}\right)$ for $k=1,2,3,4$ exist for large enough $p$. In addition, there exist two positive constants $\underline{b}$ and $\bar{B}$ such that $\underline{b} \leqslant \lim \sup _p \mathbb{E}\left(R_i / \sqrt{p}\right)^{-k} \leqslant \bar{B}$ for $k=1,2,3,4$.
\end{assumption}
\begin{assumption}\label{max3}
    The shape matrix $\R=\mathbf D^{-1/2}\mathbf\Gamma\mathbf\Gamma^\top \mathbf D^{-1/2}=\left(\sigma_{j \ell}\right)_{p \times p}$  satisfies $\max _{j=1,\cdots,p}\sum_{\ell=1}^p\left|\sigma_{j \ell}\right| \leqslant a_0(p)$ where $a_0(p)$ is a constant depending  only on  dimension $p$. In addition, $\lim\inf_{p\rightarrow \infty}\min_{j=1,2,\cdots,p}d_j>\underline{d}$ for some constant $\underline d>0$, where $\mathbf D=\operatorname{diag}\{d_1^2,d_2^2,\cdots,d_p^2\}$.
\end{assumption}

\begin{remark}{1}
Assumption \ref{max1} is the same as Condition C.1 in \cite{cheng2023}, which ensure that $\bm \theta$ in model (\ref{modelx}) is the population spatial median and $W_{i,j}$ has a sub-exponential distribution. If $\bm W_i\sim N(\bm 0, \I_p)$, $\bm X_i$ follows a elliptical symmetric distribution. Assumption \ref{max2} extend the Assumption 1 in \cite{zou2014multivariate}, which indicates that $\zeta_k\asymp p^{-k/2}$. {  It is a mild assumption and introduced to avoid $\boldsymbol X_i$ from concentrating too much near $\boldsymbol\theta$. For three commonly used distribution, multivariate normal, student-$t$ and mixtures of multivariate normal distributions, Assumption 2 are satisfied. It also ensures the existence of the moments $\mathbb E R_i^{-k}$, $k=1,2,3,4$. For example, for standard multivariate normal distribution, $\mathbb E R_i^{-4}$ equals to $1/(p-2)(p-4))$ which restricts the dimension $p>4$.  See also discussions in  \cite{zou2014multivariate,cheng2023} on similar assumptions.} Assumption \ref{max3} means the correlation between those variables could not be too large, which is similar to the matrix class in \cite{bickel2008covariance}.
\end{remark}

The following lemma shows a Bahadur representation of $\hat{\bm\theta}$,  which is the basis of Gaussian approximation result in Theorem \ref{thm1}.
\begin{lemma}\label{lemma1}
    (Bahadur representation) Assume Assumptions \ref{max1}-\ref{max3} with $a_0(p)\asymp p^{1-\delta}$ for some positive constant $\delta\leq 1/2$ hold. If $\log p=o(n^{1/3})$ and $\log n=o(p^{1/3\wedge \delta})$, then
    $$
    n^{1/2}\hat{\mathbf D}^{-1/2}(\hat{\boldsymbol \theta}-\boldsymbol \theta)=n^{-1/2}\zeta_1^{-1}\sum_{i=1}^n \boldsymbol U_i+C_n,
    $$
    where
$$
\begin{aligned}
\Vert C_n\Vert_\infty=O_p\{n^{-1/4}\log^{1/2}(np)+p^{-(1/6\wedge \delta/2)}\log^{1/2}(np)+n^{-1/2}(\log p)^{1/2}\log^{1/2}(np)\}.
\end{aligned}
$$
\end{lemma}
\begin{remark}{2}
\cite{feng2016multivariate} derived the Bahadur representation of the estimator $\hat{\bm \theta}$, where the remainder term $||C_n||$ is $o_p(\zeta_1^{-1})$, assuming a symmetric elliptical distribution. In this context, we provide the rate of the remainder term $C_n$ subject to a maximum-norm constraint. It's important to note that in this Lemma, we scale the spatial-median estimator $\hat{\bm \theta}$ by $\hat{\bf D}^{-1/2}$. This is a departure from much of the existing literature on the Bahadur representation of the spatial median, which does not exhibit scalar invariance. Such works include \cite{zou2014multivariate}, \cite{cheng2019testing},  \cite{li2022asymptotic}, and  \cite{cheng2023}.
\end{remark}

Let $\mathcal{A}^{\text {re }}=\left\{\prod_{j=1}^p\left[a_j, b_j\right]:-\infty \leqslant a_j \leqslant b_j \leqslant \infty, j=1, \ldots, p\right\}$ be the class of rectangles in $\mathbb R^p$. Based on the Bahadur representation of $\hat{\boldsymbol\theta}$, we acquire the following Gaussian approximation of $\hat{\mathbf D}^{-1/2}(\hat{\boldsymbol{\theta}}-\boldsymbol{\theta})$ in rectangle $\mathcal{A}^{\text {re }}$.
\begin{lemma}\label{lemma2}
    (Gaussian approximation) Assume Assumptions \ref{max1}-\ref{max3} with $a_0(p)\asymp p^{1-\delta}$ for some positive constant $\delta \leqslant 1 / 2$ hold. If $\log p=o\left(n^{1 / 5}\right)$ and $\log n=o\left(p^{1 / 3 \wedge \delta}\right)$, then
$$
\rho_n\left(\mathcal{A}^{\mathrm{re}}\right)=\sup _{A \in \mathcal{A}^{\mathrm{re}}}\left|\mathbb{P}\left\{n^{1 / 2}\hat{\mathbf D}^{-1/2}\left(\hat{\boldsymbol{\theta}}-\boldsymbol{\theta}\right) \in A\right\}-\mathbb{P}(\boldsymbol G \in A)\right| \rightarrow 0,
$$
as $n \rightarrow \infty$, where $\boldsymbol G \sim N\left(0, \zeta_1^{-2} \mathbf \Sigma_u\right)$ with $\mathbf \Sigma_u=\mathbb{E}\left(\boldsymbol U_1 \boldsymbol U_1^{\top}\right)$.
\end{lemma}
The Gaussian approximation for $\hat{\boldsymbol{\theta}}$ indicates that the probabilities $\mathbb{P}\left\{n^{1 / 2}\hat{\mathbf D}^{-1/2}\left(\hat{\boldsymbol{\theta}}-\boldsymbol{\theta}\right) \in A\right\}$ can be approximated by that of a centered Gaussian random vector with covariance matrix $\zeta_1^{-2} \mathbf \Sigma_u$ for hyperrectangles $A \in \mathcal{A}^{\text {re }}$.

Since the region $\mathcal A^t=\left\{\prod_{j=1}^p\left[a_j, b_j\right]:-\infty = a_j \leqslant b_j=t, j=1, \ldots, p\right\}$ used in the following corollary is contained in the set $\mathcal{A}^{\mathrm{re}}$, it is clear that the Corollary \ref{coro1} follows.
\begin{coro}\label{coro1}
Under the assumptions of Lemma \ref{lemma2}, as $n\rightarrow \infty$, we have
$$
\rho_n=\sup _{t \in \mathbb{R}}\left|\mathbb{P}\left(n^{1 / 2}\Vert\hat{\mathbf D}^{-1/2}(\hat{\boldsymbol{\theta}}-\boldsymbol{\theta})\Vert_{\infty} \leqslant t\right)-\mathbb{P}\left(\Vert\boldsymbol G\Vert_{\infty} \leqslant t\right)\right| \rightarrow 0 .
$$
where $\boldsymbol G \sim N\left(0, \zeta_1^{-2} \mathbf \Sigma_u\right)$ .
\end{coro}
Taking into account the relationships between $\mathbf \Sigma_u$ and $p^{-1}\R $, we propose a more straightforward Gaussian approximation.
\begin{lemma}\label{lemma3}
    (Variance approximation) Suppose $\boldsymbol G\sim N\left(0, \zeta_1^{-2} \mathbf \Sigma_u\right)$ and $\boldsymbol Z\sim N\left(0, \zeta_1^{-2}p^{-1}\mathbf R \right)$, under the assumptions of Lemma \ref{lemma2}, as $(n,p)\rightarrow \infty$, we have
    $$
    \sup _{t \in \mathbb{R}}\left|\mathbb{P}\left(\Vert \boldsymbol Z\Vert_{\infty}\leqslant t\right)-\mathbb{P}\left(\Vert \boldsymbol G\Vert_{\infty} \leqslant t\right)\right| \rightarrow 0 .
    $$
\end{lemma}
By integrating Corollary \ref{coro1} and Lemma \ref{lemma3}, we can readily derive the principal theorem of Gaussian approximation.
\begin{theorem}\label{thm1}
   Assume Assumptions \ref{max1}-\ref{max3} with $a_0(p)=p^{1-\delta}$ for some positive constant $\delta \leqslant 1 / 2$ hold. If $\log p=o\left(n^{1 / 5}\right)$ and $\log n=o\left(p^{1 / 3 \wedge \delta}\right)$, then
 $$
\tilde\rho_n=\sup _{t \in \mathbb{R}}\left|\mathbb{P}\left(n^{1 / 2}\Vert\hat{\mathbf D}^{-1/2}(\hat{\boldsymbol{\theta}}-\boldsymbol{\theta})\Vert_{\infty} \leqslant t\right)-\mathbb{P}\left(\Vert \boldsymbol Z\Vert_{\infty} \leqslant t\right)\right| \rightarrow 0 ,
$$
where $\boldsymbol Z \sim N\left(0, \zeta_1^{-2}p^{-1}\R \right)$.
\end{theorem}
\subsection{Max-type test procedure}
In order to guarantee that the maximum value of a sequence of normal variables adheres to a Gumbel limiting distribution, we introduce Assumption \ref{max4}. This assumption is employed to specify the necessary correlation among variables.
\begin{assumption}\label{max4}
    For some $\varrho \in(0,1)$, assume $|\sigma_{ij}|\leq \varrho$ for all $1\leq i<j\leq p$ and $p
    \geq 2$. Suppose $\left\{\delta_p ; p \geq 1\right\}$ and $\left\{\kappa_p ; p \geq 1\right\}$ are positive constants with $\delta_p=o(1 / \log p)$ and $\kappa=\kappa_p \rightarrow 0$ as $p \rightarrow \infty$. For $1 \leq i \leq p$, define $B_{p, i}=\left\{1 \leq j \leq p ;\left|\sigma_{i j}\right| \geq \delta_p\right\}$ and $C_p=\left\{1 \leq i \leq p ;\left|B_{p, i}\right| \geq p^\kappa\right\}$. We assume that $\left|C_p\right| / p \rightarrow 0$ as $p \rightarrow \infty$.
\end{assumption}
\begin{remark}{3}
Assumption \ref{max4} aligns with Assumption (2.2) in \cite{feng2022asymptotic}. This assumption stipulates that for each variable, the count of other variables that exhibit a strong correlation with it cannot be excessively large. To the best of our understanding, this is the least restrictive assumption in the literature that allows for the limiting null distribution of the maximum of correlated normal random variables to follow a Gumbel distribution.
Both Assumption \ref{max3} and \ref{max4} pertain to the correlation matrix $\R$. We examine two specific cases that satisfy both of these conditions. The first case is the classic AR(1) structure, denoted as $\R=(\rho^{|i-j|})_{1\le i,j\le p}, \rho\in (-1,1)$. In this scenario, $\sum_{l=1}^p |\sigma_{jl}|\to \frac{1}{1+\rho}$, which allows $\delta$ to be one in Assumption \ref{max3}. For Assumption \ref{max4}, we set $\delta_p=(\log p)^{-2}$, leading to $B_{p,i}=\{j:|i-j|\le -2\log \log p/\log |\rho|\}$. As a result, $|B_{p,i}|\le -4\log \log p/\log |\rho|<p^\kappa$ with $\kappa=5\log\log p/\log p$, which implies $|C_p|=0$ and Assumption \ref{max4} is satisfied. The second case involves a banded correlation matrix, where $\sigma_{ij}=0$ if $|i-j|> \ell$. Here, $\sum_{j=1}^p|\sigma_{ij}|=O(\ell)$ and $|B_{p,i}|\le \ell$ for $\delta_p=(\log p)^{-2}$. Therefore, Assumptions \ref{max3} and \ref{max4} will hold if $\ell =o(p^\kappa)$ for any $\kappa\to 0$.
\end{remark}

 Suppose Assumption \ref{max1}-\ref{max4} hold, by the Theorem 2 in \cite{feng2022asymptotic}, we can see that $p\zeta_1^2\max _{1 \leq i \leq p} Z_i^2-2 \log p+\log \log p$ converges to a Gumbel distribution with cdf $F(x)=\exp \left\{-\frac{1}{\sqrt{\pi}} e^{-x / 2}\right\}$ as $p \rightarrow \infty$. In combining with the Theorem \ref{thm1} we can conclude that,
\begin{equation}\label{eq:Tmax0}
    \mathbb P\left(n^{1/2}\left\|\hat{\mathbf D}^{-1 / 2}(\hat{\boldsymbol{\theta}}-\boldsymbol{\theta})\right\|_{\infty}^2 p \zeta_1^2-2 \log p+\log \log p \leq x\right) \rightarrow \exp \left\{-\frac{1}{\sqrt{\pi}} e^{-x / 2}\right\}.
\end{equation}
Next we replace $E\left(R^{-1}\right)$ with its estimators. We denote $\hat{R}_i=\Vert\hat{\mathbf{D}}^{-1 / 2}(\mathbf{X}_i-\hat{\boldsymbol{\theta}})\Vert$. Then the estimator is defined as $\hat \zeta_{1}:=\frac{1}{n} \sum_{i=1}^n \hat{R}_i^{-1}$, and the proof of consistency is shown in Lemma \ref{lemma:consistency_zeta}. Because the convergence rate of maximum is very slow, we propose a adjust max-type test statistic which based on the scalar-invariant spatial median $\hat{\boldsymbol \theta}$,
$$
T_{M A X}=n\left\|\hat{\mathbf{D}}^{-1 / 2} \hat{\boldsymbol{\theta}}\right\|_{\infty}^2  \hat{\zeta}_{1}^2p\cdot \left(1-n^{-1 / 2}\right).
$$
\begin{theorem}\label{thm:max_dist}
Suppose the assumptions in Theorem \ref{thm1} and Assumption 4 hold. Under the null hypothesis, as $(n,p)\to \infty$, we have
$$
P\left(T_{M A X}-2 \log p+\log \log p \leq x\right) \rightarrow \exp \left\{-\frac{1}{\sqrt{\pi}} e^{-x / 2}\right\}.
$$
\end{theorem}

According to Theorem \ref{thm:max_dist}, a level- $\alpha$ test will then be performed through rejecting $H_0$ when $T_{MAX}-2 \log p+\log \log p$ is larger than the $(1-\alpha)$ quantile $q_{1-\alpha}=-\log \pi-2 \log \log (1-\alpha)^{-1}$ of the Gumbel distribution $F(x)$.

The following theorem demonstrates the consistency of the proposed max-type test.
\begin{theorem}\label{thm:consistency}
Suppose the Assumptions in Theorem \ref{thm:max_dist} hold. For any given $\alpha\in (0,1)$, if $\Vert \boldsymbol \theta\Vert_\infty\geq \tilde C n^{-1/2} \{\log p-2\log\log (1-\alpha)^{-1}\}^{1/2}$ for some large enough constant $\tilde C$, we have
$$
\mathbb P(T_{MAX}-2\log p+\log \log p>q_{1-\alpha}\mid H_1)\rightarrow 1,
$$
as $n\rightarrow \infty$.
\end{theorem}

Given a fixed significant level $\alpha$, the test $T_{MAX}$ attains consistency if $\Vert \boldsymbol \theta\Vert_\infty\geq \tilde C\sqrt{\log p/n}$, provided that $\tilde{C}$ is sufficiently large. This is the minimax rate optimal for testing against sparse alternatives, as stated in Theorem 3 of \cite{CLX14}. If $\tilde{C}$ is adequately small, then it becomes impossible for any $\alpha$-level test to reject the null hypothesis with a probability approaching one. Therefore, Theorem \ref{thm:consistency} also demonstrates the optimality of our proposed test $T_{MAX}$.

To show the high dimensional asymptotic relative efficiency, we consider a special alternative hypothesis: $$H_1: \bm \theta=(\theta_1,0,\cdots,0)^\top, \theta_1>0,$$ which means there are only one variable has nonzero mean. Let $x_{\alpha}=2\log p-\log \log p+q_{1-\alpha}$. In this case,
\begin{align*}
\mathbb P\left(\hat{d}_1^{-2}\hat{\theta}^2_1n\hat{\zeta}_1^2p\ge x_\alpha\right)\le \mathbb P\left(T_{MAX}\ge x_{\alpha}\right)\le \mathbb P\left(\hat{d}_1^{-2}\hat{\theta}^2_1n\hat{\zeta}_1^2p\ge x_\alpha\right)+\mathbb P\left(\max_{2\le i\le p}\hat{d}_i^{-2}\hat{\theta}^2_in\hat{\zeta}_1^2p\ge x_\alpha\right).
\end{align*}
Under this special alternative hypothesis, we can easily have $$\mathbb P\left(\max_{2\le i\le p}\hat{d}_i^{-2}\hat{\theta}^2_in\hat{\zeta}_1^2p\ge x_\alpha\right)\to \alpha,$$ and
\begin{align*}
\mathbb P\left(\hat{d}_1^{-2}\hat{\theta}^2_1n\hat{\zeta}_1^2p\ge x_\alpha\right)\to \Phi\left(-\sqrt{x_{\alpha}}+(np)^{1/2}d_1^{-1}\theta_1\zeta_1\right).
\end{align*}
So the power function of our proposed $T_{MAX}$ test is
\begin{align*}
\beta_{MAX}(\bm \theta)\in (\Phi\left(-\sqrt{x_{\alpha}}+(np)^{1/2}d_1^{-1}\theta_1\zeta_1\right),\Phi\left(-\sqrt{x_{\alpha}}+(np)^{1/2}d_1^{-1}\theta_1\zeta_1\right)+\alpha).
\end{align*}
Similarly, the power function of \cite{CLX14}'s test is
\begin{align*}
\beta_{CLX}(\bm \theta)\in (\Phi\left(-\sqrt{x_{\alpha}}+n^{1/2}\varsigma_1^{-1}\theta_1\right),\Phi\left(-\sqrt{x_{\alpha}}+n^{1/2}\varsigma_1^{-1}\theta_1\right)+\alpha),
\end{align*}
where $\varsigma^2_i$ is the variance of $X_{ki}, i=1,\cdots,p$. Thus, the asymptotic relative efficiency of $T_{MAX}$ with respective to \cite{CLX14}'s test could be approximated as
\begin{align*}
ARE(T_{MAX},T_{CLX})=\{E(R_i^{-1})\}^2E(R_i^2),
\end{align*}
which is the same as the asymptotic relative efficiency of \cite{feng2016}'s test with respective to \cite{srivastava2009test}'s test. If $\boldsymbol{X}_i$ are generated from standard multivariate $t$-distribution with $\nu$ degrees of freedom $(\nu>2)$,
$$
ARE(T_{MAX},T_{CLX})=\frac{2}{\nu-2}\left(\frac{\Gamma((\nu+1) / 2)}{\Gamma(\nu / 2)}\right)^2 .
$$
For different $\nu=3,4,5,6$, the above ARE are $2.54,1.76,1.51,1.38$, respectively. Under the multivariate normal distribution $(\nu=\infty)$, our $T_{MAX}$ test is the same powerful as \cite{CLX14}'s test. However, our $T_{MAX}$ test is much more powerful than \cite{CLX14}'s test under the heavy-tailed distributions.

\section{Maxsum test}
\subsection{Exisiting Sum-type test}\label{sumtypesec}
\cite{feng2016multivariate} proposed the following sum-type test statistic:
$$T_{SUM}=\frac{2}{n(n-1)} {\sum \sum}_{i<j} U\left(\hat{\mathbf{D}}_{i j}^{-1 / 2} \boldsymbol{X}_i\right)^T U\left(\hat{\mathbf{D}}_{i j}^{-1 / 2} \boldsymbol{X}_j\right),$$
where $\hat{\mathbf{D}}_{i j}$ are the corresponding diagonal matrix estimator using leave-two-out sample $\left\{{\bm X}_k\right\}_{k \neq i, j}^n$.

By \cite{feng2016multivariate}, we have the following theorem and assumptions:\par
\begin{assumption}\label{sum0}
     Variables $\left\{{\bm X}_1, \ldots, {\bm X}_n\right\}$ in the $n$-th row are independently and identically distributed (i.i.d.) from $p$-variate elliptical distribution with density functions   $\operatorname{det}(\boldsymbol{\Sigma})^{-1 / 2}\cdot$ $ g\left(\left\|\boldsymbol{\Sigma}^{-1 / 2}(\mathbf{x}-\boldsymbol{\theta})\right\|\right)$ where $\boldsymbol{\theta}$ 's are the symmetry centers and $\boldsymbol{\Sigma}$ 's are the positive definite symmetric $p \times p$ scatter matrices.
\end{assumption}

\begin{assumption}\label{sum1}
$\operatorname{tr}\left(\R^4\right)=o\left(\operatorname{tr}^2\left(\R^2\right)\right)$.
\end{assumption}
\begin{assumption}\label{sum2}
    (i) $\operatorname{tr}\left(\R^2\right)-p=o\left(n^{-1} p^2\right)$, (ii) $n^{-2} p^2 / \operatorname{tr}\left(\R^2\right)=O(1)$ and $\log p=o(n)$.
\end{assumption}
\begin{remark}{4}
    { Assumption \ref{sum1} is a common condition for sum-type test statistic in high dimensions, see  \cite{chen2010two,feng2016multivariate,feng2016}, which requires that the eigenvalues of $\R$ not diverge excessively. If all the eigenvalues of $\R$ are bounded, $\tr(\R^2)=O(p),\tr(\R^4)=O(p)$. So the Assumption \ref{sum1} holds trivially. In this case, Assumption \ref{sum2} becomes $p=O(n^2)$ and $p/n\rightarrow \infty$. Actually, Assumption \ref{sum2} indicates $p=O(n^2)$ and it is not necessary for the eigenvalues to be bounded. For $b$ unbounded eigenvalues with respect dimension $p$, the sufficient condition for Assumption \ref{sum1} is $\lambda_{(p)}/\lambda_{(1)}=o((p-b)^{1/2}b^{-1/4})$. For Assumption \ref{sum2}, a sufficient condition is $\lambda_i=1+o(p^{1/2}n^{-1/2})$. When the assumptions about the spectrum of $\R$ do not hold, there would be a bias therm in sum-type statistic that difficult to calculate and deserves to be investigated further , see the Supplemental Material of \cite{feng2016multivariate} for more details.}
\end{remark}

The following Lemma restate the Theorem 1 in \cite{feng2016}, which gives the asymptotic null distribution of $T_{SUM}$ under the symmetric elliptical distribution assumption.
\begin{lemma}\label{thm2}
Under Assumptions \ref{sum0}-\ref{sum2}. and $H_0$, as $\min \left(p, n\right) \rightarrow \infty, T_{SUM} / \sigma_n \stackrel{d}{\rightarrow} N(0,1)$, where $\sigma_n^2=\frac{2}{n(n-1) p^2} \operatorname{tr}\left({\R}^2\right)$.
\end{lemma}
To broaden the application, we re-derive the limiting null distribution of $T_{SUM}$ under a more generalized model (\ref{modelx}).
\begin{theorem}\label{lemma_sum}
Under Assumptions \ref{max1}-\ref{max3},\ref{sum1}-\ref{sum2} and $H_0$, as $\min \left(p, n\right) \rightarrow \infty, T_{SUM} / \sigma_n \stackrel{d}{\rightarrow} N(0,1)$.
\end{theorem}
Similar to \cite{feng2016}, we propose the following estimator to estimate the trace term in $\sigma_n^2$
$$
\widehat{\operatorname{tr}\left({\R}^2\right)}=\frac{p^2}{n(n-1)} \sum_{i=1}^n \sum_{j \neq i}^n\left(U\left(\hat{\mathbf{D}}_{i j}^{-1 / 2}\left(\boldsymbol{X}_i-\hat{\boldsymbol{\theta}}_{i j}\right)\right)^T U\left(\hat{\mathbf{D}}_{i j}^{-1 / 2}\left(\boldsymbol{X}_j-\hat{\boldsymbol{\theta}}_{i j}\right)\right)\right)^2,
$$
where $\left(\hat{\boldsymbol{\theta}}_{i j}, \hat{\mathbf{D}}_{i j}\right)$ are the corresponding spatial median and diagonal matrix estimators using leave-two-out sample $\left\{\boldsymbol{X}_k\right\}_{k \neq i, j}^n$. Similar to the proof of Proposition 2 in \cite{feng2016multivariate}, we can easily obtain that $\widehat{\operatorname{tr}\left(\mathbf{R}^2\right)} / \operatorname{tr}\left(\mathbf{R}^2\right) \xrightarrow{p} 1$ as $p, n \rightarrow \infty$ under model (\ref{modelx}). Consequently, a ratio-consistent estimator of $\sigma_n^2$ under $H_0$ is $\hat{\sigma}_n^2=\frac{2}{n(n-1) p^2} \widehat{\operatorname{tr}\left(\mathbf{R}^2\right)}$. And then we reject the null hypothesis with $\alpha$ level of significance if $T_{SUM}/ \hat{\sigma}_n>z_\alpha$, where $z_\alpha$ is the upper $\alpha$ quantile of $N(0,1)$.

And we also re-derive the asymptotic distribution of $T_{SUM}$ under the following alternative hypothesis:
\begin{equation}\label{H1_dense}
   H_1: \boldsymbol{\theta}^\top \mathbf{D}^{-1} \boldsymbol{\theta}=O\left(\zeta_1^{-2} \sigma\right)\text{ and } \boldsymbol{\theta}^\top \R \boldsymbol{\theta}=o\left(\zeta_1^{-2} n p \sigma^2\right)
\end{equation}

\begin{theorem}\label{lemma5}
    Under Assumptions \ref{max1}-\ref{max3},\ref{sum1}-\ref{sum2} and the alternative hypothesis (\ref{H1_dense}), as $\min \left(p, n\right) \rightarrow \infty$,$$ \frac{T_{SUM}-\zeta_1^2 \boldsymbol{\theta}^\top \mathbf{D}^{-1} \boldsymbol{\theta}}{\sigma_n}  \stackrel{d}{\rightarrow} N(0,1).$$
\end{theorem}
By Theorem \ref{lemma_sum} and \ref{lemma5}, the power function of $T_{SUM}$ can be approximated as
\begin{equation*}
\beta_{SUM}\left(\boldsymbol{\theta}\right)=\Phi\left(-z_\alpha+\frac{\zeta_1^2 n p \boldsymbol{\theta}^T \mathbf{D}^{-1} \boldsymbol{\theta}}{\sqrt{2 \operatorname{tr}\left(\mathbf{R}^2\right)}}\right).
\end{equation*}
Hence, $T_{SUM}$ is expected to perform well under the dense alternative hypothesis. For a more detailed discussion on the asymptotic relative efficiency of $T_{SUM}$ compared to other tests, refer to \cite{feng2016}. The power comparison between $T_{SUM}$ and $T_{MAX}$ will be addressed in the following subsection.

\subsection{Maxsum test}
In this subsection, we demonstrate that our proposed test statistic $T_{MAX}$ is asymptotically independent of the statistic $T_{SUM}$ presented in  \cite{feng2016}. This allows us to carry out a Cauchy $p$-value combination of the two asymptotically independent $p$-values, resulting in a new test. This test is tailored to accommodate both sparse and dense alternatives.

\begin{assumption}\label{ind1}
    There exist $C>0$ and $\varrho \in(0,1)$ so that $\max _{1 \leq i<j \leq p}\left|\sigma_{i j}\right| \leq \varrho$ and
 $\max _{1 \leq i \leq p} \sum_{j=1}^p \sigma_{i j}^2 \leq(\log p)^C$ for all $p \geq 3 ; p^{-1 / 2}(\log p)^C \ll \lambda_{\min }({\R}) \leq \lambda_{\max }(\R) \ll \sqrt{p}(\log p)^{-1}$ and $\lambda_{\max }(\R) / \lambda_{\min }(\R)=O\left(p^\tau\right)$ for some $\tau \in(0,1 / 4)$.
\end{assumption}
\begin{remark}{4}
Assumption \ref{ind1} is the same as the condition (2.3) in \cite{feng2022a}. As shown in \cite{feng2022a}, Assumption \ref{ind1} is more restrictive than Assumption \ref{max3}, \ref{max4} and \ref{sum1}. Under Assumption \ref{ind1}, we have $p^{1/2}(\log p)^{C}\lesssim \tr(\R^2)\lesssim p^{3/2}\log^{-1} p$. So Assumption \ref{sum2} will hold if $n=o(p^{3/2}
\log p)$ and $p^{3/4}(\log p)^{-C/2}=O(n)$. Intuitively speaking, if all the eigenvalues of $\R$ are bounded and $p/n\to c\in (0,\infty)$, all the assumptions \ref{max3}, \ref{max4}, \ref{sum1}, \ref{sum2} and \ref{ind1} hold.
\end{remark}

\begin{theorem}\label{thm3}
    Under Assumptions \ref{max1}-\ref{max4}, \ref{sum2}-\ref{ind1} and $H_0$, $T_{SUM}$ and $T_{MAX}$ are asymptotically independent, i.e.
\begin{align*}
\mathbb P\left(T_{SUM}/\sigma_n\le x, T_{MAX}-2\log p+\log\log p\le y\right)\to \Phi(x)F(y),
\end{align*}
as $n,p\to \infty$.
\end{theorem}

According to Theorem \ref{thm3}, we suggest combining the corresponding $p$-values by using Cauchy Combination Method \citep{liu2020}, to wit,
\begin{align*}
    p_{CC}&=1-G[0.5\tan\{(0.5-p_{MAX})\pi\}+0.5\tan\{(0.5-p_{SUM})\pi\}],\\
    p_{MAX}&=1-F(T_{MAX}-2\log p+\log \log p),\\
    p_{SUM}&=1-\Phi(T_{SUM}/\hat{\sigma}_n),
\end{align*}
where $G(\cdot)$ is the CDF of the the standard Cauchy distribution. If the final $p$-value is less than some pre-specified significant level $\alpha\in(0,1)$, then we reject $H_0$.

Next, we consider the relationship between $T_{SUM}$ and $T_{MAX}$ under local alternative hypotheses:
\begin{equation}\label{H_1_comb}
H_1: \Vert\boldsymbol \theta\Vert=O(\zeta_1^{-2}\sigma),\Vert \R^{1/2}\boldsymbol\theta\Vert=o\left(\zeta_1^{-2} n p \sigma^2\right)\text{ and }\vert \mathcal A\vert=o\left(\frac{\lambda_{\min}(\R)[\text{tr}(\R^2)]^{1/2}}{(\log p)^C}\right),
\end{equation}
where $\mathcal A=\{i\mid \theta_i\not=0,1\leq i\leq p\}$, $\boldsymbol\theta=(\theta_1,\theta_2,\cdots,\theta_p)^\top$.
The following theorem establish the asymptotic independence between $T_{SUM}$ and $T_{MAX}$ under this special alternative hypothesis.
\begin{theorem}\label{thm4}
    Under Assumptions \ref{max1}-\ref{max4}, \ref{sum1}-\ref{ind1} and the alternative hypothesis (\ref{H_1_comb}), $T_{SUM}$ and $T_{MAX}$ are asymptotically independent, i.e.
\begin{align*}
&\mathbb P\left(T_{SUM}/\sigma_n\le x, T_{MAX}-2\log p+\log\log p\le y\right)\to \\
&\mathbb P\left(T_{SUM}/\sigma_n\le x\right)\mathbb P\left(T_{MAX}-2\log p+\log\log p\le y\right),
\end{align*}
as $n,p\to \infty$.
\end{theorem}

According to \cite{li2023}, the Cauchy combination-based test has more power than the test based on the minimum of $p_{MAX}$ and $p_{SUM}$, which is also known as the minimal p-value combination. This is represented as $\beta_{M\wedge S, \alpha}=P(\min\{{\rm p}_{MAX},{\rm p}_{SUM}\}\leq 1-\sqrt{1-\alpha})$.

It is clear that:
\begin{align}\label{power_H1}
\beta_{M\wedge S, \alpha} &\ge P(\min\{{\rm p}_{MAX},{\rm p}_{SUM}\}\leq \alpha/2)\nonumber\\
&= \beta_{MAX,\alpha/2}+\beta_{SUM,\alpha/2}-P({\rm p}_{MAX}\leq \alpha/2, {\rm p}_{SUM}\leq \alpha/2)\nonumber\\
&\ge \max\{\beta_{MAX,\alpha/2},\beta_{SUM,\alpha/2}\}.
\end{align}

On the other hand, under the local alternative hypothesis (\ref{H_1_comb}), we have:
\begin{align}\label{power_H1np}
\beta_{M\wedge S, \alpha} \ge \beta_{MAX,\alpha/2}+\beta_{SUM,\alpha/2}-\beta_{MAX,\alpha/2}\beta_{SUM,\alpha/2}+o(1),
\end{align}
which is due to the asymptotic independence implied by Theorem \ref{thm4}.

For a small $\alpha$, the difference between $\beta_{MAX,\alpha}$ and $\beta_{MAX,\alpha/2}$ is small, and the same applies to $\beta_{SUM,\alpha}$. Therefore, according to equations \eqref{power_H1} and \eqref{power_H1np}, the power of the adaptive test is at least as large as, or even significantly larger than, that of either the max-type or sum-type test. For a detailed comparison of the performance of each test type under varying conditions of sparsity and signal strength, please refer to Table 1 in \cite{ma2024testing}.

\section{Simulation}
In this section, we incorporated various methods into our study:
\begin{itemize}
\item the proposed test $T_{MAX}$, referred as SS-MAX;
\item sum-type test proposed by \cite{feng2016}, referred as SS-SUM;
\item the proposed test $T_{CC}$, referred as SS-CC;
\item  max-type method proposed by \cite{CLX14}
, referred as
MAX;
\item sum-type method proposed by \cite{srivastava2009test}, referred as
SUM;
\item combination test proposed by \cite{feng2022asymptotic}, referred as
COM.
\end{itemize}

The following scenarios are firstly considered.
\begin{itemize}
\item[(I)] Multivariate normal distribution. $\X_i\sim N(\bm\theta,\bms)$.
\item[(II)] Multivariate $t$-distribution $t_{p,4}$.   $\X_{i}$'s are generated from $t_{p,4}$ with location parameter $\bth$ and scatter matrix $\bms$.
\item[(III)] Multivariate mixture normal distribution $\mbox{MN}_{p_n,\gamma,9}$. $\X_{i}$'s are generated from  $\gamma
f_{p_n}(\bm\theta,\bms)+(1-\gamma)f_{p_n}(\bm\theta,9\bms)$, denoted
by $\mbox{MN}_{p_n,\gamma,9}$, where $f_{p_n}(\cdot;\cdot)$ is the
density function of $p_n$-variate multivariate normal distribution.
$\gamma$ is chosen to be 0.8.
\end{itemize}
Here we consider the scatter matrix $\bms=(0.5^{|i-j|})_{1\le i,j\le p}$.   Two sample
sizes $n = 50, 100$ and three dimensions $p = 200, 400, 600$ are considered. All the findings in this section are derived from 1000 repetitions. Table \ref{tab:size} presents the empirical sizes of the six tests mentioned above. It was observed that the spatial-sign based tests–SS-MAX, SS-SUM, and SS-CC–are able to effectively manage the empirical sizes in a majority of scenarios. {     Under the normality assumption, the Type I error of the MAX method increases as the ratio \( p/n \) grows. This may be due to the component \( n \bar{x}_i/\hat{\sigma}_{ii} \) of the max statistic following a \( t(n) \) distribution, which deviates from the normal distribution. In contrast, the SS-MAX method is more suitable for heavy-tailed data. When \( p \) is fixed and \( n \) increases, leading to a closer approximation to the normal distribution, the MAX method exhibits improved control over the Type I error.
} However, when dealing with distributions that are not multivariate normal, the SUM test tends to have empirical sizes that fall below the nominal level. Similarly, the COM test also exhibits smaller sizes when operating under non-normal distributions.

To compare the power performance of each test, we consider $\bm \theta=(\kappa,\kappa,\kappa,0,\cdots,0)$ where the first $s$ components of $\bm\theta$ are all equal to $\kappa=\sqrt{0.5/s}$. Figure \ref{power} illustrates the power curves for each test. In the case of the multivariate normal distribution, SS-SUM and SUM exhibit similar performance, aligning with the findings of \cite{feng2016}. The spatial-sign based max-type test procedure, SS-MAX, is slightly less powerful than its mean-based counterpart, MAX. The two combination type test procedures demonstrate comparable performance in this scenario. However, when dealing with non-normal distributions, the spatial-sign based test procedures surpass the mean-based ones. Moreover, the newly proposed test, SS-CC, outperforms the others in most scenarios. In extremely sparse scenarios ($s<5$), SS-CC’s performance is akin to SS-MAX. In highly dense scenarios ($s>10$), SS-CC performs similarly to SS-SUM. However, when the signal is neither very sparse nor very dense, SS-CC proves to be the most effective among all test procedures. This underscores the superiority of our proposed max-sum procedures, not only in handling signal sparsity but also in dealing with heavy-tailed distributions.

        \begin{table}[!ht]
        \begin{center}
    	\caption{  {Empirical size comparison of various tests with a nominal level 5\%. \label{tab:size}} }
                    \vspace{0.5cm}
                     \renewcommand{\arraystretch}{1.2}
                     \setlength{\tabcolsep}{10pt}
    {
    		\begin{tabular}{l|ccc|ccc}
    			\hline\hline
    			& \multicolumn{3}{c}{$n=50$}&\multicolumn{3}{c}{$n=100$} \\ \hline
    $p$&200&400&600&200&400&600\\ \hline
    &\multicolumn{6}{c}{Multivariate Normal Distribution}\\ \hline
SS-MAX&0.051&0.061&0.049&0.025&0.04&0.032\\
SS-SUM&0.061&0.056&0.041&0.06&0.059&0.068\\
SS-CC&0.071&0.065&0.048&0.057&0.056&0.056\\
MAX&0.095&0.125&0.116&0.052&0.081&0.072\\
SUM&0.076&0.086&0.054&0.069&0.064&0.081\\
COM&0.095&0.108&0.089&0.063&0.072&0.058\\ \hline
    &\multicolumn{6}{c}{Multivariate $t_3$ Distribution}\\ \hline
SS-MAX&0.063&0.062&0.063&0.061&0.063&0.058\\
SS-SUM&0.067&0.053&0.064&0.062&0.064&0.053\\
SS-CC&0.058&0.061&0.068&0.058&0.052&0.059\\
MAX&0.044&0.052&0.047&0.033&0.036&0.04\\
SUM&0.005&0.001&0.001&0.002&0.001&0\\
COM&0.021&0.03&0.027&0.019&0.014&0.019\\ \hline
    &\multicolumn{6}{c}{Multivariate Mixture Normal Distribution}\\ \hline
SS-MAX&0.056&0.061&0.07&0.037&0.037&0.044\\
SS-SUM&0.066&0.05&0.051&0.054&0.058&0.061\\
SS-CC&0.067&0.058&0.064&0.052&0.046&0.059\\
MAX&0.037&0.042&0.056&0.031&0.038&0.03\\
SUM&0.004&0&0&0.007&0.002&0\\
COM&0.021&0.019&0.028&0.013&0.02&0.01\\ \hline \hline
    	\end{tabular}}
    \end{center}
    \end{table}

\begin{figure}[!ht]
\includegraphics[width=1\textwidth]{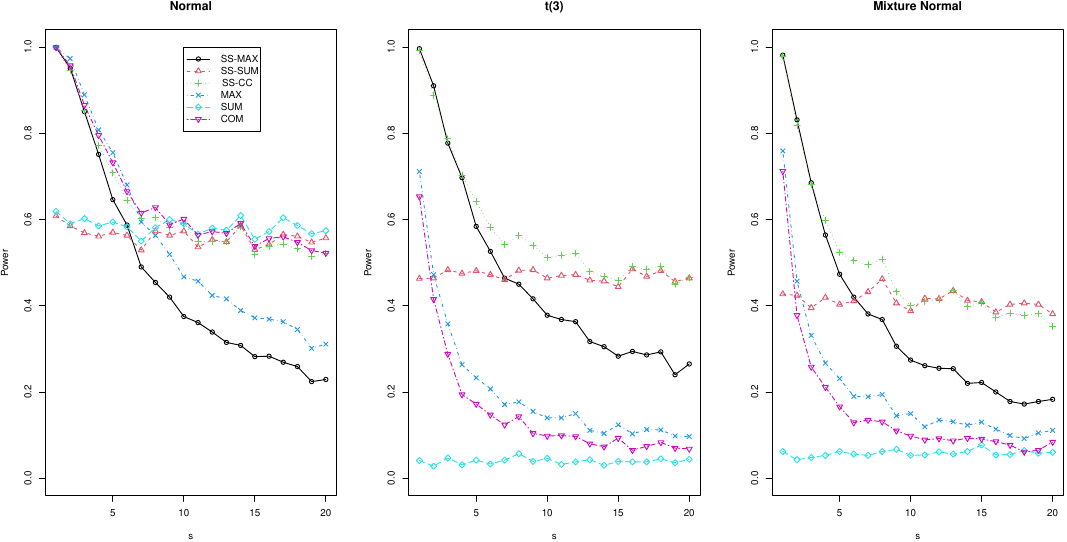}
\caption{Power of tests with different sparsity levels over $(n,p)=(100,200)$. \label{power}}
\end{figure}

Next, we consider the power comparison of those tests under different signal strength. Here we consider three sparsity level $s=2,20,50$ and the signal parameter $\kappa=\sqrt{\delta/s}$. Figures \ref{power1} through \ref{power3} present the power curves for various testing methods under Scenarios I to III. As the signal strength increases, the power of all tests also increases. Despite the presence of heavy-tailed distributions, spatial-sign based testing methods continue to surpass those based on means. Among all the tests, the proposed CC test consistently delivers the best performance.

\begin{figure}[!ht]
\includegraphics[width=1\textwidth]{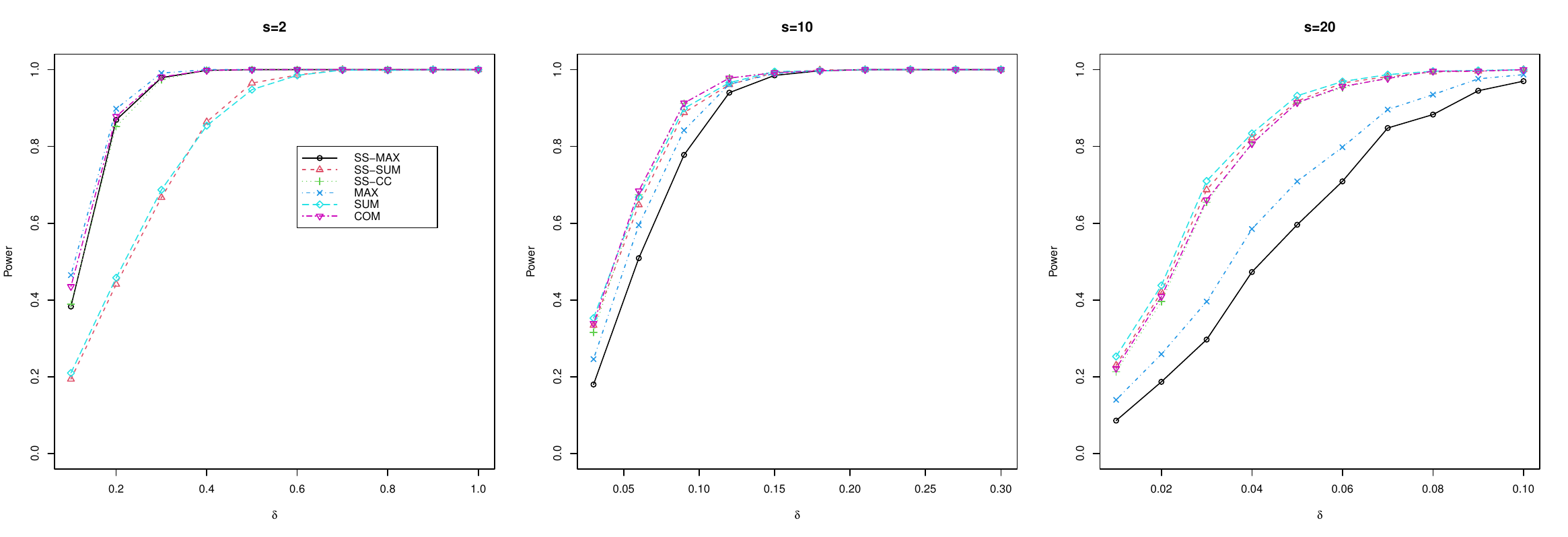}
\caption{Power of tests with different signal strength for multivariate normal distribution over $(n,p)=(100,200)$. \label{power1}}
\end{figure}

\begin{figure}[!ht]
\includegraphics[width=1\textwidth]{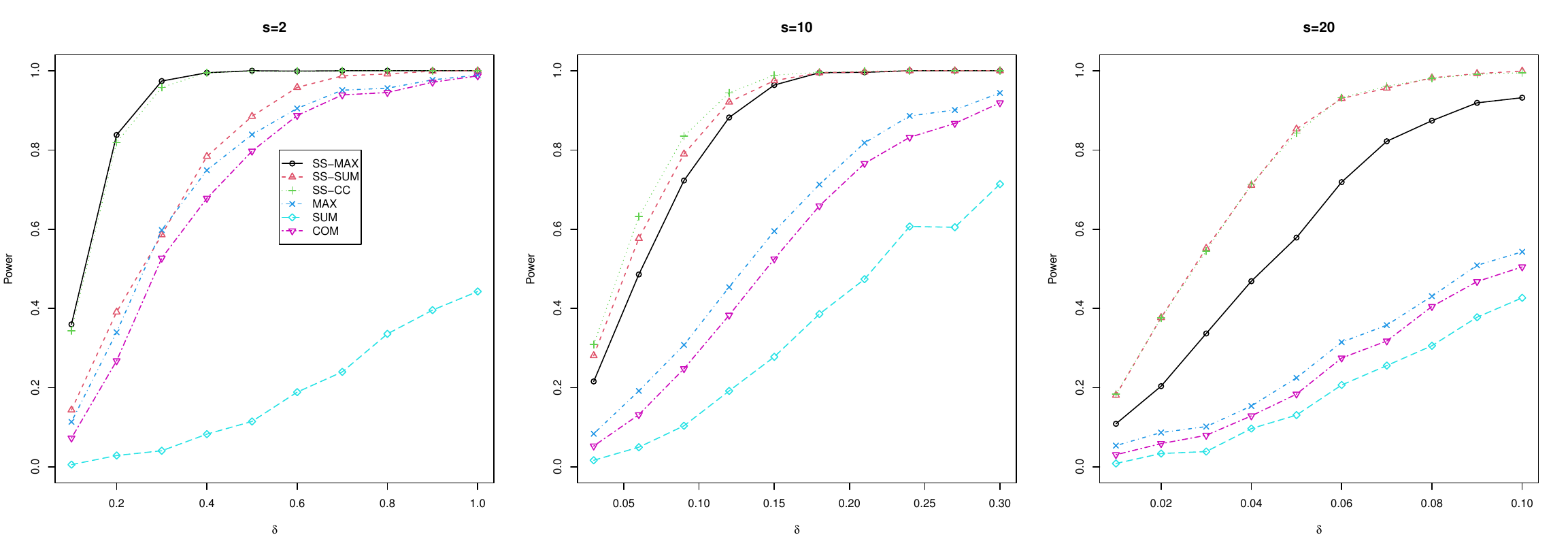}
\caption{Power of tests with different signal strength for multivariate $t_3$ distribution over $(n,p)=(100,200)$. \label{power2}}
\end{figure}

\begin{figure}[!ht]
\includegraphics[width=1\textwidth]{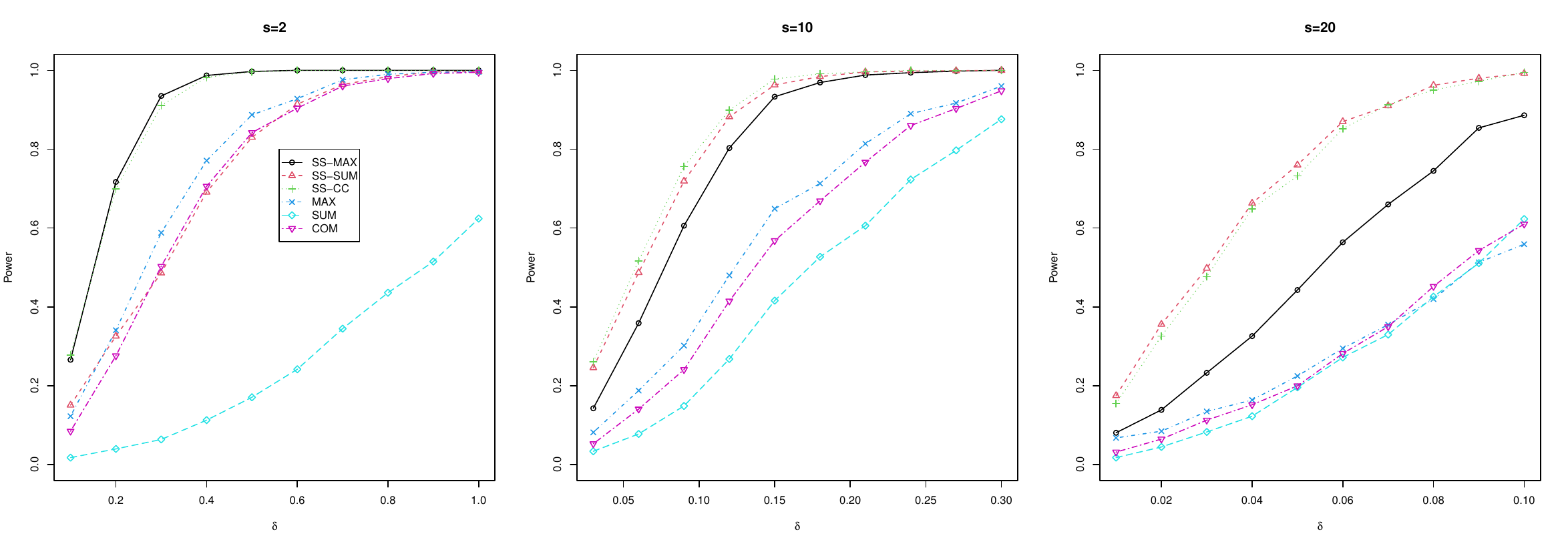}
\caption{Power of tests with different signal strength for multivariate mixture normal distribution over $(n,p)=(100,200)$. \label{power3}}
\end{figure}

As shown in \cite{feng2016}, for the sum-type test procedure, SS-SUM is more powerful than the non scalar-invariant test \citep{wang2015high}. Here we also compare our proposed test $T_{MAX}$ with \cite{cheng2023}'s test (abbreviated as CPZ hereafter) to show the importance of scalar-invariant for max-type test procedure. We consider two scatter matrix case for $\bm \Sigma$: (i) $\bm \Sigma=(0.5^{|i-j|})_{1\le i,j\le p}$; (ii) $\bm \Sigma=\D^{1/2}\R \D^{1/2}, \R=(0.5^{|i-j|})_{1\le i,j\le p}, \D=\diag\{d_1,\cdots,d_p\}$ where $d_i=1, i\le p/2, d_i=3, i >p/2$. The other settings are all the same as above. Table \ref{tab:size2} presents the empirical sizes of the SS-MAX and CPZ tests. Both tests are capable of controlling the empirical sizes in the majority of cases. Moreover, we conduct a power comparison of these two max-type tests under identical settings as previously mentioned, but with two distinct scatter matrix cases. Figures \ref{power4} and \ref{power5} depict the power curves of SS-MAX and CPZ under scatter matrix cases (i) and (ii), respectively. We observe that SS-MAX performs comparably to CPZ when all elements of the diagonal matrix of $\bm \Sigma$ are equal. However, SS-MAX exhibits greater power than CPZ when the elements of the diagonal matrix of the scatter matrix are unequal, underscoring the necessity of scalar-invariance.

        \begin{table}[!ht]
        \begin{center}
    	\caption{  {Empirical size comparison of SS-MAX and CPZ tests with a nominal level 5\%. \label{tab:size2}} }
                    \vspace{0.5cm}
                     \renewcommand{\arraystretch}{1.2}
                     \setlength{\tabcolsep}{10pt}
    {
    		\begin{tabular}{l|ccc|ccc}
    			\hline\hline
    			& \multicolumn{3}{c}{$n=50$}&\multicolumn{3}{c}{$n=100$} \\ \hline
    $p$&200&400&600&200&400&600\\ \hline
&\multicolumn{6}{c}{Scatter Matrix Case (i)}\\ \hline
&\multicolumn{6}{c}{Multivariate Normal Distribution}\\ \hline
SS-MAX&0.046&0.04&0.054&0.03&0.036&0.058\\
CPZ&0.076&0.054&0.076&0.056&0.08&0.088\\ \hline
&\multicolumn{6}{c}{Multivariate $t_3$ Distribution}\\ \hline
SS-MAX&0.06&0.09&0.096&0.06&0.054&0.054\\
CPZ&0.06&0.064&0.06&0.052&0.07&0.038\\ \hline
&\multicolumn{6}{c}{Multivariate Mixture Normal Distribution}\\ \hline
SS-MAX&0.072&0.06&0.072&0.038&0.052&0.038\\
CPZ&0.086&0.08&0.074&0.078&0.066&0.046\\ \hline \hline
&\multicolumn{6}{c}{Scatter Matrix Case (ii)}\\ \hline
&\multicolumn{6}{c}{Multivariate Normal Distribution}\\ \hline
SS-MAX&0.038&0.054&0.032&0.044&0.044&0.032\\
CPZ&0.07&0.052&0.074&0.058&0.056&0.05\\ \hline
&\multicolumn{6}{c}{Multivariate $t_3$ Distribution}\\ \hline
SS-MAX&0.064&0.08&0.068&0.058&0.068&0.064\\
CPZ&0.07&0.05&0.064&0.064&0.086&0.066\\ \hline
&\multicolumn{6}{c}{Multivariate Mixture Normal Distribution}\\ \hline
SS-MAX&0.062&0.056&0.048&0.038&0.036&0.05\\
CPZ&0.068&0.05&0.05&0.084&0.052&0.052\\ \hline \hline
    	\end{tabular}}
    \end{center}
    \end{table}

\begin{figure}[!ht]
\includegraphics[width=1\textwidth]{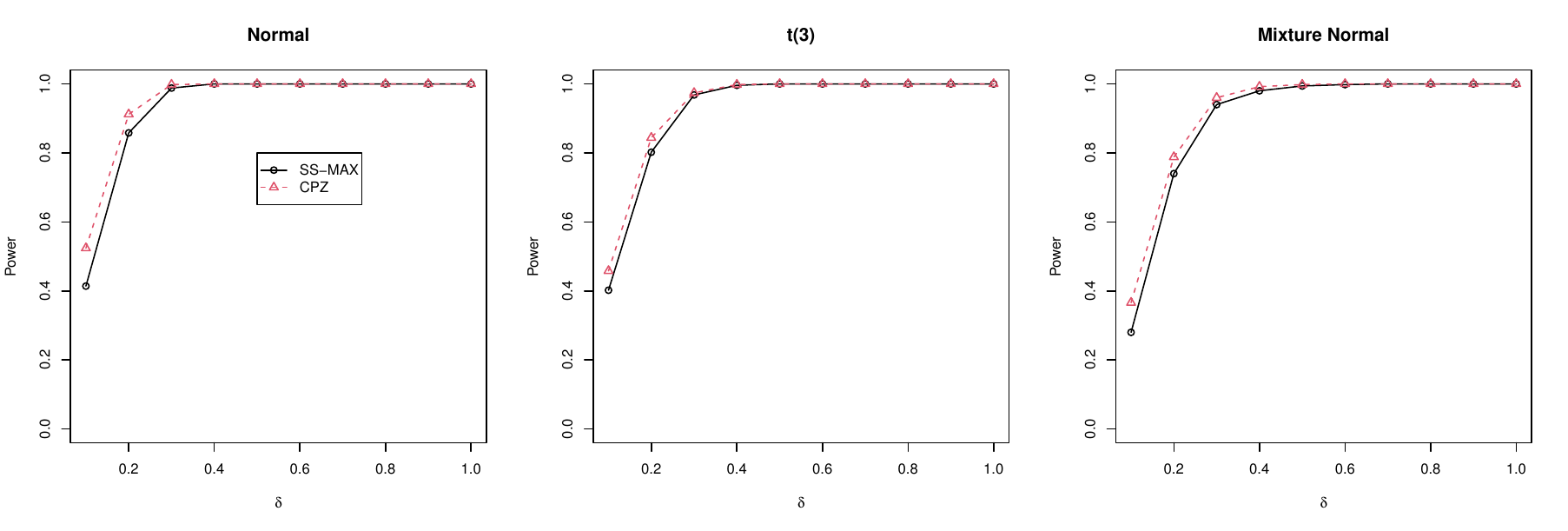}
\caption{Power of max-type tests with different signal strength for matrix case (i) over $(n,p)=(100,200)$. \label{power4}}
\end{figure}

\begin{figure}[!ht]
\includegraphics[width=1\textwidth]{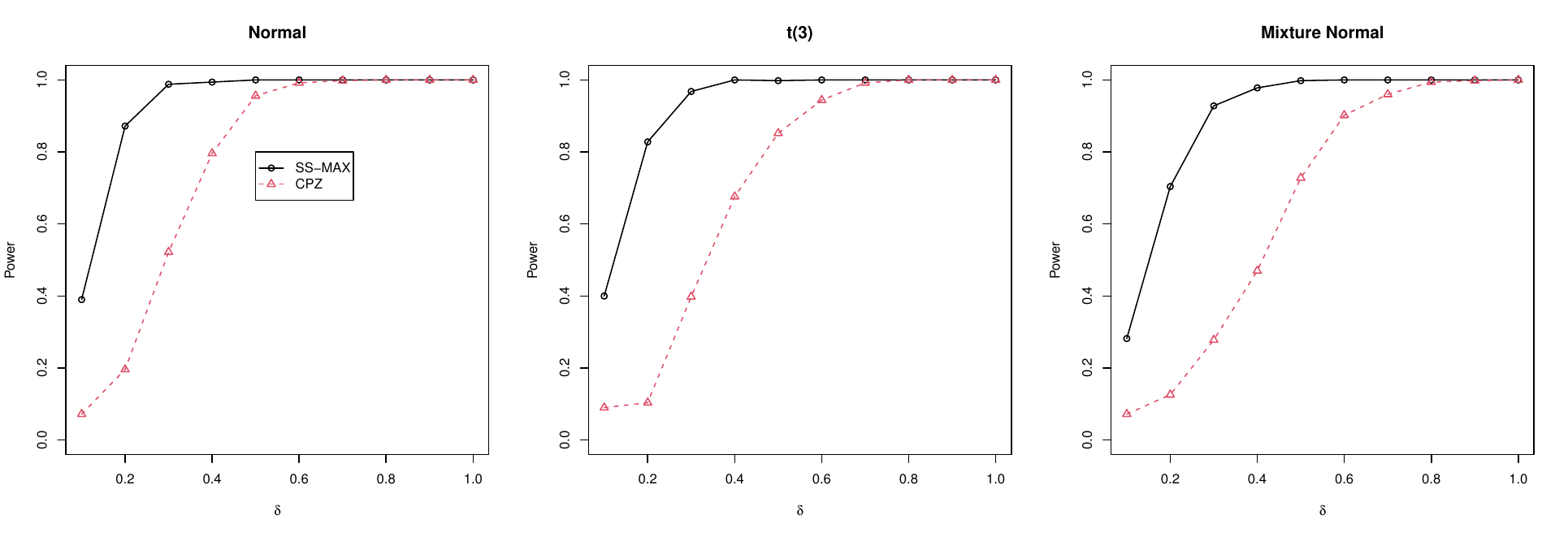}
\caption{Power of tests with different signal strength for matrix case (ii) over $(n,p)=(100,200)$. \label{power5}}
\end{figure}

\section{Application}
\subsection{US stocks data}\label{sec:5.1}
In this section, we utilize our methods to tackle a financial pricing problem. Our goal is to test whether the expected returns of all assets are equivalent to their respective risk-free returns. Let $X_{ij}\equiv R_{ij}-\textrm{rf}_i$ denote the excess return of the $j$th asset at
time $i$ for $i=1,\cdots,n$ and $j=1,\cdots,p$, where $R_{ij}$ is the
return on asset $j$ during period
$i$ and $\textrm{rf}_i$ is the risk-free return rate of all asset during period $i$. We
study the following pricing model
\begin{align}\label{rf}
X_{ij}=\mu_j+\xi_{ij},
\end{align}
for $i=1,\cdots,n$ and $j=1,\cdots,p$, or, in vector form,
$\X_i=\bm \mu+\bm\xi_i$,  where $\X_i=(X_{i1}, \dots, X_{ip})\trans$,
$\bm \mu=(\mu_1, \dots, \mu_p)\trans$, and $\bm\xi_i=(\xi_{i1}, \dots,
\xi_{ip})\trans$ is the zero-mean error vector. We consider the following hypothesis
\begin{align*}
H_0: \bm \mu=\bm 0~~\text{versus}~~ H_1: \bm \mu\not=\bm 0.
\end{align*}

We examined the weekly return rates of stocks that are part of the S\&P 500 index from January 2005 to November 2018. The weekly data were derived from the stock prices every Friday. Over time, the composition of the index changed and some stocks were introduced during this period. Therefore, we only considered a total of 424 stocks that were consistently included in the S\&P 500 index throughout this period. We compiled a total of 716 weekly return rates for each stock during this period, excluding Fridays that were holidays. Given the possibility of autocorrelation in the weekly stock returns, we applied the Ljung-Box test \cite{ljung1978measure} at a 0.05 level for zero autocorrelations to each stock. We retained 280 stocks for which the Ljung-Box test at a 0.05 level was not rejected. It's important to note that if we had used all 424 stocks, there might be autocorrelation between observations, which would violate our assumption and necessitate further studies.

Figure \ref{sd} show the histogram of standard deviation of those 280 securities. We found that the variances of those assets are obviously not equal. So the scalar-invariant test procedure is preferred. Thus, We apply the above six test procedures--SS-SUM,SS-MAX,SS-CC,MAX,SUM,COM on the total samples. All the tests reject the null hypothesis significantly. To evaluate the performance of our proposed tests and other competing tests for both small and large sample sizes, we randomly sampled $n=52K$ observations from the 716 weekly returns, where $K$ ranges from 3 to 8. This experiment was repeated 1000 times for each $n$ value.

Table \ref{R1} presents the rejection rates of six tests. We discovered that spatial-sign based test procedures outperform mean-based test procedures. This is primarily due to the heavy-tailed nature of asset returns. Figure \ref{qq} displays Q-Q plots of the weekly return rates of some stocks in the S\&P 500 index. We observed that all data deviate from a normal distribution and exhibit heavy tails. Additionally, sum-type tests perform better than max-type test procedures, mainly because the alternative is dense. Figure \ref{ttest} illustrates the $t$-test statistic for each stock. We noticed that many $t$-test statistics are larger than 2, and most of them are positive. Among these tests, the SS-CC test performs the best. Although the SS-SUM outperforms the SS-MAX, the SS-MAX still retains some power in all cases. As indicated in the theoretical results in Subsection 3.2, our proposed Cauchy Combination would be more powerful than both max-type and sum-type tests in this scenario. Therefore, the application of real data also demonstrates the superiority of our proposed maxsum test procedure.

It's worth noting that the rejection of the null hypothesis, which suggests that return rates are not solely composed of risk-free rates on average, aligns with the perspectives of numerous economists. Indeed, the consideration of a non-zero excess return rate and the attempt to model it has spurred a vast amount of research on factor pricing models in finance \citep{Sharpe1964CAPITALAP,Fama1993,Fama2015}. These models, which have many practical applications, operate under the Arbitrage Pricing Theory \citep{Rose1976}. Recently, numerous studies have focused on the high-dimensional alpha test under the linear factor pricing model, including works by \cite{Fan2015,Pesaran2017,feng2022high,liu2023high}. Notably, \cite{liu2023high} proposed a spatial-sign based sum-type test procedure for testing alpha for heavy-tailed distributions. It would be intriguing to extend the methods presented in this paper to propose a spatial sign based max-type and maxsum-type test procedures for testing alpha. This is an area that warrants further exploration.

\begin{figure}[!ht]
\begin{center}
\includegraphics[width=0.6\textwidth]{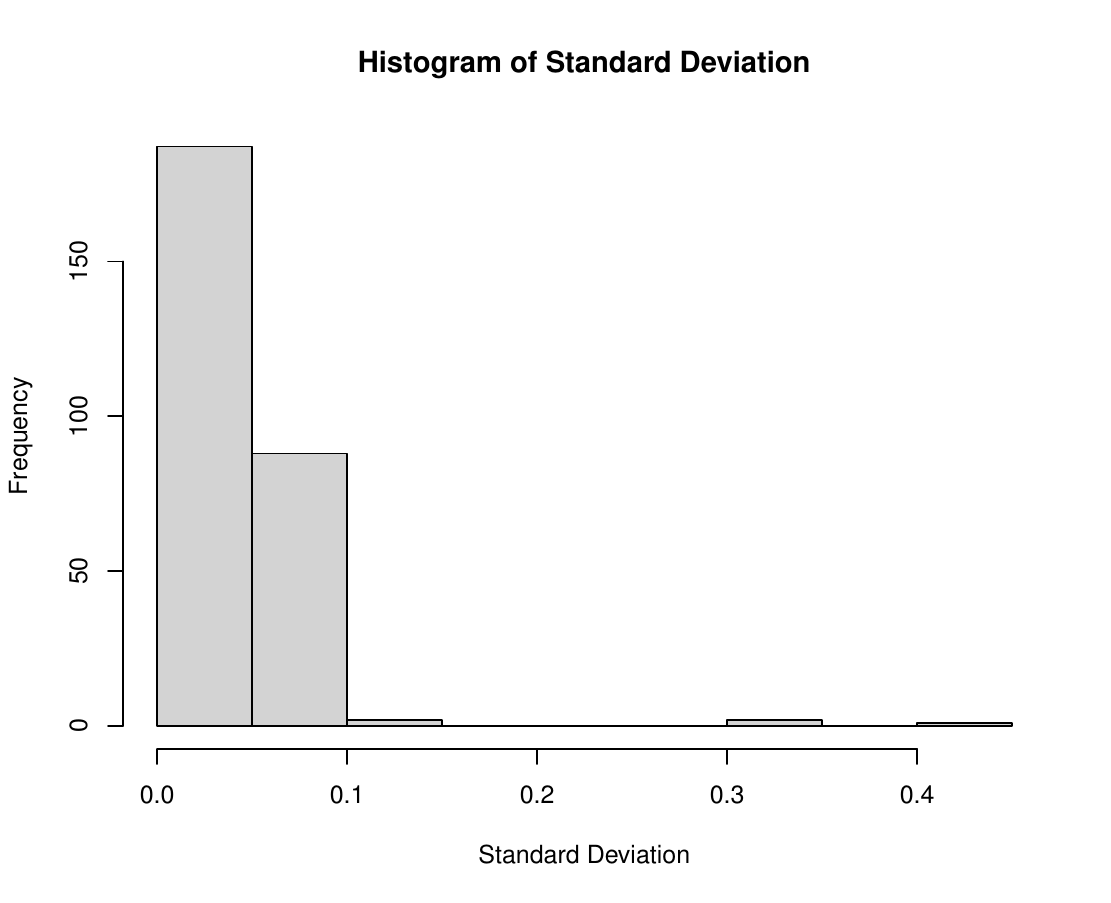}
\caption{Histogram of standard deviation of US securities. \label{sd}}
\end{center}
\end{figure}

\begin{figure}[!ht]
\includegraphics[width=\textwidth]{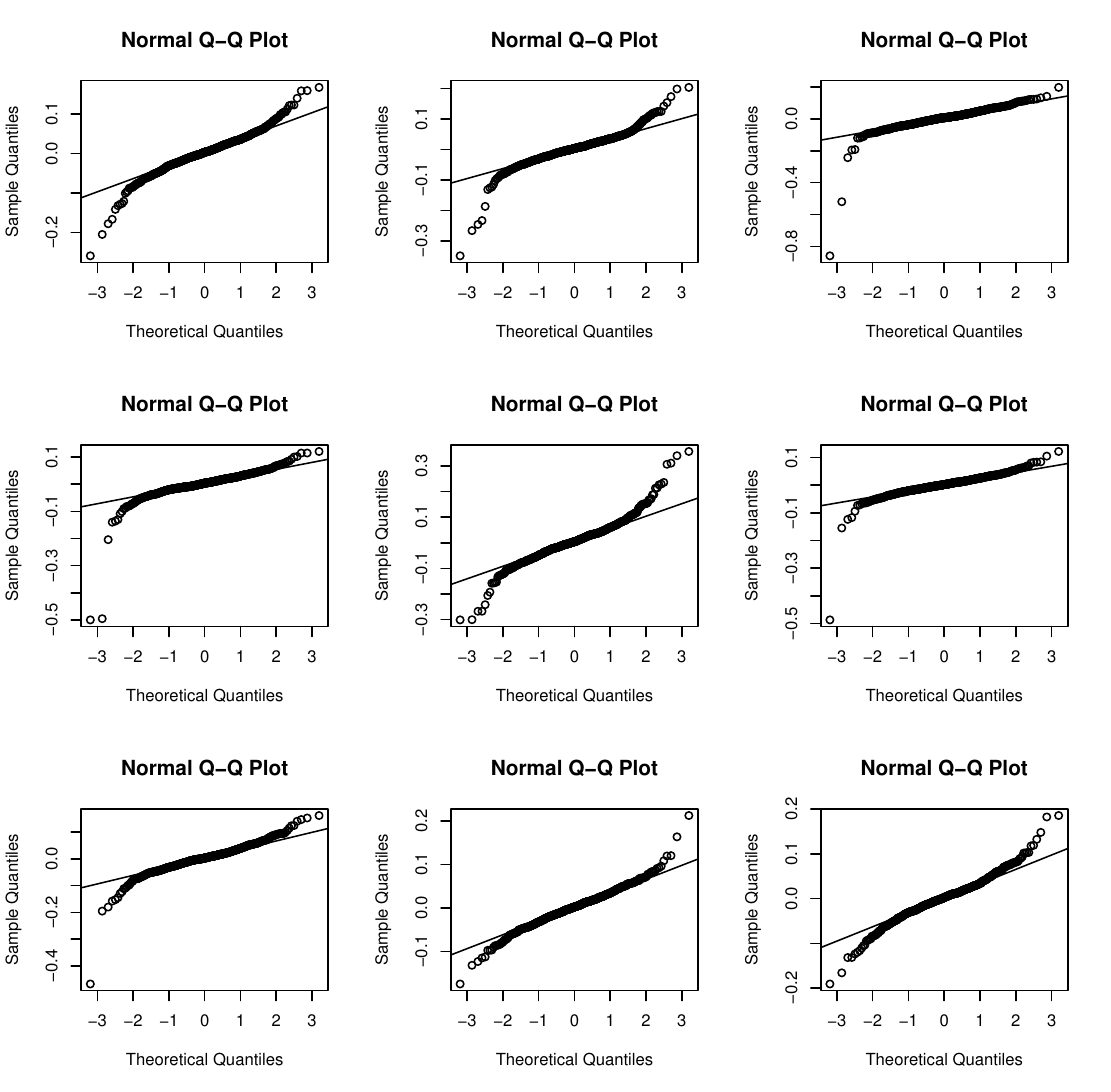}
\caption{Q-Q plots of the weekly return rates of some stocks with heavy-tailed distributions in the S\&P500 index. \label{qq}}
\end{figure}

\begin{table}[htbp]
 \centering
 \caption{{\color{black}The rejection rates of testing excess returns of
     the S\&P stocks for $p=280$ and $n=52K$ with $K=3,\cdots,8$.
For each $n$, we sampled  1000 data sets.} }
 \vspace{0.2cm}
 \begin{tabular}{ccccccc}\hline \hline
    &SS-MAX &SS-SUM & SS-CC&MAX&SUM&COM \\
$n=156$&0.295&0.361&0.380&0.124&0.219&0.204\\
$n=208$&0.364&0.448&0.458&0.128&0.217&0.206\\
$n=260$&0.424&0.542&0.556&0.140&0.276&0.246\\
$n=312$&0.506&0.633&0.645&0.137&0.272&0.236\\
$n=364$&0.652&0.738&0.758&0.143&0.317&0.282\\
$n=416$&0.753&0.821&0.843&0.163&0.339&0.303\\
\hline \hline
\end{tabular}\label{R1}
\end{table}

\begin{figure}[!ht]
\begin{center}
\includegraphics[width=0.6\textwidth]{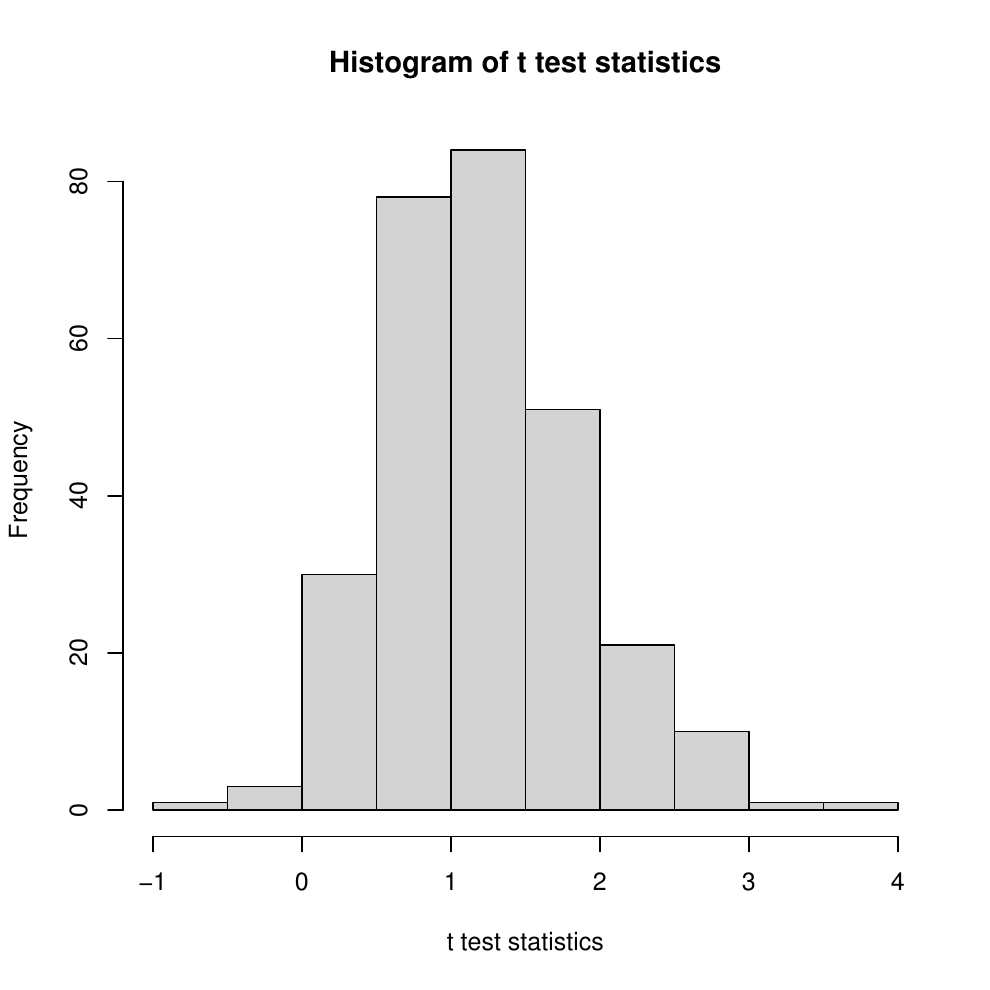}
\caption{$t$ test statistics of the weekly excess return rates of each stock. \label{ttest}}
\end{center}
\end{figure}
\subsection{  Paired colon dataset}
Another important application of the one-sample test discussed in this paper is assessing the mean difference between two paired samples. In this section, we utilize our methods to test the mean difference for paired samples and consider the colon dataset provided by \cite{alon1999broad}. The colon dataset includes gene expression data from 40 colon cancer patients, comprising 22 paired samples from normal and tumor colon tissue and additional 18 samples from tumor tissue, with each sample containing 2,000 gene expressions. Our objective is to assess whether the mean gene expression levels differ between normal and tumor tissues. To streamline the analysis, we exclude unpaired samples and retain only the $n=22$ paired normal and tumor tissue samples. To compare the methods, MAX, SUM and COM methods are also displayed. We observed that SUM test fails to reject the null hypothesis while others strongly reject if the significant level is set to $\alpha=0.05$, see Figure \ref{R2}. It is aligned with the simulation results of non-normal cases, suggesting a significant difference in gene expression between normal and tumor tissues, warranting further investigation. In addition, we find that the COM test successfully rejects the null hypothesis, whereas the SUM test does not.  This indicates that methods based on the theorem of the independence between the test statistic can enhance test power while ensuring that Type I error remains controlled, particularly when data sparsity is uncertain.

\begin{table}[htbp]
 \centering
 \caption{{\color{black}The p-values of testing the difference of
     gene expression levels of the normal and tumor colon tissues.} }
 \vspace{0.2cm}
 \begin{tabular}{cccccc}\hline \hline
    SS-MAX &SS-SUM & SS-CC&MAX&SUM&COM \\
$6.88\times 10^{-5}$&$6.71\times 10^{-7}$&$1.33\times 10^{-6}$&$1.68\times 10^{-3}$&$9.62\times 10^{-2}$&$3.30\times 10^{-3}$\\
\hline \hline
\end{tabular}\label{R2}
\end{table}

\section{Conclusion}
In this paper, we address a one-sample testing problem in high-dimensional settings for heavy-tailed distributions. We begin by providing a Bahadur representation and Gaussian approximation of the spatial median estimator, as discussed in \cite{feng2016multivariate}. Following this, we introduce a spatial-sign based max-type test procedure for sparse alternatives and establish the limit null distribution and consistency of the proposed max-type test statistic. Next, we reformulate the sum-type test statistic, originally proposed by \cite{feng2016}, under a general model. This sum-type test exhibits superior performance under dense alternatives. Finally, we demonstrate the asymptotic independence between the aforementioned max-type test statistic and the sum-type test statistic, given some mild conditions. We then propose a Cauchy combination test procedure, which performs exceptionally well under both sparse and dense alternatives. Both simulation studies and real data applications underscore the superiority of the proposed maxsum-type test procedure.

We propose several directions for future research. Firstly, the sum-type test statistic in \cite{feng2016} only takes into account the direction of the sample, neglecting the information of the sample's radius. \cite{feng2021inverse} introduced a more powerful inverse norm sign test. It would be intriguing to derive a max-type test statistic that also considers the radius of the sample. Furthermore, it remains an open question whether this new max-type test statistic maintains asymptotic independence with the sum-type test statistic proposed by \cite{feng2021inverse}.

Secondly, the newly proposed methods can be extended to address other high-dimensional testing problems. These include the high-dimensional two-sample location problem \citep{chen2010two,feng2016multivariate}, high-dimensional covariance matrix tests \citep{chen2010tests,li2012two,cutting2017testing,cheng2019testing}, testing the martingale difference hypothesis in high dimension \citep{chang2023testing} and high-dimensional white noise test \citep{paindaveine2016high,chang2017testing,feng2022testingwhite,zhao2023spatial}. Additionally, the alpha test in the high-dimensional linear factor pricing model is a significant problem that has been explored in practical applications.

Thirdly, our paper's theoretical results are predicated on the assumption of identical and independent distribution. However, there may occasionally be auto-correlations among the sample sizes. Recent literature, such as \cite{zhang2018gaussian} and \cite{chang2024central}, has considered the Gaussian approximation of the sample mean under a dependent assumption. Therefore, it would be intriguing to establish the Bahadur representation and Gaussian approximation of the spatial median in the context of dependent observations. Building on these findings, we can also suggest the implementation of max-type and maxsum-type testing methods for addressing high-dimensional location problem in the context of dependent observations \citep{ayyala2017mean,ma2024testing}.

\subsection*{Acknowledgement}
The authors thank the editor, the associate editor and three anonymous referees for helpful comments and discussions. This paper is partially supported by Shenzhen Wukong Investment Company, Tianjin Science Fund for Outstanding Young Scholar (23JCJQJC00150), the Fundamental Research Funds for the Central Universities under Grant No. ZB22000105 and 63233075, the China National Key R\&D Program (Grant Nos. 2019YFC1908502, 2022YFA1003703, 2022YFA1003802, 2022YFA1003803) and the National Natural Science Foundation of China Grants (Nos. 12271271, 11925106, 12231011, 11931001 and 11971247). Long Feng and Ping Zhao are the co-corresponding authors.

\section{Appendix}
Recall that $\mathbf D=\operatorname{diag}\{d_1^2,d_2^2,\cdots,d_p^2\}$. For $i=1,2,\cdots,n$,
$\boldsymbol U_i=U(\mathbf D^{-1/2}(\bm X_i-\boldsymbol \theta))$ and $R_i=\Vert \mathbf D^{-1/2}(\bm X_i-\boldsymbol \theta)\Vert$ as the scale-invariant spatial-sign and radius of $\bm X_i-\boldsymbol \theta$, where $U(\bm X)=\bm X/\Vert \bm X\Vert \mathbb I(\bm X\not= 0)$ is the multivariate sign function of $\bm X$, with $\mathbb I(\cdot)$ being the indicator function. The moments of $R_i$ is defined as   $\zeta_k=\mathbb{E}\left(R_i^{-k}\right) $ for  k=1,2,3,4.

 We denote the $\mathbf D$-estimated version $\bm U_i$ and $R_i$ as $\hat R_i=\Vert  \hat{\bf D}^{-1/2}(\bm X_i-\boldsymbol\theta)\Vert$ and
$\hat U_i=\Vert \hat{\bf D}^{-1/2}(\bm X_i-\boldsymbol\theta)/\Vert \hat{\bf D}^{-1/2}(\bm X_i-\boldsymbol\theta)\Vert$, respectively, $i=1,2,\cdots,n$.
\subsection{Proof of main lemmas}
\begin{lemma}\label{lemma1_like_2016}
    Under Assumption \ref{max1}, we have $\mathbb E ( U(\bm W_i)^\top \mathbf{M}  U(\boldsymbol{W}_i))^2=O(p^{-2}\tr(\mathbf{M}^\top \mathbf{M}))$.
\end{lemma}
\begin{proof}
    By Cauchy inequality and Assumption \ref{max1}, we have
$$
\begin{aligned}
&\mathbb E U(\bm W_i)_l^2 U(\bm W_i)_k^2 \leq\frac{1}{p^2}\mathbb E\sum_{s=1}^p\sum_{t=1}^p U(\bm W_i)_s^2 U(\bm W_i)_t^2=p^{-2}\\
&\mathbb E U(\bm W_i)_l^4 \leq\frac{1}{p}\mathbb E\sum_{s=1}^p U(\bm W_i)_s^4 \leq\frac{1}{p}\mathbb E\sum_{s=1}^p\sum_{t=1}^p U(\bm W_i)_s^2 U(\bm W_i)_t^2= p^{-1},\\
\end{aligned}
$$
and
$$
\mathbb E\left( U(\bm W_i)_l U(\bm W_i)_k U(\bm W_i)_s U(\bm W_i)_t\right) \leq \sqrt{E\left(U(\bm W_i)_l^2 U(\bm W_i)_k^2\right) E\left(U(\bm W_i)_s^2 U(\bm W_i)_t^2\right)}.
$$
By the Cauchy inequality,
$$
\sum_{i,k,s,t}a_{lk}a_{st}\leq \sqrt{\sum_{l,k}a_{lk}^2\sum_{s,t}a_{st}^2}\\
\leq \sqrt{\sum_{l,k}^p a_{lk}^2\sum_{s,t}^pa_{st}^2}=\text{tr}(\mathbf M^\top \mathbf M).
$$
Thus, we get
$$
\begin{aligned}
&E\left(\left(U(\bm W_i)^\top \bf{M} U(\bm W_i)\right)^2\right)\\
=&\sum_{l\not= k=1}^p \sum_{s\not= t=1}^p a_{l k} a_{s t} \mathbb E\left( U(\bm W_i)_l U(\bm W_i)_k U(\bm W_i)_s U(\bm W_i)_t\right) +\sum_{l=1}^p \sum_{s=1}^p a_{l l} a_{s s} \mathbb E\left( U(\bm W_i)_l^2 U(\bm W_i)_s^2 \right) \\
\leq & p^{-2}\frac{p^4-p^2}{p^4}\text{tr}(\mathbf{M}^\top \mathbf M)+ p^{-1}\frac{p^2}{p^4}\text{tr}(\mathbf M^\top \mathbf M)=O(p^{-2}\text{tr}(\bf M^\top \bf M)).
\end{aligned}
$$

\end{proof}
\begin{lemma}\label{lemma:D_rate}
Under Assumptions \ref{max1}and \ref{sum2}, we have $\max_{j=1,2,\cdots p} (\hat d_j -d_j )=O_p(n^{-1/2}(\log p)^{1/2})$.
\end{lemma}
\begin{proof}
    The proof of this lemma is along the same lines as the proof of Lemma A.2. in \cite{feng2016multivariate}, but differs in that the assumptions about the model in this paper are more general, with different constraints controlling the correlation matrix $R$.

    Denote $\boldsymbol\eta=(\boldsymbol\theta^\top, d_1,d_2,\cdots,d_p)^\top$ and $\hat{\boldsymbol\eta}$ as the estimated version. By first-order Taylor expansion, we have
    \begin{equation}\label{eq:expansion_of_U_i}
            \begin{aligned}
        U(\mathbf D^{-1/2}(\bm X_i-\bm \theta))&=\frac{\mathbf D^{-1/2}\mathbf{\Gamma} U(\bm W_i)}{(1+U(\bm W_i)^\top (\mathbf R-\mathbf I_p)U(\boldsymbol W_i))^{1/2}}\\
        &=\mathbf D^{-1/2}\mathbf\Gamma U(\boldsymbol W_i)+C_1 U(\boldsymbol W_i)^\top (\mathbf R-\mathbf I_p)U(\boldsymbol W_i) \mathbf D^{-1/2}\mathbf\Gamma U(\boldsymbol W_i),
    \end{aligned}
    \end{equation}
    where $C_1$ is a bounded random variable between $-0.5$ and $-0.5(1+U(\boldsymbol W_i)^\top (\mathbf R-\mathbf I_p)U(\boldsymbol W_i))^{-3/2}$. By Cauchy inequality and Lemma \ref{lemma1_like_2016} and Assumption \ref{sum2}, we get
    $$
\begin{aligned}
    \mathbb E\left(U(\mathbf D^{-1/2}(\boldsymbol X_i-\boldsymbol\theta)\right)&\leq C_2\left\{ \mathbb E(U(\boldsymbol W_i)^\top (\mathbf R-\mathbf I_p)U(\boldsymbol W_i))^2 \mathbb E(\mathbf D^{-1/2}(\boldsymbol X_i-\boldsymbol\theta)^2 \right\}^{1/2} \\
    &=O(p^{-1}\sqrt{\text{tr}(\mathbf R^2)-p})=o(n^{-1/2}).
\end{aligned}
  $$

    Similarly, we can show that
    $$
\begin{aligned}
    &\mathbb E\left(\operatorname{diag}\left\{ U(\mathbf D^{-1/2}(\boldsymbol X_i-\boldsymbol\theta))U(\mathbf D^{-1/2}(\boldsymbol X_i-\boldsymbol\theta))^\top\right\}-\frac{1}{p}\mathbf I_p\right)\leq O(n^{-1/2}),
\end{aligned}
    $$
    by first-order Taylor expansion for $U(\mathbf D^{-1/2}(\boldsymbol X_i-\boldsymbol\theta))U(\mathbf D^{-1/2}(\boldsymbol X_i-\boldsymbol\theta))^\top$, Cauchy inequality and Lemma  \ref{lemma1_like_2016}. The above two equations define the functional equation for each component of $\boldsymbol\eta$,
    \begin{equation}\label{eq:T_j}
        T_j(F,  \eta_j)=o_p(n^{-1/2}),
    \end{equation}
    where $F$ represent the distribution of $\boldsymbol X$, $\boldsymbol\eta=(\eta_1,\cdots,\eta_{2p})$. Similar to \cite{hettmansperger2002practical}, the linearisation of this equation shows
    $$
    n^{1/2}\left(\hat\eta_j-\eta_j\right)=-\mathbf H_{j}^{-1}n^{1/2}\left( T_j(F_n,\eta_j)-T_j(F,\eta_j)\right)+o_p(1),
    $$
    where $F_n$ represents the empirical distribution function based on $\boldsymbol X_1,\boldsymbol X_2,\cdots,\boldsymbol X_n$, $\mathbf H_j$ is the corresponding Hessian matrix of the functional defined in Equation \ref{eq:T_j} and
    $$
T(F_n,\boldsymbol\eta)=\left(n^{-1}\sum_{j=1}^n U(\mathbf D^{-1/2}(\boldsymbol X_i-\boldsymbol\theta))^\top,\vec( \operatorname{diag}(n^{-1}U(\mathbf D^{-1/2}(\boldsymbol X_i-\boldsymbol\theta))U(\mathbf D^{-1/2}(\boldsymbol X_i-\boldsymbol\theta))^\top-\frac{1}{p}\mathbf I_p))  \right).
    $$
    Thus, for each $\hat d_j$ we have
    $$
    \sqrt n (\hat d_j-d_j)\stackrel{d}{\rightarrow}N(0,\sigma^2_{d,j}).
    $$
    where $\sigma^2_{d,j}$ is the corresponding asymptotic variance. Define $\sigma_{d,max}=\max_{1\leq j\leq p}\sigma_{d,j}$. As $p\rightarrow\infty$,
    $$
\begin{aligned}
    \mathbb P&(\max_{j=1,2,\cdots,p}(\hat d_j-d_j)>\sqrt 2\sigma_{d,max}n^{-1/2}(\log p)^{1/2})\\
    &\leq \sum_{j=1}^p \mathbb P(\sqrt n(\hat d_j-d_j)>\sqrt 2\sigma_{d,max}(\log p)^{1/2})\\
    &=\sum_{j=1}^p (1-\Phi(\sqrt 2 \sigma_{d,max}\sigma_{d,j}^{-1}(\log p)^{1/2}))\leq p(1-\Phi((2\log p)^{1/2}))\\
    &\leq \frac{p}{\sqrt{2\pi \log p}}\exp(-\log p)=(4\pi)^{-1/2}(\log p)^{-1/2}\rightarrow 0,
\end{aligned}
    $$
    which means that $\max_{j=1,2,\cdots p} (\hat d_j -d_j )=O_p(n^{-1/2}(\log p)^{1/2})$.
\end{proof}

\begin{lemma}\label{lemma:consistency_zeta}
    Suppose the Assumptions in Lemma \ref{lemma:D_rate} hold, then $\hat\zeta_1\stackrel{p}{\rightarrow}\zeta_1$.
\end{lemma}
\begin{proof}
Denote $\hat{\boldsymbol\mu}=\hat{\boldsymbol\theta}-\boldsymbol\theta$.
    $$
\begin{aligned}
    \Vert \hat{\mathbf D}^{-1/2}(\boldsymbol X_i-\hat{\boldsymbol\theta})\Vert=\Vert {\mathbf D}^{-1/2}(\boldsymbol X_i-{\boldsymbol\theta})\Vert(& 1+ R_i^{-2}\Vert (\hat{\mathbf D}^{-1/2}-{\mathbf D}^{-1/2})(\boldsymbol X_i-{\boldsymbol\theta})\Vert^2
    \\
    &+ R_i^{-2}\Vert \hat{\mathbf{D}}^{-1/2}\hat {\boldsymbol\mu}\Vert^2+2R_i^{-2}\boldsymbol U_i^\top (\hat{\mathbf D}^{-1/2}-{\mathbf D}^{-1/2}){\mathbf D}^{1/2}\boldsymbol U_i)\\
    &-2R_i^{-1}\boldsymbol U_i^\top \hat{\mathbf D}^{-1/2}\hat{\boldsymbol\mu}-2R_i^{-1}\boldsymbol U_i\mathbf D^{1/2}(\hat{\mathbf D}^{-1/2}-{\mathbf D}^{-1/2})\hat{\mathbf D}^{-1/2}\hat{\boldsymbol\mu})^{1/2}.
\end{aligned}
    $$
    According to the proof and conclusion in Lemma \ref{lemma:D_rate}, we can show that $R_{i}^{-2}\Vert(\hat{\mathbf{D}}^{-1 / 2}-\mathbf{D}^{-1 / 2})(\boldsymbol{X}_{i}-\boldsymbol{\theta})\Vert^2=O_p\left((\log p / n)^{1 / 2}\right)=o_p(1)$ and $R_{i }^{-2}\Vert\hat{\mathbf{D}}^{-1 / 2} \hat{\boldsymbol\mu}\Vert^2=O_p(n^{-1})=o_p(1)$ and by the Cauchy inequality, the other parts are also $o_p(1)$. So,
$$
n^{-1} \sum_{i=1}^{n}\left\|\hat{\mathbf{D}}^{-1 / 2}\left(\boldsymbol{X}_{i }-\hat{\boldsymbol{\theta}}\right)\right\|^{-1}=\left(n^{-1} \sum_{i=1}^{n}\left\|\mathbf{D}^{-1 / 2}\left(\boldsymbol{X}_{i}-\boldsymbol{\theta}\right)\right\|^{-1}\right)\left(1+o_p(1)\right) .
$$
Obviously, $E\left(n^{-1} \sum_{i=1}^{n} R_{i}^{-1}\right)=\zeta_i$ and $\operatorname{var}\left(n^{-1} \zeta_i^{-1} \sum_{i=1}^{n} R_{i}^{-1}\right)=O\left(n^{-1}\right)$. Finally, the proof is completed.
\end{proof}
\begin{lemma}\label{Qjl}
    Suppose Assumptions \ref{max1}-\ref{max3} holds with $a_0(p)\asymp p^{1-\delta}$ for some positive constant $\delta\leq 1/2$. Define a random $p\times p$ matrix $\hat{\mathbf Q}=n^{-1}\sum_{i=1}^n \hat R^{-1}_i \hat{\boldsymbol U}_i \hat{\boldsymbol U}_i^\top$ and let $\hat{\mathbf Q}_{jl}$ be the $(j,l)$th element of $\hat{\mathbf Q}$. Then,
    $$\left|\hat{\mathbf Q}_{j \ell}\right| \lesssim \zeta_1 p^{-1}\left|\sigma_{j \ell}\right|+O_p\left(\zeta_1 n^{-1 / 2} p^{-1}+\zeta_1 p^{-7 / 6}+\zeta_1 p^{-1-\delta / 2}+\zeta_1n^{-1/2}(\log p)^{1/2}( p^{-5/2}+p^{-1-\delta/2})\right).$$
\end{lemma}
\begin{proof}
   Recall that $\hat{\mathbf Q}=\frac{1}{n}\sum_{i=1}^n \hat R_i^{-1}\hat{\boldsymbol U}_i\hat{\boldsymbol U}_i^\top$, then,
$$
\begin{aligned}
\hat{\mathbf Q}_{jl}=&\frac{1}{n}\sum_{i=1}^n \hat R_i^{-1}\hat{\boldsymbol U}_{i,j}\hat{\boldsymbol U}_{i,l}\\
=&\frac{1}{n}\sum_{i=1}^nv_i^{-1}(\hat{\mathbf D}^{-1/2}\mathbf\Gamma_j \boldsymbol W_i)(\hat{\mathbf D}^{-1/2}\mathbf\Gamma_l \boldsymbol W_i)^\top\Vert \hat{\mathbf D}^{-1/2}\mathbf\Gamma \boldsymbol W_i\Vert^{-3}\\
=&\hat{ \mathbf D}^{-1/2}\mathbf D^{1/2}\left\{\frac{1}{n}\sum_{i=1}^n v_i^{-1}( \mathbf D^{-1/2}\mathbf\Gamma_j \boldsymbol W_i)( \mathbf D^{-1/2}\mathbf\Gamma_l \boldsymbol W_i)^\top\Vert \hat{\mathbf D}^{-1/2}\mathbf\Gamma \boldsymbol W_i\Vert^{-3}\right\}\hat{\mathbf D}^{-1/2}\mathbf D^{1/2}.
\end{aligned}
$$

We first consider the middle term,
\begin{equation}\label{eq:Qjl}
\begin{aligned}
&\frac{1}{n}\sum_{i=1}^n v_i^{-1}( \mathbf D^{-1/2}\mathbf\Gamma_j\boldsymbol W_i)( \mathbf D^{-1/2}\mathbf\Gamma_l \boldsymbol W_i)^\top\Vert \hat{ \mathbf D}^{-1/2}\mathbf\Gamma \boldsymbol W_i\Vert^{-3}\\
=&\frac{1}{n}\sum_{i=1}^n v_i^{-1}( \mathbf D^{-1/2}\mathbf\Gamma_j \boldsymbol W_i)( \mathbf D^{-1/2}\mathbf\Gamma_l \boldsymbol W_i)^\top\left\{\Vert \hat{ \mathbf D}^{-1/2}\mathbf\Gamma \boldsymbol W_i\Vert^{-3}-\Vert  \mathbf D^{-1/2}\mathbf\Gamma \boldsymbol W_i\Vert^{-3}\right\}\\
&\quad+\frac{1}{n}\sum_{i=1}^n v_i^{-1}( \mathbf D^{-1/2}\mathbf\Gamma_j \boldsymbol W_i)( \mathbf D^{-1/2}\mathbf\Gamma_l \boldsymbol W_i)^\top\left\{\Vert  \mathbf D^{-1/2}\mathbf\Gamma \boldsymbol W_i\Vert^{-3}-p^{-3/2}\right\}\\
&\quad\quad+n^{-1}p^{-3/2}\sum_{i=1}^n v_i^{-1}( \mathbf D^{-1/2}\mathbf\Gamma_j \boldsymbol W_i)( \mathbf D^{-1/2}\mathbf\Gamma_l \boldsymbol W_i)^\top.
\end{aligned}
\end{equation}
The first part in Equation \ref{eq:Qjl}:
\begin{equation}\label{eq:Qjl_p1}
\begin{aligned}
&\mathbb E[\frac{1}{n}\sum_{i=1}^n v_i^{-1}( \mathbf D^{-1/2}\mathbf\Gamma_j \boldsymbol W_i)( \mathbf D^{-1/2}\mathbf\Gamma_l \boldsymbol W_i)^\top\left\{\Vert \hat{\mathbf D}^{-1/2}\mathbf\Gamma \boldsymbol W_i\Vert^{-3}-\Vert  \mathbf D^{-1/2}\mathbf\Gamma \boldsymbol W_i\Vert^{-3}\right\}]\\
=&\mathbb E[v_i^{-1}( \mathbf D^{-1/2}\mathbf\Gamma_j\boldsymbol W_i)( \mathbf D^{-1/2}\mathbf\Gamma_l\boldsymbol W_i)^\top\left\{\Vert \hat{\mathbf D}^{-1/2}\mathbf\Gamma \boldsymbol W_i\Vert^{-3}-\Vert  \mathbf D^{-1/2}\mathbf\Gamma \boldsymbol W_i\Vert^{-3}\right\}]\\
=&\mathbb E[v_i^{-1}( \mathbf D^{-1/2}\mathbf\Gamma_j\boldsymbol W_i)( \mathbf D^{-1/2}\mathbf\Gamma_l \boldsymbol W_i)^\top\left\{\Vert \hat{\mathbf D}^{-1/2}\mathbf\Gamma \boldsymbol W_i\Vert^{-3}\Vert  \mathbf D^{-1/2}\Gamma W_i\Vert^{-3}\left(\Vert  \mathbf D^{-1/2}\mathbf\Gamma \boldsymbol W_i\Vert^{3}-\Vert \hat{\mathbf D}^{-1/2}\mathbf\Gamma \boldsymbol W_i\Vert^{3}\right)\right\}].\\
\end{aligned}
\end{equation}

To compute the order of Equation \ref{eq:Qjl_p1}, we consider $\Vert \hat{\mathbf D}^{-1/2}\mathbf\Gamma \boldsymbol W_i\Vert^{k}-\Vert  \mathbf D^{-1/2}\mathbf\Gamma \boldsymbol W_i\Vert^{k}$ for $k=1,2,\cdots$.

Firstly, for $k=2$, By the Lemma \ref{lemma:D_rate} and Assumption \ref{max3}, we can see that,
$$
\begin{aligned}
&\max_{i=1,2,\cdots,p} ({\frac{d_i}{\hat d_i}}-1)
=\max_{i=1,2,\cdots,p} \frac{d_i-\hat d_i}{{\hat d_i}}=O_p(n^{-1/2}(\log p)^{1/2}).\\
\end{aligned}
$$
So, for $\Vert \hat{\mathbf D}^{-1/2}\mathbf\Gamma \boldsymbol W_i\Vert^{2}$,
$$
\begin{aligned}
&\Vert \hat{\mathbf D}^{-1/2}\mathbf\Gamma \boldsymbol W_i\Vert^{2}\\
=&\Vert (\hat{\mathbf D}^{-1/2}\mathbf D^{1/2}-\mathbf I_p)\mathbf D^{-1/2}\mathbf\Gamma \boldsymbol W_i+ \mathbf D^{-1/2}\mathbf\Gamma \boldsymbol W_i\Vert^2\\
=&\Vert \mathbf D^{-1/2}\mathbf\Gamma \boldsymbol W_i\Vert^2+\Vert (\hat{\mathbf D}^{-1/2}\mathbf D^{1/2}-\mathbf I_p)\mathbf D^{-1/2}\mathbf\Gamma \boldsymbol W_i\Vert^2+\boldsymbol W_i\mathbf\Gamma^\top \mathbf D^{-1/2} (\hat{\mathbf D}^{-1/2}\mathbf D^{1/2}\mathbf I_p)\mathbf D^{-1/2}\mathbf\Gamma \boldsymbol W_i\\
\leq &\Vert \mathbf D^{-1/2}\mathbf\Gamma \boldsymbol W_i\Vert^2\left(1+\max_{i=1,2,\cdots,p} ({\frac{d_i}{\hat d_i}}-1)^2+\max_{i=1,2,\cdots,p} ({\frac{d_i}{\hat d_i}}-1)\right)\\
:=&\Vert \mathbf D^{-1/2}\mathbf\Gamma \boldsymbol W_i\Vert^2\left(1+H\right),\\
\end{aligned}
$$
where $H:=\max_{i} ({\frac{d_i}{\hat d_i}}-1)^2+\max_{i} ({\frac{d_i}{\hat d_i}}-1)=O_p(n^{-1/2}(\log p)^{1/2})$.

For all integer $k$,
\begin{equation}\label{eq:k}
\begin{aligned}
\Vert \hat{\mathbf D}^{-1/2}\mathbf\Gamma \boldsymbol W_i\Vert^{k}
=&\Vert (\hat{\mathbf D}^{-1/2}\mathbf D^{1/2}-\mathbf I_p)\mathbf D^{-1/2}\mathbf\Gamma \boldsymbol W_i+ \mathbf D^{-1/2}\mathbf\Gamma \boldsymbol W_i\Vert^{k}\\
=&\left\{\Vert (\hat{\mathbf D}^{-1/2}\mathbf D^{1/2}-\mathbf I_p)\mathbf D^{-1/2}\mathbf\Gamma \boldsymbol W_i+ \mathbf D^{-1/2}\mathbf\Gamma \boldsymbol W_i\Vert^{2}\right\}^{k/2}\\
\leq &\Vert \mathbf D^{-1/2}\mathbf\Gamma \boldsymbol W_i\Vert^{k}\left(1+H\right)^{k/2}\\
:=&\Vert \mathbf D^{-1/2}\mathbf\Gamma \boldsymbol W_i\Vert^{k}\left(1+H_k\right),
\end{aligned}
\end{equation}
where $H_k$ is defined as $H_k=(1+H)^{k/2}-1=O_p(n^{-1/2}(\log p)^{1/2})$.

Thus, from the proof of Lemma A3. in \cite{cheng2023}, Equation \ref{eq:Qjl_p1} equals
$$
\begin{aligned}
&\mathbb E[\frac{1}{n}\sum_{i=1}^n v_i^{-1}( \mathbf D^{-1/2}\mathbf\Gamma_j \boldsymbol W_i)( \mathbf D^{-1/2}\mathbf\Gamma_l \boldsymbol W_i)^\top\left\{\Vert \hat{\mathbf D}^{-1/2}\mathbf\Gamma \boldsymbol W_i\Vert^{-3}-\Vert  \mathbf D^{-1/2}\mathbf\Gamma \boldsymbol W_i\Vert^{-3}\right\}]\\
=&\mathbb E[v_i^{-1}( \mathbf D^{-1/2}\mathbf\Gamma_j\boldsymbol W_i)( \mathbf D^{-1/2}\mathbf\Gamma_l\boldsymbol W_i)^\top\left\{\Vert \hat{\mathbf D}^{-1/2}\mathbf\Gamma \boldsymbol W_i\Vert^{-3}-\Vert  \mathbf D^{-1/2}\mathbf\Gamma \boldsymbol W_i\Vert^{-3}\right\}]\\
=&\mathbb E[v_i^{-1}( \mathbf D^{-1/2}\mathbf\Gamma_j \boldsymbol W_i)( \mathbf D^{-1/2}\mathbf\Gamma_l\boldsymbol W_i)^\top\Vert  \mathbf D^{-1/2}\mathbf\Gamma \boldsymbol W_i\Vert^{-3}H_3]\\
=&\mathbb E[v_i^{-1}( \mathbf D^{-1/2}\mathbf\Gamma_j\boldsymbol W_i)( \mathbf D^{-1/2}\mathbf\Gamma_l\boldsymbol W_i)^\top(\Vert  \mathbf D^{-1/2}\mathbf\Gamma \boldsymbol W_i\Vert^{-3}-p^{-3/2})H_3]\\
&\quad+p^{-3/2}\mathbb E[v_i^{-1}( \mathbf D^{-1/2}\mathbf\Gamma_j\boldsymbol W_i)( \mathbf D^{-1/2}\mathbf\Gamma_l\boldsymbol W_i)^\top H_3]\\
\lesssim & n^{-1/2}(\log p)^{1/2}\cdot \zeta_1 p^{-5/2}(1+p^{3/2-\delta/2})\\
=&\zeta_1n^{-1/2}(\log p)^{1/2}( p^{-5/2}+p^{-1-\delta/2}).\\
\end{aligned}
$$
The second and last part in  Equation \ref{eq:Qjl}:

From the proof of Lemma A3. in \cite{cheng2023}, we can conclude,
$$
\begin{aligned}
\frac{1}{n}\sum_{i=1}^nv_i^{-1}( \mathbf D^{-1/2}\mathbf\Gamma_j\boldsymbol W_i)( \mathbf D^{-1/2}\mathbf\Gamma_l\boldsymbol W_i)^\top\left\{\Vert  \mathbf D^{-1/2}\mathbf\Gamma \boldsymbol W_i\Vert^{-3}-p^{-3/2}\right\}
=O_p(\zeta_1 p^{-1-\delta/2}),
\end{aligned}
$$
and
$$
n^{-1}p^{-3/2}\sum_{i=1}^nv_i^{-1}( \mathbf D^{-1/2}\mathbf\Gamma_j\boldsymbol W_i)( \mathbf D^{-1/2}\mathbf\Gamma_l\boldsymbol W_i)^\top \lesssim \zeta_1 p^{-1}\left|\sigma_{j \ell}\right|+O_p\left(\zeta_1 n^{-1 / 2} p^{-1}+\zeta_1 p^{-7 / 6}\right).
$$
It follows that,
$$
\begin{aligned}
\mathbf Q_{jl}&=\left(n^{-1}p^{-3/2}\sum_{i=1}^nv_i^{-1}( \mathbf D^{-1/2}\mathbf\Gamma_j\boldsymbol W_i)( \mathbf D^{-1/2}\mathbf\Gamma_l\boldsymbol W_i)^\top+O_p(A_n)\right)(1+O_p(n^{-1}\log p)),
\end{aligned}
$$
where $A_n=\zeta_1n^{-1/2}(\log p)^{1/2}( p^{-5/2}+p^{-1-\delta/2})+\zeta_1 p^{-1-\delta/2}$.

Thus,
$$
\left|\mathbf Q_{j \ell}\right| \lesssim \zeta_1 p^{-1}\left|\sigma_{j \ell}\right|+O_p\left(\zeta_1 n^{-1 / 2} p^{-1}+\zeta_1 p^{-7 / 6}+\zeta_1 p^{-1-\delta / 2}+\zeta_1n^{-1/2}(\log p)^{1/2}( p^{-5/2}+p^{-1-\delta/2})\right).
$$
\end{proof}
\begin{lemma}\label{zetaU}
    Suppose Assumptions \ref{max1}-\ref{max3} hold with $a_0(p)\asymp p^{1-\delta}$ for some positive constant $\delta\leq 1/2$. Then, if $\log p=o(n^{1/3})$,

 (i)$\Vert n^{-1}\sum_{i=1}^n\zeta_1^{-1}\hat{\boldsymbol U}_i\Vert_\infty=O_p\left\{n^{-1/2}\log ^{1 / 2}(n p)\right\}$,

 (ii)$\Vert \zeta_1^{-1}n^{-1}\sum_{i=1}^n \delta_{1i}\hat{\boldsymbol U}_i\Vert_\infty=O_p(n^{-1})$.
\end{lemma}
\begin{proof}
    For any $j\in\{1,2,\cdots,p\}$,
\begin{equation}\label{eq:H_u}
\begin{aligned}
\hat{\boldsymbol U}_{ij}-{\boldsymbol U}_{ij}&=\frac{\Vert \mathbf D^{-1/2} X_i\Vert}{\Vert \hat{\mathbf D}^{-1/2} X_i\Vert}\cdot \frac{d_j}{\hat d_j}{\boldsymbol U}_{ij}-{\boldsymbol U}_{ij}\\
&\leq(1+H)(1+H){\boldsymbol U}_{ij}-{\boldsymbol U}_{ij}\\
&=H_u {\boldsymbol U}_{ij},
\end{aligned}
\end{equation}
where $H_u=O_p(H^2+2H)=O_p(n^{-1/2}(\log p)^{1/2})$. Thus, $\hat{\boldsymbol U}_i-{\boldsymbol U}_i=H_u {\boldsymbol U}_i$.

(i)By Equation \ref{eq:H_u}, we have,
$$
\begin{aligned}
&\left\Vert n^{-1}\sum_{i=1}^n\zeta_1^{-1}\hat{\boldsymbol U}_i\right\Vert_\infty=\left\Vert n^{-1}\sum_{i=1}^n\zeta_1^{-1}(1+H_u) {\boldsymbol U}_i\right\Vert_\infty\\
\leq &|1+H_u|\cdot \left\Vert n^{-1}\sum_{i=1}^n\zeta_1^{-1}{\boldsymbol U}_i\right\Vert_\infty
=O_p\left\{n^{-1/2}\log ^{1 / 2}(n p)\right\}.
\end{aligned}
$$.

(ii)Similarly,
$$
\begin{aligned}
&\left| \zeta_1^{-1}n^{-1}\sum_{i=1}^n \delta_{1i}\hat{\boldsymbol U}_i\right|_\infty
\leq|1+H_u|\cdot \left\Vert \zeta_1^{-1}n^{-1}\sum_{i=1}^n \delta_{1i} {\boldsymbol U}_i\right\Vert_\infty\\
\leq & O_p(n^{-1}(1+n^{-1/2}\log^{1/2}p))= O_p(n^{-1}).
\end{aligned}
$$

\end{proof}
\subsubsection{Proof of Lemma \ref{lemma1}(Bahadur representation)}
\begin{proof}
    As $\boldsymbol\theta$ is a location parameter, we assume $\boldsymbol\theta=0$ without loss of generality. Then ${\boldsymbol U}_i$ can be written as ${\boldsymbol U}_i=\mathbf D^{-1/2}\boldsymbol X_i/ \Vert \mathbf D^{-1/2}\boldsymbol X_i \Vert=\mathbf D^{-1/2}\mathbf \Gamma \boldsymbol W_i/\Vert \mathbf D^{-1/2}\mathbf \Gamma \boldsymbol W_i\Vert$ for $i=1,2,\cdots,n$. The estimator $\hat{\boldsymbol\theta}$ satisfies $\sum_{i=1}^n U(\hat{\mathbf D}^{-1/2}(\boldsymbol X_i-\hat{\boldsymbol\theta}))=0$, which is is equivalent to
$$
\frac{1}{n}\sum_{i=1}^n (\hat{\boldsymbol U}_i-\hat R_i^{-1}\hat{\mathbf D}^{-1/2}\hat{\boldsymbol\theta})(1-2\hat R_i^{-1}\hat{\boldsymbol U}_i^\top \hat{\mathbf D}^{-1/2}\hat{ \boldsymbol\theta}+\hat R_i^{-2}\hat{ \boldsymbol\theta}^\top\hat{\mathbf D}^{-1}\hat{ \boldsymbol\theta})^{-1/2}=0.
$$
From the proof of lemma A.3 in \cite{feng2016multivariate}, we can see that $\Vert \hat{\boldsymbol \theta}\Vert=O_p(\zeta_1^{-1}n^{-1/2})$. By the first-order Taylor expansion, the above equation can be rewritten as:
$$
n^{-1} \sum_{i=1}^n\left(\hat{\boldsymbol U}_i-\hat R_i^{-1} \hat{\mathbf D}^{-1/2}\hat{\boldsymbol{\theta}}\right)\left(1+\hat R_i^{-1} \hat{\boldsymbol U}_i^{\top} \hat{ \mathbf D}^{1/2}\hat{\boldsymbol{\theta}}-2^{-1} \hat R_i^{-2}\left\|\hat{ \mathbf D}^{-1/2}\hat{\boldsymbol{\theta}}\right\|^2+\delta_{1 i}\right)=0,
$$
where $\delta_{1 i}=O_p\left\{\left(\hat R_i^{-1} \hat{\boldsymbol U}_i^{\top} \hat{ \mathbf D}^{1/2}\hat{\boldsymbol{\theta}}-2^{-1} \hat R_i^{-2}\left\|\hat{\mathbf D}^{-1/2}\hat{\boldsymbol{\theta}}\right\|^2\right)^2\right\}=O_p\left(n^{-1}\right)$, which implies
\begin{equation}\label{A1}
\begin{aligned}
&\frac{1}{n}\sum_{i=1}^n(1-\frac{1}{2}\hat R_i^{-2}\hat{\boldsymbol\theta}^\top\hat{\mathbf D}^{-1}\hat{\boldsymbol \theta}+\delta_{1i})\hat{\boldsymbol U}_i+\frac{1}{n}\sum_{i=1}^n \hat R_i^{-1}(\hat{\boldsymbol U}_i^\top\hat{\mathbf D}^{-1/2}\hat{\boldsymbol \theta})\hat U_i\\
=&\frac{1}{n}\sum_{i=1}^n(1+\delta_{1i}+\delta_{2i})\hat R_i^{-1}\hat{\mathbf D}^{-1/2}\hat{\boldsymbol\theta},\\
\end{aligned}
\end{equation}
where $\delta_{2i}=O_p(\hat R_i^{-1} \hat{\boldsymbol U}_i^{\top} \hat{\mathbf D}^{1/2}\hat{\boldsymbol{\theta}}-2^{-1} \hat R_i^{-2}\left\|\hat{\mathbf D}^{-1/2}\hat{\boldsymbol{\theta}}\right\|^2)=O_p(\delta_{1i}^{1/2})$.

Similar to \cite{cheng2023}, by Assumption \ref{max2} and Markov inequality, we have that: $\max R_i^{-2}=O_p(\zeta_1^2 n^{1/2})$, $\max \delta_{1i}=O_p(\Vert \hat{\mathbf D}^{-1/2}\boldsymbol\theta\Vert^2 \max \hat R_i^{-2}=O_p(n^{-1/2})$ and $\max \delta_{2i}=O_p(n^{-1/4})$.

Considering the second term in Equation \ref{A1},
$$
\frac{1}{n}\sum_{i=1}^n \hat R_i^{-1}(\hat U_i^\top\hat{\mathbf D}^{-1/2}\hat{\boldsymbol \theta})\hat U_i=\frac{1}{n}\sum_{i=1}^n \hat R_i^{-1}(\hat U_i\hat U_i^\top\hat{\mathbf D}^{-1/2})\hat{\boldsymbol \theta}=\hat Q\hat{\mathbf D}^{-1/2}\hat{\boldsymbol \theta},
$$
where $\hat{\boldsymbol Q}=\frac{1}{n}\sum_{i=1}^n \hat R_i^{-1}\hat{\boldsymbol U}_i\hat{\boldsymbol U}_i^\top$. From Lemma \ref{Qjl} we acquire,
$$
\left|\mathbf Q_{j \ell}\right| \lesssim \zeta_1 p^{-1}\left|\omega_{j \ell}\right|+O_p\left(\zeta_1 n^{-1 / 2} p^{-1}+\zeta_1 p^{-7 / 6}+\zeta_1 p^{-1-\delta / 2}+\zeta_1n^{-1/2}(\log p)^{1/2}( p^{-5/2}+p^{-1-\delta/2})\right),
$$
and this implies that,
\begin{equation}
    \begin{aligned}
&\left\Vert \mathbf Q \hat{\mathbf D}^{-1/2}\hat{\boldsymbol{\theta}}\right\Vert_{\infty} \\
\leqslant&\|\mathbf Q\|_1 \left\Vert\hat{\mathbf D}^{-1/2}\hat{\boldsymbol{\theta}}\right\Vert_{\infty}\\
\lesssim&\zeta_1 p^{-1}\|\mathbf \Omega\|_1\left\Vert\hat{\mathbf D}^{-1/2}\hat{\boldsymbol{\theta}}\right\Vert_{\infty}\\
+&O_p\left(\zeta_1 n^{-1 / 2} p^{-1}+\zeta_1 p^{-7 / 6}+\zeta_1 p^{-1-\delta / 2}+\zeta_1n^{-1/2}(\log p)^{1/2}( p^{-5/2}+p^{-1-\delta/2})\right)\left\Vert\hat{\mathbf D}^{-1/2}\hat{\boldsymbol{\theta}}\right\Vert_{\infty}\\
&\quad {O_p\left(\zeta_1 n^{-1 / 2} p^{-1}+\zeta_1 p^{-7 / 6}+\zeta_1 p^{-1-\delta / 2}+\zeta_1n^{-1/2}(\log p)^{1/2} p^{-5/2}\right)\left\Vert\hat{\mathbf D}^{-1/2}\hat{\boldsymbol{\theta}}\right\Vert_{\infty}}.
\end{aligned}
\end{equation}
According to Lemma \ref{zetaU}, $\left\Vert n^{-1}\sum_{i=1}^n\zeta_1^{-1}\hat{\boldsymbol U}_i\right\Vert_\infty=O_p\left\{n^{-1/2}\log ^{1 / 2}(n p)\right\}$ and $\left\Vert \zeta_1^{-1}n^{-1}\sum_{i=1}^n \delta_{1i}\hat{\boldsymbol U}_i\right\Vert_\infty=O_p(n^{-1})$. In addition, we obtain,
$$
\begin{aligned}
&\left\Vert\zeta_1^{-1}n^{-1}\sum_{i=1}^n\hat R_i^{-2}\Vert\hat{\mathbf D}^{-1/2}\hat{\boldsymbol \theta}\Vert^2\hat{\boldsymbol U}_i\right\Vert_\infty
\leq|1+H_u|\cdot \left\Vert\zeta_1^{-1}n^{-1}\sum_{i=1}^n\hat R_i^{-2}\Vert\hat{\mathbf D}^{-1/2}\hat{\boldsymbol \theta}\Vert^2 {\boldsymbol U}_i\right\Vert_\infty\\
\lesssim&O_p(n^{-1})(1+{O_p(n^{-3/2}(\log p)^{1/2})})
=O_p(n^{-1}).
\end{aligned}
$$

Considering the third term :

Using the fact that $\zeta_1^{-1}n^{-1}\sum_{i=1}^n R_i^{-1}=1+O_p(n^{-1/2})$ and Equation \ref{eq:k}, we have
$$
\begin{aligned}
&\frac{1}{n}\zeta_1^{-1}\sum_{i=1}^n\hat R_i^{-1}\\
=&\frac{1}{n}\zeta_1^{-1}\sum_{i=1}^n R_i^{-1}(1+O_p(n^{-1/2}(\log p)^{1/2}))\\
=&\left\{1+O_p(n^{-1/2})\right\}\left\{1+O_p(n^{-1/2}(\log p)^{1/2})\right\}\\
=&1+O_p(n^{-1/2}(\log p)^{1/2}).
\end{aligned}
$$
We final obtain:
$$
\begin{aligned}
\left.\hat{\mathbf D}^{-1/2}\hat{\boldsymbol{\theta}}\right|_{\infty} & \lesssim\left\Vert\zeta_1^{-1} n^{-1} \sum_{i=1}^n \hat{\boldsymbol U}_i\right\Vert_{\infty}+\zeta_1^{-1}\left\Vert \mathbf Q \hat{\mathbf D}^{-1/2}\hat{\boldsymbol{\theta}}\right\Vert_{\infty} \\
& \lesssim p^{-1} a_0(p)\left\Vert\hat{\mathbf D}^{-1/2}\hat{\boldsymbol{\theta}}\right\Vert_{\infty}+O_p\left\{n^{-1 / 2} \log ^{1 / 2}(n p)\right\}\\
&\quad +O_p\left(n^{-1 / 2}+p^{-(1 / 6 \wedge \delta / 2)}+n^{-1/2}(\log p)^{1/2} p^{-3/2}\right)\left\Vert\hat{\mathbf D}^{-1/2}\hat{\boldsymbol{\theta}}\right\Vert_{\infty}.
\end{aligned}
$$
Thus we conclude that:
$$
\left\Vert\hat{\mathbf  D}^{-1/2}\hat{\boldsymbol \theta} \right\Vert_\infty=O_p(n^{-1 / 2} \log ^{1 / 2}(n p)),
$$
as $a_0(p)\asymp p^{1-\delta}$.

In addition, we have
$$
\left\Vert\zeta_1^{-1}\mathbf Q\hat{\mathbf D}^{-1/2}\hat{\boldsymbol \theta}\right\Vert_\infty=O_p\left(n^{-1}\log^{1/2}(np)+ n^{-1/2}p^{-(1 / 6\wedge\delta/2)}\log^{1/2}(np)+{n^{-1}(\log p)^{1/2} p^{-3/2}\log^{1/2}(np)}\right),
$$
and
$$
\begin{aligned}
&n^{-1} \sum_{i=1}^n \hat R_i^{-1}\left(1+\delta_{1 i}+\delta_{2 i}\right)\\
=&\zeta_1\left\{1+O_p\left(n^{-1 / 4}\right)\right\}\left\{1+O_p(n^{-1/2}(\log p)^{1/2}\right\}\\
=&\zeta_1\left\{ 1+{O_p(n^{-1/4}+n^{-1/2}(\log p)^{1/2})}\right\}.
\end{aligned}
$$
Finally, we can write
\begin{equation}\label{eq:bahaudr}
\begin{aligned}
n^{1/2}\hat{\mathbf D}^{-1/2}(\hat{\boldsymbol \theta}-\boldsymbol\theta)=n^{-1/2}\zeta_1^{-1}\sum_{i=1}^n {\boldsymbol U}_i+C_n ,
\end{aligned}
\end{equation}
where
$$
\begin{aligned}
\Vert C_n\Vert_\infty=&O_p(n^{-1/2}\log^{1/2}(np)+p^{-(1/6\wedge \delta/2)}\log^{1/2}(np)+n^{-1/2}(\log p)^{1/2}p^{-3/2}\log ^{1/2}(np))\\
&+O_p(n^{-1/4}\log^{1/2}(np)+n^{-1/2}(\log p)^{1/2}\log^{1/2}(np))\\
=&O_p(n^{-1/4}\log^{1/2}(np)+p^{-(1/6\wedge \delta/2)}\log^{1/2}(np)+n^{-1/2}(\log p)^{1/2}\log^{1/2}(np)).
\end{aligned}
$$

\end{proof}
\subsubsection{Proof of Lemma \ref{lemma2}(Gaussian approximation)}
\begin{proof}
   Let $L_{n,p}=n^{-1/4}\log^{1/2}(np)+p^{-(1/6\wedge \delta/2)}\log^{1/2}(np)+n^{-1/2}(\log p)^{1/2}\log^{1/2}(np)$. Then for any sequence $\eta_n\rightarrow \infty$ and any $t\in \mathbb R^{p}$,
   $$
   \begin{aligned}
       \mathbb P(n^{1/2}\hat{\mathbf D}^{-1/2}(\hat{\boldsymbol \theta}-\boldsymbol\theta)\leq t)&=\mathbb P(n^{-1/2}\zeta_1^{-1}\sum_{i=1}^n {\boldsymbol U}_i+C_n \leq t)\\
       &\leq  \mathbb P(n^{-1/2}\zeta_1^{-1}\sum_{i=1}^n {\boldsymbol U}_i\leq t+\eta_n L_{n,p})+\mathbb P(\Vert C_n\Vert_\infty>\eta_n L_{n,p}).
   \end{aligned}
   $$
   According to Lemma \ref{LemmaA4},  $\mathbb E\{(\zeta_1^{-1}{\boldsymbol U}_{i,j})^4\}\lesssim\bar M^2$ and $\mathbb E\{(\zeta_1^{-1}{\boldsymbol U}_{i,j})^2\}\gtrsim \underline{m}$ for all $i=1,2,\cdots,n$, $j=1,2,\cdots,p$, and the Gaussian approximation for independent partial sums in \cite{koike2021notes}, let $\boldsymbol G \sim N\left(0, \zeta_1^{-2} \mathbf \Sigma_u\right)$ with $\mathbf \Sigma_u=\mathbb{E}\left({\boldsymbol U}_1 {\boldsymbol U}_1^{\top}\right)$, we have
   $$
\begin{aligned}
    \mathbb P(n^{1/2}\zeta_1^{-1}\sum_{i=1}^n {\boldsymbol U}_i\leq t+\eta_n L_{n,p})&\leq \mathbb P(\boldsymbol G\leq t+\eta_n L_{n,p})+O(\{n^{-1}\log^5(np)\}^{1/6})\\
    &\leq \mathbb P(\boldsymbol G\leq t)+O\{\eta_nL_{n,p}\log^{1/2}(p)\}+O(\{n^{-1}\log^5(np)\}^{1/6})
\end{aligned}
   $$
   where the second inequality holds from Nazarov's inequality in Lemma \ref{LemmaA6}. Thus,
   $$
   \begin{aligned}
     \mathbb P(n^{1/2}\hat{\mathbf D}^{-1/2}(\hat{\boldsymbol \theta}-\boldsymbol\theta)\leq t)\leq &\mathbb P(\boldsymbol G\leq t)+O\{\eta_n L_{n,p}\log^{1/2}(p)\}+O(\{n^{-1}\log^5(np)\}^{1/6})\\
     &+\mathbb P(\vert C_n\vert_\infty>\eta_n l_{n,p}).
   \end{aligned}
   $$
   On the other hand, we have
   $$
\begin{aligned}
     \mathbb P(n^{1/2}\hat{\mathbf D}^{-1/2}(\hat{\boldsymbol \theta}-\boldsymbol\theta)\leq t)\geq &\mathbb P(\boldsymbol G\leq t)-O\{\eta_n L_{n,p}\log^{1/2}(p)\}-O(\{n^{-1}\log^5(np)\}^{1/6})\\
     &-\mathbb P(\Vert C_n\Vert_\infty>\eta_n l_{n,p}).
   \end{aligned}
   $$
   where $\mathbb P(\Vert C_n\Vert_\infty>\eta_n l_{n,p})\rightarrow 0$ as $n\rightarrow \infty$ by Lemma \ref{lemma1}.

   Then we have that, if $\log p=o(n^{1/5})$ and $\log n=o(p^{1/3 \wedge \delta})$,
   $$
\sup_{t\in \mathbb R^p}\vert\mathbb P(n^{1/2}\hat{\mathbf D}^{-1/2}(\hat{\boldsymbol \theta}-\boldsymbol\theta)\leq t)-\mathbb P(\boldsymbol G\leq t)\vert\rightarrow 0.
   $$
   Further,
   $$
\rho_n(\mathcal A^{re})=\sup_{A\in\mathcal A^{re}}\vert\mathbb P(n^{1/2}\hat{\mathbf D}^{-1/2}(\hat{\boldsymbol\theta}-\boldsymbol\theta)\in A)-\mathbb P(\boldsymbol G\in A)\vert\rightarrow 0,
   $$
   by the Corollary 5.1 in \cite{chernozhukov2017central}.
\end{proof}
\subsubsection{Proof of Lemma \ref{lemma3}(Variance approximation)}
\begin{proof}
    $\mathbb E Z_j^2=\zeta_1^{-2}E(R_i^2)^{-1}\leq\bar {B}$ by Assumption \ref{max2} and $\mathbb E[\max_{1\leq j\leq p}Z_j]\asymp  (\sqrt{\log p+\log\log p})$ by Theorem 2 in \cite{feng2022asymptotic}. Let $\Delta_0=\max_{1\leq j,k\leq p}|p(\mathbb E U_1U_1^\top)_{j,k}-R_{j,k}|$, by Lemma \ref{LemmaA4},
    $$
\begin{aligned}
 \Delta_0=&\max_{1\leq j,k\leq p}|p(\mathbb E {\boldsymbol U}_1{\boldsymbol U}_1^\top)_{j,k}-\mathbf R_{j,k}|=O(p^{-\delta/2}).
\end{aligned}
    $$

According to Lemma \ref{LemmaJ2} ,we have
$$
\sup _{t \in \mathbb{R}}\left|\mathbb{P}\left(\Vert \boldsymbol Z\Vert_{\infty}\leqslant t\right)-\mathbb{P}\left(\Vert \boldsymbol G\Vert_{\infty} \leqslant t\right)\right|\leqslant  C^\prime n^{-1/3}\left(1 \vee \log \left(np \right)\right)^{2 / 3}
 \rightarrow 0.
$$

\end{proof}

\subsection{Proof of main results}
\subsubsection{Proof of Theorem \ref{thm1}(Limit distribution of maxima)}
\begin{proof}
    $$
    \begin{aligned}
        \tilde\rho_n&=\sup _{t \in \mathbb{R}}\left|\mathbb{P}\left(n^{1 / 2}\left\Vert\hat{\mathbf D}^{-1/2}(\hat{\boldsymbol{\theta}}-\boldsymbol{\theta})\right\Vert_{\infty} \leqslant t\right)-\mathbb{P}\left(\Vert \boldsymbol Z\Vert_{\infty} \leqslant t\right)\right| \\
        &\leq \sup _{t \in \mathbb{R}}\left|\mathbb{P}\left(n^{1 / 2}\left\Vert\hat{\mathbf D}^{-1/2}(\hat{\boldsymbol{\theta}}-\boldsymbol{\theta})\right\Vert_{\infty} \leqslant t\right)-\mathbb{P}\left(\Vert \boldsymbol G\Vert_{\infty} \leqslant t\right)\right|+\sup _{t \in \mathbb{R}}\left|\mathbb{P}\left(\left\Vert \boldsymbol G\right\Vert_{\infty} \leqslant t\right)-\mathbb{P}\left(\Vert \boldsymbol Z\Vert_{\infty} \leqslant t\right)\right|\\
        & \rightarrow 0 .
    \end{aligned}
$$
The last step holds from Lemma \ref{lemma2} and \ref{lemma3}.
\end{proof}
\subsubsection{Proof of Theorem \ref{thm:max_dist}(Exact limit distribution of maxima)}
\begin{proof}
    According to the Theorem 2 in \cite{feng2022asymptotic}, we have  $$
    \mathbb P(p\zeta_1^2\max _{1 \leq i \leq p} Z_i^2-2 \log p+\log \log p\leq x)\rightarrow F(x)=\exp \left\{-\frac{1}{\sqrt{\pi}} e^{-x / 2}\right\},
    $$
 a cdf of the Gumbel distribution, as $p \rightarrow \infty$.
 Thus, according to Lemma \ref{lemma:consistency_zeta} and Theorem \ref{thm1}.
 $$
 \begin{aligned}
     &\left| \mathbb P(T_{MAX}-2 \log p+\log \log p\leq x)-F(x)\right|\\
     \leq &\left| \mathbb P(\zeta_1^2\hat{\zeta}_1^{-2}T_{MAX}-2 \log p+\log \log p\leq x)-F(x)\right|+o(1)\\
     \leq &\left| \mathbb P(\zeta_1^2\hat{\zeta}_1^{-2}T_{MAX}-2 \log p+\log \log p\leq x)-\mathbb P(p\zeta_1^2\max _{1 \leq i \leq p} Z_i^2-2 \log p+\log \log p\leq x)\right|\\
     &+\left| \mathbb P(p\zeta_1^2\max _{1 \leq i \leq p} Z_i^2-2 \log p+\log \log p\leq x)-F(x)\right|+o(1)\rightarrow 0,
 \end{aligned}
 $$
 for any $x\in \mathbb R$.
\end{proof}
\subsubsection{Proof of Theorem \ref{thm:consistency}(Consistency for max-type test)}
\begin{proof}
   Recall that $u_p(y)=y+2\log p-\log \log p$, $T_{M A X}=n\Vert\hat{\mathbf{D}}^{-1 / 2} \hat{\boldsymbol{\theta}}\Vert_{\infty}^2  \hat{\zeta}_{1}^2p\cdot \left(1-n^{-1 / 2}\right)$. In order to make the following proof process briefly, we abbreviate $u_p(0)$ to $u_p$, define $\tilde q_{1-\alpha}=(\max\{q_{1-\alpha}+u_p,0\})^{1/2}=O_p(\{\log p-2\log\log (1-\alpha)^{-1}\}^{1/2})$, $T=T_{MAX}^{1/2}=n^{1/2}\Vert \hat{\mathbf D}^{-1/2}\hat{\boldsymbol\theta}\Vert_{\infty}\cdot \hat\zeta_1 p^{1/2}(1-n^{-1/2})^{1/2}$ and  $T^c=n^{1/2}\Vert \hat{\mathbf D}^{-1/2}(\hat{\boldsymbol\theta}-\boldsymbol\theta)\Vert_{\infty}\cdot \hat\zeta_1 p^{1/2}(1-n^{-1/2})^{1/2}$, which has the same distribution of $T$ under $H_0$.

   It is clear that, $T\geq n^{1/2}\Vert \hat{\mathbf D}^{-1/2}\boldsymbol \theta\Vert_\infty\cdot \hat\zeta_1 p^{1/2}(1-n^{-1/2})^{1/2}-T^c$. Combined with Assumption \ref{max2} and Lemma \ref{lemma:D_rate}, we get
 $$
   \begin{aligned}
       &\mathbb P(T_{MAX}-u_p\geq q_{1-\alpha}\mid H_1)\\
       \geq & \mathbb P(n^{1/2}\Vert \hat{\mathbf D}^{-1/2}\boldsymbol \theta\Vert_\infty\cdot \hat\zeta_1 p^{1/2}(1-n^{-1/2})^{1/2}-T^c\geq \tilde q_{1-\alpha}\mid H_1)\\
      =&\mathbb P(T^c\leq n^{1/2}\Vert \hat{\mathbf D}^{-1/2}\boldsymbol \theta\Vert_\infty\cdot \hat\zeta_1 p^{1/2}(1-n^{-1/2})^{1/2}-\tilde q_{1-\alpha}\mid H_1)\\
      \geq & \mathbb P(T^c\leq n^{1/2}\left(\Vert {\mathbf D}^{-1/2}\boldsymbol \theta\Vert_\infty- \Vert (\hat{\mathbf D}^{-1/2}-{\mathbf D}^{-1/2})\boldsymbol \theta\Vert_\infty\right)\cdot \hat\zeta_1 p^{1/2}(1-n^{-1/2})^{1/2}-\tilde q_{1-\alpha}\mid H_1)\\
      \geq &\mathbb P(T^c\leq n^{1/2}\Vert {\mathbf D}^{-1/2}\boldsymbol \theta\Vert_\infty\cdot (1+O_p(n^{1/2}\log ^{1/2}(np)))\cdot \hat\zeta_1 p^{1/2}(1-n^{-1/2})^{1/2}-\tilde q_{1-\alpha}\mid H_1)\rightarrow 1,
   \end{aligned}
   $$
if $\Vert\boldsymbol \theta\Vert_\infty\geq \tilde C n^{-1/2} \{\log p-2\log\log (1-\alpha)^{-1}\}^{1/2}$ for some large enough constant $\tilde C$.

The last inequality holds since
$$
\begin{aligned}
    \Vert (\hat{\mathbf D}^{-1/2}-{\mathbf D}^{-1/2})\boldsymbol \theta\Vert_\infty&=\max_{i=1,2,\cdots,p}\frac{\hat d_i-d_i}{\hat d_i d_i}\theta_i\leq \max_{i=1,2,\cdots,p}\vert 1-\frac{d_i}{\hat d_i}\vert\cdot \Vert {\mathbf D}^{-1/2}\boldsymbol \theta\Vert_\infty\\
&\leq O_p(n^{-1/2}\log^{1/2} (np))\Vert {\mathbf D}^{-1/2}\boldsymbol \theta\Vert_\infty.
\end{aligned}
$$

\end{proof}
\subsubsection{Proof of Theorem \ref{lemma5}(Without elliptical distribution)}
Theorem \ref{lemma_sum} is the special case of Theorem \ref{lemma5} with $\boldsymbol\theta=0$, so we only need to show that Theorem \ref{lemma5} holds.

\begin{proof}
    The following proof is based on the idea of the proof in article \cite{feng2016}, with modifications in some equations. We restate the equations in \cite{feng2016} on $U(\hat{\mathbf D}_{ij}^{-1/2}\boldsymbol X_i)^\top U(\hat{\mathbf D}_{ij}^{-1/2}\boldsymbol X_i)$.
    By the definition, we have
    $$
\begin{aligned}
& \frac{2}{n(n-1)} \sum_{1\leq i<j\leq n}  U\left(\hat{\mathbf{D}}_{i j}^{-1 / 2} \boldsymbol{X}_i\right)^\top U\left(\hat{\mathbf{D}}_{i j}^{-1 / 2} \boldsymbol{X}_j\right) \\
= & \frac{2}{n(n-1)} \sum_{1\leq i<j\leq n} \left(\boldsymbol{U}_i+R_i^{-1} \mathbf{D}^{-1 / 2} \boldsymbol{\theta}+\left(\hat{\mathbf{D}}_{i j}^{-1 / 2} \mathbf{D}^{1 / 2}-\mathbf{I}_{p}\right) \boldsymbol{U}_i\right)^\top \\
& \times\left(\boldsymbol{U}_j+R_j^{-1} \mathbf{D}^{-1 / 2} \boldsymbol{\theta}+\left(\hat{\mathbf{D}}_{i j}^{-1 / 2} \mathbf{D}^{1 / 2}-\mathbf{I}_{p}\right) \boldsymbol{U}_j\right)\left(1+\alpha_{i j}\right)^{-1 / 2}\left(1+\alpha_{j i}\right)^{-1 / 2} \\
= & \frac{2}{n(n-1)} \sum_{1\leq i<j\leq n}  \boldsymbol{U}_i^\top \boldsymbol{U}_j+\frac{2}{n(n-1)} \sum_{1\leq i<j\leq n}  R_i^{-1} R_j^{-1} \boldsymbol{\theta}^\top \mathbf{D}^{-1} \boldsymbol{\theta} \\
& +\frac{2}{n(n-1)} \sum_{1\leq i<j\leq n} \boldsymbol{U}_i^\top \boldsymbol{U}_j\left[\left(1+\alpha_{i j}\right)^{-1 / 2}\left(1+\alpha_{j i}\right)^{-1 / 2}-1\right] \\
& +\frac{4}{n(n-1)} \sum_{1\leq i<j\leq n}  \boldsymbol{U}_i^\top\left(\hat{\mathbf{D}}_{i j}^{-1 / 2} \mathbf{D}^{1 / 2}-\mathbf{I}_{p}\right) \boldsymbol{U}_j\left(1+\alpha_{i j}\right)^{-1 / 2}\left(1+\alpha_{j i}\right)^{-1 / 2} \\
& +\frac{2}{n(n-1)} \sum_{1\leq i<j\leq n}  \boldsymbol{U}_i^\top\left(\hat{\mathbf{D}}_{i j}^{-1 / 2} \mathbf{D}^{1 / 2}-\mathbf{I}_{p}\right)^2 \boldsymbol{U}_j\left(1+\alpha_{i j}\right)^{-1 / 2}\left(1+\alpha_{j i}\right)^{-1 / 2} \\
& +\frac{2}{n(n-1)} \sum_{1\leq i<j\leq n}  R_j^{-1} \boldsymbol{\theta}^\top \mathbf{D}^{-1} \boldsymbol{\theta}\left[\left(1+\alpha_{i j}\right)^{-1 / 2}\left(1+\alpha_{j i}\right)^{-1 / 2}-1\right] \\
& +\frac{4}{n(n-1)} \sum_{1\leq i<j\leq n}  \boldsymbol{U}_i^\top \mathbf{D}^{-1 / 2} \boldsymbol{\theta}\left(1+\alpha_{i j}\right)^{-1 / 2}\left(1+\alpha_{j i}\right)^{-1 / 2} \\
& +\frac{2}{n(n-1)} \sum_{1\leq i<j\leq n}  \boldsymbol{U}_i^\top\left(\hat{\mathbf{D}}_{i j}^{-1 / 2} \mathbf{D}^{1 / 2}-\mathbf{I}_{p}\right) \mathbf{D}^{-1 / 2} \boldsymbol{\theta}\left(1+\alpha_{i j}\right)^{-1 / 2}\left(1+\alpha_{j i}\right)^{-1 / 2} \\
:= & \frac{2}{n(n-1)} \sum_{1\leq i<j\leq n}  \boldsymbol{U}_i^\top \boldsymbol{U}_j+\frac{2}{n(n-1)} \sum_{1\leq i<j\leq n}  R_i^{-1} R_j^{-1} \boldsymbol{\theta}^\top \mathbf{D}^{-1} \boldsymbol{\theta} \\
& +A_{n 1}+A_{n 2}+A_{n 3}+A_{n 4}+A_{n 5}+A_{n 6}.
\end{aligned}
  $$
    where
    $$
    \alpha_{i j}=2 R_i^{-1} \boldsymbol{U}_i^\top \left(\hat{\mathbf{D}}_{i j}^{-1 / 2}-\mathbf{D}^{-1 / 2}\right) \boldsymbol{X}_i+R_i^{-2}\left\|\left(\hat{\mathbf{D}}_{i j}^{-1 / 2}-\mathbf{D}^{-1 / 2}\right) \boldsymbol{X}_i\right\|^2+\left.2 R_i^{-1} \mathbf{D}^{-1 / 2} \boldsymbol{\theta}+R_i^{-2} \boldsymbol{\theta}^\top \mathbf{D}^{-1} \boldsymbol{\theta}\right)^{-1 / 2}.
    $$
 Note that $R_i^{-1} \boldsymbol{U}_i^\top\left(\hat{\mathbf{D}}_{i j}^{-1 / 2}-\mathbf{D}^{-1 / 2}\right) \boldsymbol{X}_i=$ $\boldsymbol{U}_i^\top\left(\hat{\mathbf{D}}_{i j}^{-1 / 2} \mathbf{D}^{1 / 2}-\mathbf{I}_{p}\right) \boldsymbol{U}_i+R_i^{-1} \boldsymbol{U}_i^\top\left(\hat{\mathbf{D}}_{i j}^{-1 / 2}-\mathbf{D}^{-1 / 2}\right) \boldsymbol{\theta}=O_p\left(n^{-1 / 2}\left(\log p\right)^{1 / 2}\right)$ and $R_i^{-2}\left\|\left(\hat{\mathbf{D}}_{i j}^{-1 / 2}-\mathbf{D}^{-1 / 2}\right) \boldsymbol{X}_i\right\|^2=O_p\left(n^{-1} \log p\right)$ by Lemma \ref{lemma:D_rate}. By Assumption \ref{sum2} and  Equation \ref{H1_dense}, $R_i^{-1} \mathbf{D}^{-1 / 2} \boldsymbol{\theta}=O_p\left(\sigma_n^{1 / 2}\right)=O_p\left(n^{-1}\right)$ and $R_i^{-2} \boldsymbol{\theta}^\top \mathbf{D}^{-1} \boldsymbol{\theta}=$ $O_p\left(\sigma_n\right)=O_p\left(n^{-2}\right)$ where $\sigma_n^2=\frac{2}{n(n-1) p^2} \operatorname{tr}\left(\mathbf{R}^2\right)$. So $\alpha_{i j}=O_p\left(n^{-1 / 2}\left(\log p\right)^{1 / 2}\right)$.

Similarly, we will show that $A_{n 1}=o_p\left(\sigma_n\right)$. Under some calculations, we get $\mathbb E [(\boldsymbol U_i^\top \boldsymbol U_j)^2]=\text{tr}(\mathbf \Sigma^2_u) $. By Lemma \ref{LemmaA4}, we find that $\mathbf\Sigma_{u,i,j}=p^{-1}\sigma_{i,j}+O(p^{-1-\delta/2})$. Thus we have,
\begin{equation}\label{eq:U_iU_j^2}
    \begin{aligned}
\text{tr}(\mathbf\Sigma^2_u)&=\sum_{i=1}^p\sum_{j=1}^p\mathbf\Sigma_{u,i,j}^2=\sum_{i=1}^p\sum_{j=1}^p\left( p^{-2}\sigma^2_{i,j}+\sigma_{i,j}O(p^{-2-\delta/2})\right)\\
&=\sum_{i=1}^p\sum_{j=1}^p p^{-2}\sigma^2_{i,j}+
\sum_{p^{-\delta/2}=O(\sigma_{ij})} \sigma_{i,j}O(p^{-2-\delta/2})\\
&+\sum_{\sigma_{ij}\in[C_1 \frac{O(\text{tr}(\mathbf R^2))}{p^{2-\delta/2}},C_2 p^{-\delta/2}]}\sigma_{i,j}O(p^{-2-\delta/2})+\sum_{\sigma_{ij}=O(\frac{O(\text{tr}(\mathbf R^2))}{p^{2-\delta/2}})}\sigma_{i,j}O(p^{-2-\delta/2})\\
&=p^{-2}\text{tr}(\mathbf R^2)(1+O(1))+O(p^{-2-\delta/2})\cdot \frac{p^{2-\delta/2}}{O(\text{tr}(\mathbf R^2))}\cdot o(\frac{p^2}{n})+p^{-2}O(\text{tr}(\mathbf R^2))\\
&=O(p^{-2}\text{tr}(\mathbf R^2)).
    \end{aligned}
\end{equation}
 By the Cauchy inequality,
$$
\begin{aligned}
E\left(A_{n 1}^2\right) & =O\left(n^{-4}\right) \sum_{i<j} E\left\{\boldsymbol{U}_i^T \boldsymbol{U}_j\left[\left(1+\alpha_{i j}\right)^{-1 / 2}\left(1+\alpha_{j i}\right)^{-1 / 2}-1\right]\right\}^2 \\
& \leq O\left(n^{-2}\right) E\left(\boldsymbol{U}_i^T \boldsymbol{U}_j\right)^2 E\left[\left(1+\alpha_{i j}\right)^{-1 / 2}\left(1+\alpha_{j i}\right)^{-1 / 2}-1\right]^2 \\
& =O\left(n^{-3} \log p\left( p^{-2}\text{tr}(\mathbf R^2)+O(p^{-2-\delta/2})(p+ po(\frac{p}{n^{1/2}})) \right)\right)=o\left(\sigma_n^2\right) .
\end{aligned}
$$
$$
\begin{aligned}
A_{n 2}  &=\frac{4}{n(n-1)} \sum_{i<j} \boldsymbol{U}_i^\top\left(\hat{\mathbf{D}}_{i j}^{-1 / 2} \mathbf{D}^{1 / 2}-\mathbf{I}_{p}\right) \boldsymbol{U}_j \\
& \quad+\frac{4}{n(n-1)} \sum_{i<j} \boldsymbol U_i^\top \left(\hat{\mathbf{D}}_{i j}^{-1 / 2} \mathbf{D}^{1 / 2}-\mathbf{I}_{p}\right) \boldsymbol{U}_j\left[\left(1+\alpha_{i j}\right)^{-1 / 2}\left(1+\alpha_{j i}\right)^{-1 / 2}-1\right] \\
& := G_{n 1}+G_{n 2} .
\end{aligned}
$$
By Lemma \ref{lemma:D_rate} and Equation \ref{eq:U_iU_j^2}, $\mathbb E((\boldsymbol U_i^\top \left(\hat{\mathbf{D}}_{i j}^{-1 / 2} \mathbf{D}^{1 / 2}-\mathbf{I}_{p}\right) \boldsymbol{U}_j)^2)\leq O(n^{-1}\log p \text{ tr}(\mathbf\Sigma_u^2))=o(p^{-2}\text{tr}(\mathbf R^2))$. Then we obtain $G_{n1}=o_p(\sigma_n)$. Similar to $A_{n1}$, we can show $G_{n2}=o_p(\sigma_n)$. Taking the same procedure as $A_{n2}$, we can obtain $A_{n3}=o_p(\sigma_n)$. Similar to the processing of Equation \ref{eq:expansion_of_U_i}, we get
$$
\frac{2}{n(n-1)}\sum_{1\leq i<j
\leq n} \boldsymbol U_i^\top \boldsymbol U_j=\frac{2}{n(n-1)}\sum_{1\leq i<j \leq n}(\mathbf D^{-1/2}\mathbf \Gamma U(\boldsymbol W_i))^\top \mathbf D^{-1/2}\mathbf \Gamma U(\boldsymbol W_j)+o_p(\sigma_n).
$$
We replace the Lemma 1 in \cite{feng2016} by Lemma \ref{LemmaA4}, and final acquire
$$
\sqrt{\frac{n(n-1) p^2}{2 \operatorname{tr}\left(\mathbf{R}^2\right)}} \frac{2}{n(n-1)}\sum_{1\leq i<j \leq n}(\mathbf D^{-1/2}\mathbf \Gamma U(\boldsymbol W_i))^\top \mathbf D^{-1/2}\mathbf \Gamma U(\boldsymbol W_j) \xrightarrow{d} N(0,1).
$$
\end{proof}
\subsubsection{Proof of Theorem \ref{thm3}(Asymptotically independent under $H_0$)}

\begin{proof}
To prove $T_{SUM}$ and $T_{MAX}$ are  asymptotically independent, it suffices to show that: Under $H_0$,
\begin{equation}\label{eq:indepent1}
\mathbb P(\frac{T_{SUM}}{\sigma_n}\leq  x,T_{MAX} -2 \log p+\log \log p \leq y)\\
\rightarrow\Phi(x)\cdot \exp \left\{-\frac{1}{\sqrt{\pi}} e^{-y / 2}\right\}.
\end{equation}
Let $u_p(y)=y+2\log p-\log \log p$, and we rewrite Equation \ref{eq:indepent1} as
\begin{equation}\label{eq:indepent2}
    \mathbb P(\frac{T_{SUM}}{\sigma_n}\leq  x,T_{MAX}\leq u_p(y))
\rightarrow\Phi(x)\cdot \exp \left\{-\frac{1}{\sqrt{\pi}} e^{-y / 2}\right\}.
\end{equation}
From the proof of Theorem 2 in \cite{feng2016}, we acquire
\begin{equation}\label{eq:Tsum1}
   T_{SUM}=\frac{2}{n(n-1)}\sum\sum_{i<j} \boldsymbol U_i^\top \boldsymbol U_j+o_p(\sigma_n) ,
\end{equation}
and it's easy to find that $\sigma_n^2=\frac{2}{n(n-1)p}+o(\frac{1}{n^3})$ according to Assumption \ref{sum2}. Combined with Lemma \ref{lemma1}, it suffice to show,
\begin{equation}\label{eq:independent3}
    \begin{aligned}
&\mathbb P(\frac{\frac{2}{n(n-1)}\sum\sum_{i<j} \boldsymbol U_i^\top \boldsymbol U_j}{\sigma_n}+o_p(1)\leq  x\\
&\qquad,p\left\|n^{-1/2}\sum_{i=1}^n \boldsymbol U_i\right\|_{\infty}^2+O_p(L_{n,p}) \leq u_p(y))\\
&\rightarrow\Phi(x)\cdot \exp \left\{-\frac{1}{\sqrt{\pi}} e^{-y / 2}\right\}.
\end{aligned}
\end{equation}
We next prove that,
\begin{equation}\label{eq:independent_final}
\begin{aligned}
&\mathbb P(\sqrt{\frac{n}{n-1}}  \left( \frac{\Vert\sqrt\frac{p}{n}\sum_{i=1}^n\boldsymbol U_i\Vert_2^2-p}{\sqrt{2\text{tr}(\mathbf R^2)}}\right)\leq  x,\left\|\sqrt\frac{p}{n}\sum_{i=1}^n \boldsymbol U_i\right\|_{\infty}^2 \leq u_p(y))\\
&\rightarrow\Phi(x)\cdot \exp \left\{-\frac{1}{\sqrt{\pi}} e^{-y / 2}\right\}.
\end{aligned}
\end{equation}
When Equation \ref{eq:independent_final} holds, combined with $O_p(L_{n,p})=o_p(1)$, Equation \ref{eq:independent3} holds obviously, which means that the independence of $T_{SUM}$ and $T_{MAX}$ follows.

\emph{Proof of Equation \ref{eq:independent_final}:}
From the Theorem 2 in \cite{feng2022asymptotic}, the Equation \ref{eq:independent_final} holds if $\boldsymbol U_i$ follows the normal distribution. We then investigate the non-normal case.
 Let $ \boldsymbol\xi_i=\boldsymbol U_i\in R^p,i=1,2,\cdots,n$. For $\boldsymbol z=(z_1,\cdots,z_q)^\top\in R^q$, we consider a smooth approximation of the maximum function, namely,
$$
F_\beta(\boldsymbol z):=\beta^{-1}\log (\sum_{j=1}^q\exp(\beta z_j)),
$$
where $\beta>0$ is the smoothing parameter that controls the level of approximation. An elementary calculation shows that for all $z\in R^q$,
$$
0\leq F_\beta (\boldsymbol z)-\max_{1\leq j\leq q}z_j\leq \beta^{-1}\log q.
$$
Define $\sigma_S^2=2n^2\text{ tr }(R^2)$,
$$
\begin{aligned}
W(\boldsymbol x_1,\cdots,\boldsymbol x_n)
&=\frac{\Vert\sqrt\frac{p}{n}\sum_{i=1}^n\boldsymbol x_i\Vert_2^2-p}{\sqrt{2\text{tr}(\mathbf R^2)}}\\
&=\frac{ p\sum_{i\not=j}\boldsymbol x_i^\top \boldsymbol x_j}{\sqrt{2n^2\text{ tr }(\mathbf R^2)}}:=\frac{ p\sum_{i\not=j}\boldsymbol x_i^\top \boldsymbol x_j}{ \sigma_S},\\
V(\boldsymbol x_1,\cdots,\boldsymbol x_n)&=\beta^{-1}\log (\sum_{j=1}^p\exp(\beta\sqrt{\frac{p}{n}}\sum_{i=1}^n \boldsymbol x_{i,j})).
\end{aligned}
$$

By setting $\beta=n^{1/8}\log(n)$, Equation \ref{eq:independent_final} is equivalent to
\begin{equation}
    P(W(\boldsymbol\xi_1,\cdots,\boldsymbol\xi_p)\leq x,V(\boldsymbol\xi_1,\cdots,\boldsymbol\xi_p)\leq u_p(y))\rightarrow\Phi(x)\cdot \exp(-\exp(y)).
\end{equation}

Suppose $\{\boldsymbol Y_1,\boldsymbol Y_2,\cdots,\boldsymbol Y_n\}$ are sample from $N(0,\mathbb E\boldsymbol U_1^\top \boldsymbol U_1)$ , and independent with $\boldsymbol U_1,\cdots,\boldsymbol U_n$(or write as $\boldsymbol \xi_1,\cdots,\boldsymbol \xi_n)$.  The key idea is to show that: $(W(\boldsymbol\xi_1,\cdots,\boldsymbol\xi_n),V(\boldsymbol\xi_1,\cdots,\boldsymbol\xi_n))$ has the same limiting distribution as $(W(\boldsymbol Y_1,\cdots,\boldsymbol Y_n),V(\boldsymbol Y_1,\cdots,\boldsymbol Y_n))$.

Let $l^2_b(\mathbb R)$ denote the class of bounded functions with bounded and continuous derivatives up to order 3.It is known that a sequence of randon variables $\{Z_n\}_{n=1}^\infty$ converges weakly to a random variable $Z$ if and only if for every $f\in l^3_b(\mathbb R)$, $\mathbb E(f(Z_n))\rightarrow \mathbb E(f(Z))$.

It suffices to show that:
$$
\mathbb E\{f(W(\boldsymbol\xi_1,\cdots,\boldsymbol\xi_n),V(\boldsymbol\xi_1,\cdots,\boldsymbol\xi_n))\}-\mathbb E\{f(W(\boldsymbol Y_1,\cdots,\boldsymbol Y_n),V(\boldsymbol Y_1,\cdots,\boldsymbol Y_n))\}\rightarrow 0,
$$
for every $f\in l_b^3(\mathbb R^2)$ as $(n,p)\rightarrow \infty$.

We introduce $\tilde { W}_d=W(\boldsymbol\xi_1,\cdots,\boldsymbol\xi_{d-1},\boldsymbol Y_d,\cdots,\boldsymbol Y_n)$  and $\tilde{ V}_d=V(\boldsymbol\xi_1,\cdots,\boldsymbol\xi_{d-1},\boldsymbol Y_d,\cdots,\boldsymbol Y_n)$ for $d=1,\cdots,n+1$, $\mathcal F_d=\sigma\{\boldsymbol\xi_1,\cdots,\boldsymbol\xi_{d-1},\boldsymbol Y_{d+1},\cdots,\boldsymbol Y_n\}$ for $d=1,\cdots,n$. If there is no danger of confusion, we simply write $\tilde { W}_d$ and $\tilde { V}_d$ as $ W_d$ and $ V_d$ respectively (only for this part). Then,
$$
\begin{aligned}
   &\mid\mathbb E\{f(W(\boldsymbol\xi_1,\cdots,\boldsymbol\xi_n),V(\boldsymbol\xi_1,\cdots,\boldsymbol\xi_n))\}-\mathbb E\{f(W(\boldsymbol Y_1,\cdots,\boldsymbol Y_n),V(\boldsymbol Y_1,\cdots,\boldsymbol Y_n))\}\mid\\
\leq &\sum_{d=1}^n  \mid \mathbb E\{f(W_d,V_d)-\mathbb E\{f(W_{d+1},V_{d+1}) \}\mid.
\end{aligned}
$$
Let
$$
\begin{aligned}
W_{d,0}&=\frac{2p\sum_{i<j}^{d-1}\boldsymbol\xi_i^\top \boldsymbol\xi_j+2p\sum_{d+1\leq i<j\leq n}\boldsymbol Y_i^\top \boldsymbol Y_j+2p\sum_{i=1}^{d-1}\sum_{j=d+1}^n\boldsymbol\xi_i^\top \boldsymbol Y_j}{\sigma_S}\in\mathcal F_d,\\
V_{d,0}&=\beta^{-1}\log (\sum_{j=1}^p\exp(\beta\sqrt{\frac{p}{n}}\sum_{i=1}^{d-1} \xi_{i,j}+\beta\sqrt{\frac{p}{n}}\sum_{i=d+1}^{n}  Y_{i,j}))\in \mathcal F_d.\\
\end{aligned}
$$
By Taylor expansion, we have,
$$
\begin{aligned}
f\left(W_d, V_d\right)-f\left(W_{d, 0}, V_{d, 0}\right)= & f_1\left(W_{d, 0}, V_{d, 0}\right)\left(W_d-W_{d, 0}\right)+f_2\left(W_{d, 0}, V_{d, 0}\right)\left(V_d-V_{d, 0}\right) \\
& +\frac{1}{2} f_{11}\left(W_{d, 0}, V_{d, 0}\right)\left(W_d-W_{d, 0}\right)^2+\frac{1}{2} f_{22}\left(W_{d, 0}, V_{d, 0}\right)\left(V_d-V_{d, 0}\right)^2 \\
& +\frac{1}{2} f_{12}\left(W_{d, 0}, V_{d, 0}\right)\left(W_d-W_{d, 0}\right)\left(V_d-V_{d, 0}\right) \\
& +O\left(\left|\left(V_d-V_{d, 0}\right)\right|^3\right)+O\left(\left|\left(W_d-W_{d, 0}\right)\right|^3\right),
\end{aligned}
$$
and
$$
\begin{aligned}
f\left(W_{d+1}, V_{d+1}\right)-f\left(W_{d, 0}, V_{d, 0}\right)= & f_1\left(W_{d, 0}, V_{d, 0}\right)\left(W_{d+1}-W_{d, 0}\right)+f_2\left(W_{d, 0}, V_{d, 0}\right)\left(V_{d+1}-V_{d, 0}\right) \\
& +\frac{1}{2} f_{11}\left(W_{d, 0}, V_{d, 0}\right)\left(W_{d+1}-W_{d, 0}\right)^2+\frac{1}{2} f_{22}\left(W_{d, 0}, V_{d, 0}\right)\left(V_{d+1}-V_{d, 0}\right)^2 \\
& +\frac{1}{2} f_{12}\left(W_{d, 0}, V_{d, 0}\right)\left(W_{d+1}-W_{d, 0}\right)\left(V_{d+1}-V_{d, 0}\right) \\
& +O\left(\left|\left(V_{d+1}-V_{d, 0}\right)\right|^3\right)+O\left(\left|\left(W_{d+1}-W_{d, 0}\right)\right|^3\right),
\end{aligned}
$$
where for $f:=f(x, y), f_1(x, y)=\frac{\partial f}{\partial x}, f_2(x, y)=\frac{\partial f}{\partial y}, f_{11}(x, y)=\frac{\partial f^2}{\partial^2 x}, f_{22}(x, y)=\frac{\partial f^2}{\partial^2 y}$ and $f_{12}(x, y)=\frac{\partial f^2}{\partial x \partial y}$.

We first consider $W_d,W_{d+1},W_{d,0}$ and notice that,
$$
\begin{aligned}
  W_d-W_{d,0}&=\frac{p\sum_{i=1}^{d-1}\boldsymbol\xi_i^\top \boldsymbol Y_d+p\sum_{i=d+1}^n \boldsymbol Y_i^\top \boldsymbol Y_d}{\sigma_S},\\
W_{d+1}-W_{d,0}&=\frac{p\sum_{i=1}^{d-1}\boldsymbol\xi_i^\top \boldsymbol\xi_d+p\sum_{i=d+1}^n \boldsymbol Y_i^\top \boldsymbol\xi_d}{\sigma_S}.\\
\end{aligned}
$$
Due to $\mathbb{E}\left(\boldsymbol\xi_t\right)=\mathbb{E}\left(\boldsymbol Y_t\right)=0$ and $\mathbb{E}\left(\boldsymbol\xi_t \boldsymbol\xi_t^{\top}\right)=\mathbb{E}\left(\boldsymbol Y_t
\boldsymbol Y_t^{\top}\right)$, it can be verified that,
$$
\mathbb{E}\left(W_d-W_{d, 0} \mid \mathcal{F}_d\right)=\mathbb{E}\left(W_{d+1}-W_{d, 0} \mid \mathcal{F}_d\right) \text { and } \mathbb{E}\left(\left(W_d-W_{d, 0}\right)^2 \mid \mathcal{F}_d\right)=\mathbb{E}\left(\left(W_{d+1}-W_{d, 0}\right)^2 \mid \mathcal{F}_d\right).
$$
Hence,
$$
\begin{aligned}
\mathbb{E}\left\{f_1\left(W_{d, 0}, V_{d, 0}\right)\left(W_d-W_{d, 0}\right)\right\} & =\mathbb{E}\left\{f_1\left(W_{d, 0}, V_{d, 0}\right)\left(W_{d+1}-W_{d, 0}\right)\right\} \text { and } \\
\mathbb{E}\left\{f_{11}\left(W_{d, 0}, V_{d, 0}\right)\left(W_d-W_{d, 0}\right)^2\right\} & =\mathbb{E}\left\{f_{11}\left(W_{d, 0}, V_{d, 0}\right)\left(W_{d+1}-W_{d, 0}\right)^2\right\}.
\end{aligned}
$$
Next we consider $V_d-V_{d,0}$. Let
$z_{d,0,j}=\sqrt{\frac{p}{n}}\sum_{i=1}^{d-1} \xi_{i,j}+\sqrt{\frac{p}{n}}\sum_{i=d+1}^{n}  Y_{i,j}, z_{d,j}=z_{d,0,j}+n^{-1/2}\sqrt p Y_{d,j},
z_{d+1,j}=z_{d,0,j}+n^{-1/2}\sqrt p\xi_{d,j}$. By Taylor expansion, we have that:
\begin{equation}\label{eq:expansion_V}
\begin{aligned}
V_d-V_{d, 0}= & \sum_{l=1}^{n } \partial_l F_\beta\left(\boldsymbol{z}_{d, 0}\right)\left( z_{d, l}- z_{d, 0, l}\right)+\frac{1}{2} \sum_{l=1}^{n } \sum_{k=1}^{n } \partial_k \partial_l F_\beta\left(\boldsymbol{z}_{d, 0}\right)\left( z_{d, l}- z_{d, 0, l}\right)\left( z_{d, k}- z_{d, 0, k}\right) \\
& +\frac{1}{6} \sum_{l=1}^{n} \sum_{k=1}^{n } \sum_{v=1}^{n } \partial_v \partial_k \partial_l F_\beta\left(\boldsymbol{z}_{d, 0}+\tilde\delta\left(\boldsymbol{z}_d-\boldsymbol{z}_{d, 0}\right)\right)\left( z_{d, l}- z_{d, 0, l}\right)\left( z_{d, k}-\boldsymbol z_{d, 0, k}\right)\left(\boldsymbol z_{d, v}-\boldsymbol z_{d, 0, v}\right),
\end{aligned}
\end{equation}
for some $\tilde\delta \in(0,1)$. Again, due to $\mathbb{E}\left(\boldsymbol\xi_t\right)=\mathbb{E}\left(\boldsymbol Y_t\right)=0$ and $\mathbb{E}\left(\boldsymbol \xi_t \boldsymbol\xi_t^{\top}\right)=\mathbb{E}\left(\boldsymbol Y_t \boldsymbol Y_t^{\top}\right)$, we can verify that
$$
\mathbb{E}\left\{\left( z_{d, l}- z_{d, 0, l}\right) \mid \mathcal{F}_d\right\}=\mathbb{E}\left\{\left( z_{d+1, l}- z_{d, 0, l}\right) \mid \mathcal{F}_d\right\} \text { and } \mathbb{E}\left\{\left( z_{d, l}- z_{d, 0, l}\right)^2 \mid \mathcal{F}_d\right\}=\mathbb{E}\left\{\left( z_{d+1, l}- z_{d, 0, l}\right)^2 \mid \mathcal{F}_d\right) \text {. }
$$
By Lemma A.2 in \cite{chernozhukov2013gaussian}, we have,
$$
\left|\sum_{l=1}^{n } \sum_{k=1}^{n } \sum_{v=1}^{n } \partial_v \partial_k \partial_l F_\beta\left(\boldsymbol{z}_{d, 0}+\tilde\delta\left(\boldsymbol{z}_d-\boldsymbol{z}_{d, 0}\right)\right)\right| \leq C \beta^2,
$$
for some positive constant $C$.

By Lemma \ref{LemmaA4}, we have that:
$\left\|\zeta_1^{-1}  U_{i, j}\right\|_{\psi_\alpha} \lesssim \bar{B},$ for all $i=1, \ldots, n$ and $j=1, \ldots, p$, which means $P(\mid \sqrt{p}\xi_{i,j}\mid\geq t)\leq 2\exp(-(ct\sqrt{p}/\zeta_1)^\alpha)\lesssim 2\exp(-(ct)^\alpha)$, $\mathrm{P}\left(\max _{1 \leq i \leq n}\left|\sqrt{p}\xi_{i j}\right|>C \log (n )\right) \rightarrow 0$ and since $\sqrt{p}Y_{t j} \sim N(0,1)$, $ \mathrm{P}\left(\max _{1 \leq i \leq n}\left|\sqrt{p}Y_{i j}\right|>C \log (n )\right) \rightarrow 0$. Hence,
$$
\begin{aligned}
&\left|\frac{1}{6} \sum_{l=1}^{n } \sum_{k=1}^{n } \sum_{v=1}^{n } \partial_v \partial_k \partial_l F_\beta\left(\boldsymbol{z}_{d, 0}+\tilde\delta\left(\boldsymbol{z}_d-\boldsymbol{z}_{d, 0}\right)\right)\left(\boldsymbol z_{d, l}-\boldsymbol z_{d, 0, l}\right)\left(\boldsymbol z_{d, k}-\boldsymbol z_{d, 0, k}\right)\left(\boldsymbol z_{d, v}-\boldsymbol z_{d, 0, v}\right)\right|\\
&\leq C\beta^2n^{-3/2}\log^3(n),\\
&\left|\frac{1}{6} \sum_{l=1}^{n } \sum_{k=1}^{n } \sum_{v=1}^{n } \partial_v \partial_k \partial_l F_\beta\left(\boldsymbol{z}_{d+1, 0}+\tilde\delta\left(\boldsymbol{z}_{d+1}-\boldsymbol{z}_{d, 0}\right)\right)\left(\boldsymbol z_{d+1, l}-\boldsymbol z_{d, 0, l}\right)\left(\boldsymbol z_{d+1, k}-\boldsymbol z_{d, 0, k}\right)\left(\boldsymbol z_{d+1, v}-\boldsymbol z_{d, 0, v}\right)\right|\\
&\leq C\beta^2n^{-3/2}\log^3(n),
\end{aligned}
$$
hold with probability approaching one. Consequently, we have that: with probability one,
$$
\left|\mathrm{E}\left\{f_2\left(W_{d, 0}, V_{d, 0}\right)\left(V_d-V_{d, 0}\right)\right\}-\mathrm{E}\left\{f_2\left(W_{d, 0}, V_{d, 0}\right)\left(V_{d+1}-V_{d, 0}\right)\right\}\right| \leq C \beta^2 n^{-3 / 2} \log ^3(n).
$$
Similarly, it can be verified that,
$$
\left|\mathrm{E}\left\{f_{22}\left(W_{d, 0}, V_{d, 0}\right)\left(V_d-V_{d, 0}\right)^2\right\}-\mathrm{E}\left\{f_{22}\left(W_{d, 0}, V_{d, 0}\right)\left(V_{d+1}-V_{d, 0}\right)^2\right\}\right| \leq C \beta^2 n^{-3 / 2} \log ^3(n),
$$
and
$$
\begin{aligned}
& \left|\mathrm{E}\left\{f_{12}\left(W_{d, 0}, V_{d, 0}\right)\left(W_d-W_{d, 0}\right)\left(V_d-V_{d, 0}\right)\right\}-\mathrm{E}\left\{f_{12}\left(W_{d, 0}, V_{d, 0}\right)\left(W_{d+1}-W_{d, 0}\right)\left(V_{d+1}-V_{d, 0}\right)\right\}\right| \\
& \quad \leq C \beta^2 n^{-3 / 2} \log ^3(n).
\end{aligned}
$$
By Equation \ref{eq:expansion_V}, $\mathrm{E}\left(\left|V_d-V_{d, 0}\right|^3\right)=O\left(n^{-3 / 2} \log ^3(n )\right)$. For $\mathbb{E}\left(\left(W_d-W_{d, 0}\right)^3\right)$, we first calculate $\mathbb{E}\left(\left(W_d-W_{d, 0}\right)^4\right)$, then it's easy to get the order for 3-order term.
\begin{equation}\label{eq:bin_exp0}
\begin{aligned}
\quad\mathbb{E}\left(\left(W_d-W_{d, 0}\right)^4\right)
&=\mathbb{E}\left(\frac{p\sum_{i=1}^{d-1}\boldsymbol\xi_i^\top \boldsymbol Y_d+p\sum_{i=d+1}^n \boldsymbol Y_i^\top \boldsymbol Y_d}{\sigma_S}\right)^4\\
&=\frac{p^4}{2n^4(\text{tr }(\mathbf R^2))^2}\mathbb{E}\left(\sum_{i=1}^{d-1}\boldsymbol \xi_i^\top \boldsymbol Y_d+\sum_{i=d+1}^n \boldsymbol Y_i^\top \boldsymbol Y_d\right)^4.
\end{aligned}
\end{equation}
We consider the binomial expansion term and calculate them separately in Equation \ref{eq:bin_exp0}:
\begin{equation}\label{eq:bin_exp}
\begin{aligned}
(i)=&\mathbb E(\sum_{i=d+1}^{n}\boldsymbol Y_i^\top \boldsymbol Y_d)^4,\
(ii)=\mathbb E\sum_{i=1}^{d-1}\boldsymbol\xi_i^\top \boldsymbol Y_d\cdot(\sum_{i=d+1}^n \boldsymbol Y_i^\top \boldsymbol Y_d)^3,\
(iii)=\mathbb E(\sum_{i=1}^{d-1}\boldsymbol\xi_i^\top \boldsymbol Y_d)^2\cdot(\sum_{i=d+1}^n \boldsymbol Y_i^\top \boldsymbol Y_d)^2,\\
(iv)=&\mathbb E(\sum_{i=1}^{d-1}\boldsymbol\xi_i^\top \boldsymbol Y_d)^3\cdot\sum_{i=d+1}^n \boldsymbol Y_i^\top \boldsymbol Y_d,\
(v)=\mathbb E(\sum_{i=1}^{d-1}\boldsymbol\xi_i^\top \boldsymbol Y_d)^4.\\
\end{aligned}
\end{equation}
Since $\mathbb E \boldsymbol Y_i=\mathbb E \boldsymbol\xi_i=0$, we easily find that Equation \ref{eq:bin_exp}
-(ii)(iv) equal to 0. Next we can get the following equations for Equation \ref{eq:bin_exp}-(iii) after some straightforward calculations.
\begin{equation}
\begin{aligned}
\mathbb E(\sum_{i=1}^{d-1}\boldsymbol\xi_i^\top \boldsymbol Y_d)^2\cdot(\sum_{i=d+1}^n \boldsymbol Y_i^\top \boldsymbol Y_d)^2
=&\mathbb  E[\mathbb E[(\sum_{i=1}^{d-1}\boldsymbol\xi_i^\top \boldsymbol Y_d)^2\cdot(\sum_{i=d+1}^n \boldsymbol Y_i^\top \boldsymbol Y_d)^2\mid \boldsymbol Y_d]]\\
=&\mathbb E[(d-1)(n-d)(\boldsymbol Y_d^\top\mathbf\Sigma_u \boldsymbol Y_d)^2]\\
=&\mathbb E[(d-1)(n-d)((\mathbf\Sigma_u^{-1/2}\boldsymbol Y_d)^\top  \mathbf\Sigma_u^2 (\mathbf\Sigma_u^{-1/2}\boldsymbol Y_d))^2]\\
=&(d-1)(n-d)\cdot 2\text{tr }(\mathbf\Sigma_u^4)\\
\leq&(d-1)(n-d)\cdot O(\text{tr}(\mathbf\Sigma_u^2)^2).\\
\end{aligned}
\end{equation}
By some properties for standard normal random variable, the last inequality holds with some simple calculations shown below.

(i)
\begin{equation}\label{eq:sigma4}
\begin{aligned}
\text{tr }(\mathbf\Sigma_u^4)&=\Vert\mathbf\Sigma_u^2\Vert_F^2=\Vert\mathbf\Sigma_u\cdot \mathbf\Sigma_u\Vert_F^2
\leq (\Vert\mathbf\Sigma_u\Vert_F\cdot\Vert \mathbf\Sigma_u\Vert_F)^2\\
&=\Vert\mathbf\Sigma_u\Vert_F^4=\text{tr }(\mathbf\Sigma_u^2)^2
.
\end{aligned}
\end{equation}

(ii)If $\boldsymbol X,\boldsymbol Y\stackrel{i.i.d.}{\sim} N(0,\mathbf I_p)$, then
\begin{equation}\label{eq:normal}
    \begin{aligned}
        \mathbb E(\boldsymbol X^\top \mathbf A\boldsymbol X)^2&=2\text{tr}(\mathbf A^2)-\text{tr}^2(\mathbf A)\leq \text{tr}(\mathbf A^2),\\
        \mathbb E(\boldsymbol X^\top \mathbf A\boldsymbol Y)^4&=\mathbb E[\mathbb E[(\boldsymbol Y^\top \mathbf A\boldsymbol X\boldsymbol X^\top \mathbf A\boldsymbol Y)^2\mid \boldsymbol X]]\leq 2\mathbb E[\text{tr}(\mathbf A\boldsymbol X\boldsymbol X^\top \mathbf A)^2]\\
        &=2\mathbb E[(\boldsymbol X^\top \mathbf A^2\boldsymbol X)^2]\leq 4\text{tr}(\mathbf A^4).
    \end{aligned}
\end{equation}

For Equation \ref{eq:bin_exp}-(1), according to $\sum_{i=d+1}^{n}\boldsymbol Y_i\sim N(0,(n-d)\mathbf\Sigma_u)$ and Equation \ref{eq:sigma4}-\ref{eq:normal}, we have,
\begin{equation}\label{eq:bin_exp1}
\begin{aligned}
\mathbb E(\sum_{i=d+1}^{n}\boldsymbol Y_i^\top \boldsymbol Y_d)^4
=& \mathbb E((\frac{1}{\sqrt{(n-d)}}\mathbf\Sigma_u^{-1/2}\sum_{i=d+1}^{n}\boldsymbol Y_i)^\top(\sqrt{n-d}\mathbf\Sigma_u)(\mathbf\Sigma_u^{-1/2}\boldsymbol Y_d))^4\\
\leq& \text{tr }((n-d)^2\mathbf\Sigma_u^4)
=(n-d)^2 O(\text{tr}(\mathbf\Sigma_u^2)^2).
\end{aligned}
\end{equation}

Similar to Equation \ref{eq:bin_exp1}, for Equation \ref{eq:bin_exp}-(v),
\begin{equation}
    \mathbb E(\sum_{i=1}^{d-1}\boldsymbol\xi_i^\top \boldsymbol Y_d)^4\leq (d-1)^2\text{tr }(\mathbf\Sigma_u^4)\leq (d-1)^2O(\text{tr}(\mathbf\Sigma_u^2)^2).
\end{equation}
Thus, in combining with the Equation \ref{eq:U_iU_j^2},
$$
\begin{aligned}
&\quad\mathbb{E}\left(\left(W_d-W_{d, 0}\right)^4\right)\\
&=\frac{p^4}{2n^4(\text{tr }(R^2))^2}\mathbb{E}\left(\sum_{i=1}^{d-1}\boldsymbol\xi_i^\top \boldsymbol Y_d+\sum_{i=d+1}^n \boldsymbol Y_i^\top \boldsymbol Y_d\right)^4\\
&\leq \frac{p^4}{2n^4(\text{tr }(\mathbf R^2))^2}\{(d-1)(n-d)+(n-d)^2+(d-1)^2\}O(\text{tr}(\mathbf\Sigma_u^2)^2)\\
&\leq \frac{p^4}{2n^4(\text{tr }(\mathbf R^2))^2}n^2O(\text{tr}(\mathbf\Sigma_u^2)^2)=O(\frac{1}{n^2}).
\end{aligned}
$$
By Jensen's inequality , we have
$$
\sum_{d=1}^n \mathrm{E}\left|W_d-W_{d, 0}\right|^3 \leq \sum_{d=1}^n\left(\mathrm{E}\left(W_d-W_{d, 0}\right)^4\right)^{3 / 4} \leq C^{\prime} n^{-1 / 2},
$$
for some positive constant $C^{\prime}$, Combining all facts together, we conclude that
$$
\sum_{d=1}^n\left|\mathrm{E}\left\{f\left(W_d, V_d\right)\right\}-\mathrm{E}\left\{f\left(W_{d+1}, V_{d+1}\right)\right\}\right| \leq C \beta^2 n^{-1 / 2} \log ^3 n+C^{\prime} n^{-1 / 2} \rightarrow 0,
$$
as $(n, p) \rightarrow \infty$. The conclusion follows.
\end{proof}
\subsubsection{Proof of Theorem \ref{thm4}(asymptotically independent under $H_{1,local}$)}
\begin{proof}
From the proof of Theorem 2 in \cite{feng2016}, we can find that
$$
T_{SUM}=\frac{2}{n(n-1)}\sum\sum_{i<j} \boldsymbol U_i^\top \boldsymbol U_j+\zeta_1^2 \boldsymbol \theta^\top \mathbf D^{-1} \boldsymbol \theta+o_p(\sigma_n),
$$
and according to Lemma \ref{lemma1} with minor modifications, we get the Bahadur representation in $L^\infty$ norm,
$$
 n^{1/2}{\mathbf D}^{-1/2}(\hat{\boldsymbol \theta}-\boldsymbol \theta)=n^{-1/2}\zeta_1^{-1}\sum_{i=1}^n (\boldsymbol U_i+\zeta_1{\mathbf D}^{-1/2}\boldsymbol \theta)+C_n.
$$

Similar to the proof in Theorem \ref{thm3}, it's suffice to show the result holds for normal version, i.e. it suffice to show that:
$$
\Vert \sqrt{\frac{p}{n}}\sum_{i=1}^n\boldsymbol Y_i\Vert^2\text{ and }\Vert \sqrt{\frac{p}{n}}\sum_{i=1}^n(\boldsymbol Y_i+\zeta_1{\mathbf D}^{-1/2}\boldsymbol \theta)\Vert_\infty^2,
$$
are asymptotic independent, where $\{\boldsymbol Y_1,\boldsymbol Y_2,\cdots,\boldsymbol Y_n\}$ are sample from $N(0,\mathbb E\boldsymbol U_1^\top \boldsymbol U_1)$.

Denote $\sqrt{\frac{p}{n}}\sum_{i=1}^n \boldsymbol Y_i:=\boldsymbol \varphi=(\varphi_1,\cdots,\varphi_p)^\top$,$\boldsymbol\varphi_{\mathcal A}=(\varphi_{j_1},\cdots,\varphi_{j_d})^\top$, and $\boldsymbol\varphi_{\mathcal A^c}=(\varphi_{j_{d+1}},\cdots,\varphi_{j_p})^\top$, where $\mathcal A=\{j_1,j_2,\cdots,j_d\}$. Then, $S=\Vert\boldsymbol\varphi\Vert^2=\Vert\boldsymbol\varphi_{\mathcal A}\Vert^2+\Vert\boldsymbol\varphi_{\mathcal A^c}\Vert^2$, $M=\Vert\boldsymbol\varphi+\sqrt{np}\zeta_1{\mathbf D}^{-1/2}\boldsymbol \theta\Vert_\infty=\max_{i\in \mathcal A}(\varphi_i+\sqrt{np}\zeta_1{\mathbf D}^{-1/2}\boldsymbol \theta)+\max_{i\in \mathcal A^c}\varphi_i$. From the proof of Theorem \ref{thm3}, we know that $\Vert\boldsymbol\varphi_{\mathcal A^c}\Vert^2$ and $\max_{i\in \mathcal A^c}\varphi_i$ are asymptotically independent. Hence, it suffice to show that $\Vert\boldsymbol\varphi_{\mathcal A^c}\Vert^2$ is asymptotically independent with $\boldsymbol\varphi_{\mathcal A}$.

By Lemma \ref{lemma_normal_decom}, $\boldsymbol\varphi_{\mathcal A^c}$ can be decomposed as $\boldsymbol\varphi_{\mathcal A^c}=\boldsymbol E+\boldsymbol F$, where $\boldsymbol E=\boldsymbol\varphi_{\mathcal A^c}-\mathbf\Sigma_{U,\mathcal A^c,\mathcal A}\mathbf\Sigma^{-1}_{U,\mathcal A,\mathcal A}\boldsymbol\varphi_{\mathcal A},\boldsymbol F=\mathbf\Sigma_{U,\mathcal A^c,\mathcal A}\mathbf\Sigma^{-1}_{U,\mathcal A,\mathcal A}\boldsymbol\varphi_{\mathcal A}$, $\mathbf\Sigma_{U}=p\mathbb E\boldsymbol U_1\boldsymbol U_1^\top=p\mathbf\Sigma_{u}$ , which fulfill the properties $\boldsymbol E\sim N(0,\mathbf\Sigma_{U,\mathcal A^c,\mathcal A^c}-\mathbf\Sigma_{U,\mathcal A^c,\mathcal A}\mathbf\Sigma^{-1}_{U,\mathcal A,\mathcal A}\mathbf\Sigma_{U,\mathcal A,\mathcal A^c})$, $\boldsymbol F\sim N(0,\mathbf\Sigma_{U,\mathcal A^c,\mathcal A}\mathbf\Sigma^{-1}_{U,\mathcal A,\mathcal A}\mathbf\Sigma_{U,\mathcal A,\mathcal A^c})$ and
$\boldsymbol E$ and $\boldsymbol \varphi_{\mathcal A}$ are independent.

Then, we rewrite
$$
\Vert\boldsymbol\varphi_{\mathcal A^c}\Vert^2=\boldsymbol E^\top \boldsymbol E+\boldsymbol F^\top \boldsymbol F+2\boldsymbol E^\top \boldsymbol F.
$$
According the proof of lemma S.7 in \cite{feng2022asymptotic}, we have that:
$$
\mathbb P(\vert \boldsymbol F^\top \boldsymbol F+2\boldsymbol E^\top \boldsymbol F\vert \geq\epsilon \nu_p)\leq \frac{3}{p^t}\rightarrow 0,
$$
by $d=o(\lambda_{\min}(\mathbf R)\text{tr}(\mathbf R^2)^{1/2}/(\log p)^C)$, where $\nu_p=[2\text{tr}(\mathbf R^2)]^{1/2}$ , $t=t_p:=C \epsilon/8 \cdot v_p/[\lambda_{\max }(\mathbf{R}) \log p]\rightarrow \infty$,$\epsilon_p:=(\log p)^C/[v_p \lambda_{\min }(\mathbf{R})]\rightarrow 0$.

\end{proof}
\subsection{Some useful  lemmas}

\begin{lemma}\label{LemmaA3}
    (Lemma A3. in \cite{cheng2023}) Suppose Assumptions \ref{max1}-\ref{max3} holds with $a_0(p)\asymp p^{1-\delta}$ for some positive constant $\delta\leq 1/2$. Define a random $p\times p$ matrix $\mathbf Q=n^{-1}\sum_{i=1}^n R^{-1}_i \boldsymbol U_i \boldsymbol U_i^\top$ and let $\mathbf Q_{jl}$ be the $(j,l)$th element of $\mathbf Q$. Then,

    (i)$\vert \mathbf Q_{jl}\vert\lesssim\zeta_1 p^{-1}\vert \sigma_{jl}\vert +O_p(\zeta_1n^{-1/2}p^{-1}+\zeta_1 p^{-7/6}+\zeta_1p^{-1-\delta/2})$.

(ii)$\mathbf Q_{jl}=\mathbf Q_{0,jl}+O_p(\zeta_1p^{-7/6}+\zeta_1p^{-1-\delta/2})$, where $\mathbf Q_{0,jl}$ is the $(j,l)$th element of
$$
\mathbf Q_0=n^{-1}p^{-1/2}\sum_{i=1}^n \nu_i^{-1} \{\mathbf D^{-1/2}\mathbf\Gamma U(\boldsymbol W_i)\}\{\mathbf D^{-1/2}\mathbf \Gamma U(\boldsymbol W_i)\}^\top.
$$
In addition, $\mathbf Q_0$ satisfies
$$
\text{tr}[\mathbb E(\mathbf Q_0^2)-\{\mathbb E(\mathbf Q_0)\}^2]=O(n^{-1}p^{-1}).
$$
\end{lemma}
\begin{lemma}\label{LemmaA4}
    (Lemma A4. in \cite{cheng2023})Suppose Assumptions \ref{max1}-\ref{max3} hold with $a_0(p)\asymp p^{1-\delta}$ for some positive constant $\delta\leq 1/2$. Then,

(i) $\mathbb E\{(\zeta_1^{-1} U_{i,j})^4\}\lesssim\bar M^2$ and $\mathbb E\{(\zeta_1^{-1} U_{i,j})^2\}\gtrsim \underline{m}$ for all $i=1,2,\cdots,n$ and $j=1,2,\cdots,p$.

(ii) $\Vert \zeta_{1}^{-1} U_{i,j}\Vert_{\psi_\alpha}\lesssim \bar B$ for all $i=1,2,
\cdots,n$ and $j=1,2,\cdots,p$.

(iii )$\mathbb E( U_{i,j}^2)=p^{-1}+O(p^{-1-\delta/2})$ for $j=1,2,\cdots,p$ and $\mathbb E( U_{i,j} U_{i,l})=p^{-1}\sigma_{j,l} +O(p^{-1-\delta/2})$ for $1\leq j\not=l\leq p$.

(iv) if $\log p=o(n^{1/3})$,
$$
\left\vert n^{-1/2}\sum_{i=1}^n\zeta_1^{-1}\boldsymbol U_i\right\vert_\infty =O_p\{\log^{1/2}(np)\}\text{ and } \left\vert n^{-1} \sum_{i=1}^n(\zeta_1^{-1}\boldsymbol U_i)^2\right\vert_\infty=O_p(1).
$$

\end{lemma}
\begin{lemma}\label{LemmaA6}
    (Nazarov's inequality) Let $\boldsymbol Y_0=(Y_{0,1},Y_{0,2},\cdots,Y_{0,p})^\top$ be a centered Gaussian random vector in $\mathbb R^p$ and $\mathbb E(Y_{0,j}^2)\geq b$ for all $j=1,2,\cdots,p$ and some constant $b>0$, then for every $y\in\mathbb R^p$ and $a>0$,
    $$
\mathbb P(\boldsymbol Y_0\leq y+a)-\mathbb P(\boldsymbol Y_0\leq y)\lesssim a\log^{1/2}(p).
    $$

\end{lemma}
\begin{lemma}(Theorem 2 in \cite{chernozhukov2015comparison})\label{LemmaJ2}
Let $\boldsymbol X=\left(X_1, \ldots, X_p\right)^\top$ and $\boldsymbol Y=\left(Y_1, \ldots, Y_p\right)^\top$ be centered Gaussian random vectors in $\mathbb{R}^p$ with covariance matrices $\mathbf \Sigma^X=$ $\left(\sigma_{j k}^X\right)_{1 \leq j, k \leq p}$ and $\mathbf\Sigma^Y=\left(\sigma_{j k}^Y\right)_{1 \leq j, k \leq p}$, respectively. In terms of $p$,
$$
\Delta:=\max _{1 \leq j, k \leq p}\left|\sigma_{j k}^X-\sigma_{j k}^Y\right|, \text { and } a_p:=\mathrm{E}\left[\max _{1 \leq j \leq p}\left(Y_j / \sigma_{j j}^Y\right)\right].
$$
Suppose that $p \geq 2$ and $\sigma_{j j}^Y>0$ for all $1 \leq j \leq p$. Then
$$
\begin{aligned}
& \sup _{x \in \mathbb{R}}\left|\mathrm{P}\left(\max _{1 \leq j \leq p} X_j \leq x\right)-\mathrm{P}\left(\max _{1 \leq j \leq p} Y_j \leq x\right)\right| \\
& \quad \leq C \Delta^{1 / 3}\left\{\left(1 \vee a_p^2 \vee \log (1 / \Delta)\right\}^{1 / 3} \log ^{1 / 3} p,\right.
\end{aligned}
$$
where $C>0$ depends only on $\min _{1 \leq j \leq p} \sigma_{j j}^Y$ and $\max _{1 \leq j \leq p} \sigma_{j j}^Y$ (the right side is understood to be 0 when $\Delta=0$ ). Moreover, in the worst case, $a_p \leq \sqrt{2 \log p}$, so that
$$
\sup _{x \in \mathbb{R}}\left|\mathrm{P}\left(\max _{1 \leq j \leq p} X_j \leq x\right)-\mathrm{P}\left(\max _{1 \leq j \leq p} Y_j \leq x\right)\right| \leq C^{\prime} \Delta^{1 / 3}\{1 \vee \log (p / \Delta)\}^{2 / 3},
$$
where as before $C^{\prime}>0$ depends only on $\min _{1 \leq j \leq p} \sigma_{j j}^Y$ and $\max _{1 \leq j \leq p} \sigma_{j j}^Y$.
\end{lemma}

\begin{lemma}\label{lemma_normal_decom}
    (Theorem 1.2.11 in \cite{muirhead2009aspects})Let $\boldsymbol{X} \sim N(\boldsymbol{\mu}, \mathbf{\Sigma})$ with invertible $\mathbf{\Sigma}$, and partition $\boldsymbol{X}, \boldsymbol{\mu}$ and $\mathbf{\Sigma}$ as
$$
\boldsymbol{X}=\left(\begin{array}{c}
\boldsymbol{X}_1 \\
\boldsymbol{X}_2
\end{array}\right), \quad \boldsymbol{\mu}=\left(\begin{array}{l}
\boldsymbol{\mu}_1 \\
\boldsymbol{\mu}_2
\end{array}\right), \quad \mathbf{\Sigma}=\left(\begin{array}{cc}
\mathbf\Sigma_{11} & \mathbf\Sigma_{12} \\
\mathbf\Sigma_{21} & \mathbf\Sigma_{22}
\end{array}\right) .
$$
Then $\boldsymbol{X}_2-\mathbf{\Sigma}_{21} \mathbf{\Sigma}_{11}^{-1} \boldsymbol{X}_1 \sim N\left(\boldsymbol{\mu}_2-\mathbf{\Sigma}_{21} \mathbf{\Sigma}_{11}^{-1} \boldsymbol{\mu}_1, \mathbf{\Sigma}_{22 \cdot 1}\right)$ and is independent of $\boldsymbol{X}_1$, where $\mathbf{\Sigma}_{22 \cdot 1}=$ $\mathbf\Sigma_{22}-\mathbf\Sigma_{21} \mathbf\Sigma_{11}^{-1} \mathbf\Sigma_{12}$.
\end{lemma}


\begin{thebibliography}{}

\bibitem[\protect\astroncite{Alon et~al.}{1999}]{alon1999broad}
Alon, U., Barkai, N., Notterman, D.~A., Gish, K., Ybarra, S., Mack, D., and
  Levine, A.~J. (1999).
\newblock Broad patterns of gene expression revealed by clustering analysis of
  tumor and normal colon tissues probed by oligonucleotide arrays.
\newblock {\em Proceedings of the National Academy of Sciences},
  96(12):6745--6750.

\bibitem[\protect\astroncite{Ayyala et~al.}{2017}]{ayyala2017mean}
Ayyala, D.~N., Park, J., and Roy, A. (2017).
\newblock Mean vector testing for high-dimensional dependent observations.
\newblock {\em Journal of Multivariate Analysis}, 153:136--155.

\bibitem[\protect\astroncite{Bai and Saranadasa}{1996}]{bai1996effect}
Bai, Z. and Saranadasa, H. (1996).
\newblock Effect of high dimension: by an example of a two sample problem.
\newblock {\em Statistica Sinica}, pages 311--329.

\bibitem[\protect\astroncite{Bickel and Levina}{2008}]{bickel2008covariance}
Bickel, P.~J. and Levina, E. (2008).
\newblock Covariance regularization by thresholding.
\newblock {\em The Annals of Statistics}, 36(6):2577--2604.

\bibitem[\protect\astroncite{Cai et~al.}{2013}]{CLX14}
Cai, T., Liu, W., and Xia, Y. (2013).
\newblock Two-sample test of high dimensional means under dependence.
\newblock {\em Journal of the Royal Statistical Society Series B: Statistical
  Methodology}, 76(2):349--372.

\bibitem[\protect\astroncite{Chang et~al.}{2024}]{chang2024central}
Chang, J., Chen, X., and Wu, M. (2024).
\newblock Central limit theorems for high dimensional dependent data.
\newblock {\em Bernoulli}, 30(1):712--742.

\bibitem[\protect\astroncite{Chang et~al.}{2023}]{chang2023testing}
Chang, J., Jiang, Q., and Shao, X. (2023).
\newblock Testing the martingale difference hypothesis in high dimension.
\newblock {\em Journal of Econometrics}, 235(2):972--1000.

\bibitem[\protect\astroncite{Chang et~al.}{2017}]{chang2017testing}
Chang, J., Yao, Q., and Zhou, W. (2017).
\newblock Testing for high-dimensional white noise using maximum
  cross-correlations.
\newblock {\em Biometrika}, 104(1):111--127.

\bibitem[\protect\astroncite{Chen and Qin}{2010}]{chen2010two}
Chen, S.~X. and Qin, Y.-L. (2010).
\newblock A two-sample test for high-dimensional data with applications to
  gene-set testing.
\newblock {\em The Annals of Statistics}, 38(2):808--835.

\bibitem[\protect\astroncite{Chen et~al.}{2010}]{chen2010tests}
Chen, S.~X., Zhang, L.-X., and Zhong, P.-S. (2010).
\newblock Tests for high-dimensional covariance matrices.
\newblock {\em Journal of the American Statistical Association},
  105(490):810--819.

\bibitem[\protect\astroncite{Cheng et~al.}{2019}]{cheng2019testing}
Cheng, G., Liu, B., Peng, L., Zhang, B., and Zheng, S. (2019).
\newblock Testing the equality of two high-dimensional spatial sign covariance
  matrices.
\newblock {\em Scandinavian Journal of Statistics}, 46(1):257--271.

\bibitem[\protect\astroncite{Cheng et~al.}{2023}]{cheng2023}
Cheng, G., Peng, L., and Zou, C. (2023).
\newblock Statistical inference for ultrahigh dimensional location parameter
  based on spatial median.
\newblock {\em arXiv preprint arXiv:2301.03126}.

\bibitem[\protect\astroncite{Chernozhukov
  et~al.}{2013}]{chernozhukov2013gaussian}
Chernozhukov, V., Chetverikov, D., and Kato, K. (2013).
\newblock Gaussian approximations and multiplier bootstrap for maxima of sums
  of high-dimensional random vectors.
\newblock {\em The Annals of Statistics}, 41(6):2786--2819.

\bibitem[\protect\astroncite{Chernozhukov
  et~al.}{2015}]{chernozhukov2015comparison}
Chernozhukov, V., Chetverikov, D., and Kato, K. (2015).
\newblock Comparison and anti-concentration bounds for maxima of gaussian
  random vectors.
\newblock {\em Probability Theory and Related Fields}, 162:47--70.

\bibitem[\protect\astroncite{Chernozhukov
  et~al.}{2017}]{chernozhukov2017central}
Chernozhukov, V., Chetverikov, D., and Kato, K. (2017).
\newblock Central limit theorems and bootstrap in high dimensions.
\newblock {\em Annals of probability: An official journal of the Institute of
  Mathematical Statistics}, 45(4):2309--2352.

\bibitem[\protect\astroncite{Cutting et~al.}{2017}]{cutting2017testing}
Cutting, C., Paindaveine, D., and Verdebout, T. (2017).
\newblock Testing uniformity on high-dimensional spheres against monotone
  rotationally symmetric alternatives.
\newblock {\em The Annals of Statistics}, 45(3):1024--1058.

\bibitem[\protect\astroncite{Fama and French}{1993}]{Fama1993}
Fama, E.~F. and French, K.~R. (1993).
\newblock Common risk factors in the returns on stocks and bonds.
\newblock {\em Journal of Financial Economics}, 33(1):3--56.

\bibitem[\protect\astroncite{Fama and French}{2015}]{Fama2015}
Fama, E.~F. and French, K.~R. (2015).
\newblock A five-factor asset pricing model.
\newblock {\em Journal of Financial Economics}, 116(1):1--22.

\bibitem[\protect\astroncite{Fan et~al.}{2015}]{Fan2015}
Fan, J., Liao, Y., and Yao, J. (2015).
\newblock Power enhancement in high dimensional cross-sectional tests.
\newblock {\em Econometrica}, 83(4):1497--1541.

\bibitem[\protect\astroncite{Fang}{2018}]{fang2018symmetric}
Fang, K.~W. (2018).
\newblock {\em Symmetric multivariate and related distributions}.
\newblock CRC Press.

\bibitem[\protect\astroncite{Feng et~al.}{2022a}]{feng2022asymptotic}
Feng, L., Jiang, T., Li, X., and Liu, B. (2022a).
\newblock Asymptotic independence of the sum and maximum of dependent random
  variables with applications to high-dimensional tests.
\newblock {\em arXiv preprint arXiv:2205.01638}.

\bibitem[\protect\astroncite{Feng et~al.}{2022b}]{feng2022a}
Feng, L., Jiang, T., Liu, B., and Xiong, W. (2022b).
\newblock Max-sum tests for cross-sectional independence of high-dimensional
  panel data.
\newblock {\em Annals of Statistics}, 50(2):1124--1143.

\bibitem[\protect\astroncite{Feng et~al.}{2022c}]{feng2022high}
Feng, L., Lan, W., Liu, B., and Ma, Y. (2022c).
\newblock High-dimensional test for alpha in linear factor pricing models with
  sparse alternatives.
\newblock {\em Journal of Econometrics}, 229(1):152--175.

\bibitem[\protect\astroncite{Feng et~al.}{2021}]{feng2021inverse}
Feng, L., Liu, B., and Ma, Y. (2021).
\newblock An inverse norm sign test of location parameter for high-dimensional
  data.
\newblock {\em Journal of Business \& Economic Statistics}, 39(3):807--815.

\bibitem[\protect\astroncite{Feng et~al.}{2022d}]{feng2022testingwhite}
Feng, L., Liu, B., and Ma, Y. (2022d).
\newblock Testing for high-dimensional white noise.
\newblock {\em arXiv preprint arXiv:2211.02964}.

\bibitem[\protect\astroncite{Feng and Sun}{2016}]{feng2016}
Feng, L. and Sun, F. (2016).
\newblock Spatial-sign based high-dimensional location test.
\newblock {\em Electronic Journal of Statistics}, 10:2420--2434.

\bibitem[\protect\astroncite{Feng et~al.}{2020}]{feng2020high}
Feng, L., Zhang, X., and Liu, B. (2020).
\newblock A high-dimensional spatial rank test for two-sample location
  problems.
\newblock {\em Computational Statistics \& Data Analysis}, 144:106889.

\bibitem[\protect\astroncite{Feng et~al.}{2016}]{feng2016multivariate}
Feng, L., Zou, C., and Wang, Z. (2016).
\newblock Multivariate-sign-based high-dimensional tests for the two-sample
  location problem.
\newblock {\em Journal of the American Statistical Association},
  111(514):721--735.

\bibitem[\protect\astroncite{Feng et~al.}{2015}]{feng2015two}
Feng, L., Zou, C., Wang, Z., and Zhu, L. (2015).
\newblock Two-sample behrens-fisher problem for high-dimensional data.
\newblock {\em Statistica Sinica}, 25:1297--1312.

\bibitem[\protect\astroncite{Hallin and
  Paindaveine}{2006}]{hallin2006semiparametrically}
Hallin, M. and Paindaveine, D. (2006).
\newblock Semiparametrically efficient rank-based inference for shape. i.
  optimal rank-based tests for sphericity.
\newblock {\em Annals of Statistics}, 34(6):2707--2756.

\bibitem[\protect\astroncite{He et~al.}{2021}]{he2021}
He, Y., Xu, G., Wu, C., and Pan, W. (2021).
\newblock Asymptotically independent u-statistics in high-dimensional testing.
\newblock {\em Annals of Statistics}, 49(1):151--181.

\bibitem[\protect\astroncite{Hettmansperger and
  Randles}{2002}]{hettmansperger2002practical}
Hettmansperger, T.~P. and Randles, R.~H. (2002).
\newblock A practical affine equivariant multivariate median.
\newblock {\em Biometrika}, 89(4):851--860.

\bibitem[\protect\astroncite{Huang et~al.}{2023}]{huang2023high}
Huang, X., Liu, B., Zhou, Q., and Feng, L. (2023).
\newblock A high-dimensional inverse norm sign test for two-sample location
  problems.
\newblock {\em Canadian Journal of Statistics}, 51(4):1004--1033.

\bibitem[\protect\astroncite{Ilmonen and
  Paindaveine}{2011}]{ilmonen2011semiparametrically}
Ilmonen, P. and Paindaveine, D. (2011).
\newblock Semiparametrically efficient inference based on signed ranks in
  symmetric independent component models.
\newblock {\em Annals of Statistics}, 39(5):2448--2476.

\bibitem[\protect\astroncite{Koike}{2021}]{koike2021notes}
Koike, Y. (2021).
\newblock Notes on the dimension dependence in high-dimensional central limit
  theorems for hyperrectangles.
\newblock {\em Japanese Journal of Statistics and Data Science}, 4:257--297.

\bibitem[\protect\astroncite{Li and Chen}{2012}]{li2012two}
Li, J. and Chen, S.~X. (2012).
\newblock Two sample tests for high-dimensional covariance matrices.
\newblock {\em Annals of Statistics}, 40(2):908--940.

\bibitem[\protect\astroncite{Li and Xu}{2022}]{li2022asymptotic}
Li, W. and Xu, Y. (2022).
\newblock Asymptotic properties of high-dimensional spatial median in
  elliptical distributions with application.
\newblock {\em Journal of Multivariate Analysis}, 190:104975.

\bibitem[\protect\astroncite{Liu et~al.}{2023}]{liu2023high}
Liu, B., Feng, L., and Ma, Y. (2023).
\newblock High-dimensional alpha test of linear factor pricing models with
  heavy-tailed distributions.
\newblock {\em Statistica Sinica}, 33:1389--1410.

\bibitem[\protect\astroncite{Liu and Xie}{2020}]{liu2020}
Liu, Y. and Xie, J. (2020).
\newblock Cauchy combination test: A powerful test with analytic p-value
  calculation under arbitrary dependency structures.
\newblock {\em Journal of the American Statistical Association},
  115(529):393--402.

\bibitem[\protect\astroncite{Ljung and Box}{1978}]{ljung1978measure}
Ljung, G.~M. and Box, G.~E. (1978).
\newblock On a measure of lack of fit in time series models.
\newblock {\em Biometrika}, 65(2):297--303.

\bibitem[\protect\astroncite{Long et~al.}{2023}]{li2023}
Long, M., Li, Z., Zhang, W., and Li, Q. (2023).
\newblock The cauchy combination test under arbitrary dependence structures.
\newblock {\em The American Statistician}, 77(2):134--142.

\bibitem[\protect\astroncite{Ma et~al.}{2023}]{ma2023adaptive}
Ma, H., Feng, L., and Wang, Z. (2023).
\newblock Adaptive testing for alphas in conditional factor models with high
  dimensional assets.
\newblock {\em arXiv preprint arXiv:2307.09397}.

\bibitem[\protect\astroncite{Ma et~al.}{2024}]{ma2024testing}
Ma, H., Feng, L., Wang, Z., and Jigang, B. (2024).
\newblock Testing alpha in high dimensional linear factor pricing models with
  dependent observations.
\newblock {\em arXiv preprint arXiv:2401.14052}.

\bibitem[\protect\astroncite{Muirhead}{2009}]{muirhead2009aspects}
Muirhead, R.~J. (2009).
\newblock {\em Aspects of multivariate statistical theory}.
\newblock John Wiley \& Sons.

\bibitem[\protect\astroncite{Nordhausen et~al.}{2009}]{nordhausen2009signed}
Nordhausen, K., Oja, H., and Paindaveine, D. (2009).
\newblock Signed-rank tests for location in the symmetric independent component
  model.
\newblock {\em Journal of Multivariate Analysis}, 100(5):821--834.

\bibitem[\protect\astroncite{Oja}{2010}]{oja2010multivariate}
Oja, H. (2010).
\newblock {\em Multivariate nonparametric methods with R: an approach based on
  spatial signs and ranks}.
\newblock Springer Science \& Business Media.

\bibitem[\protect\astroncite{Paindaveine and
  Verdebout}{2016}]{paindaveine2016high}
Paindaveine, D. and Verdebout, T. (2016).
\newblock On high-dimensional sign tests.
\newblock {\em Bernoulli}, 22(3):1745--1769.

\bibitem[\protect\astroncite{Park and Ayyala}{2013}]{park2013test}
Park, J. and Ayyala, D.~N. (2013).
\newblock A test for the mean vector in large dimension and small samples.
\newblock {\em Journal of Statistical Planning and Inference}, 143(5):929--943.

\bibitem[\protect\astroncite{Pesaran and Yamagata}{2017}]{Pesaran2017}
Pesaran, M.~H. and Yamagata, T. (2017).
\newblock Testing for alpha in linear factor pricing models with a large number
  of securities.
\newblock {\em Social Science Electronic Publishing}.

\bibitem[\protect\astroncite{Ross}{1976}]{Rose1976}
Ross, S.~A. (1976).
\newblock The arbitrage theory of capital asset pricing.
\newblock {\em Journal of Economic Theory}, 13(3):341--360.

\bibitem[\protect\astroncite{Sharpe}{1964}]{Sharpe1964CAPITALAP}
Sharpe, W.~F. (1964).
\newblock Capital asset prices: A theory of market equilibrium under conditions
  of risk.
\newblock {\em The Journal of Finance}, 19(3):425--442.

\bibitem[\protect\astroncite{Srivastava}{2009}]{srivastava2009test}
Srivastava, M.~S. (2009).
\newblock A test for the mean vector with fewer observations than the dimension
  under non-normality.
\newblock {\em Journal of Multivariate Analysis}, 100(3):518--532.

\bibitem[\protect\astroncite{Wang and Feng}{2023}]{wang2023}
Wang, G. and Feng, L. (2023).
\newblock Computationally efficient and data-adaptive changepoint inference in
  high dimension.
\newblock {\em Journal of the Royal Statistical Society Series B: Statistical
  Methodology}, 85(3):936--958.

\bibitem[\protect\astroncite{Wang et~al.}{2015}]{wang2015high}
Wang, L., Peng, B., and Li, R. (2015).
\newblock A high-dimensional nonparametric multivariate test for mean vector.
\newblock {\em Journal of the American Statistical Association},
  110(512):1658--1669.

\bibitem[\protect\astroncite{Wu et~al.}{2019}]{wu2019adaptive}
Wu, C., Xu, G., and Pan, W. (2019).
\newblock An adaptive test on high-dimensional parameters in generalized linear
  models.
\newblock {\em Statistica Sinica}, 29(4):2163--2186.

\bibitem[\protect\astroncite{Wu et~al.}{2020}]{wu2020regularization}
Wu, C., Xu, G., Shen, X., and Pan, W. (2020).
\newblock A regularization-based adaptive test for high-dimensional generalized
  linear models.
\newblock {\em The Journal of Machine Learning Research}, 21(1):5005--5071.

\bibitem[\protect\astroncite{Xu et~al.}{2016}]{xu2016adaptive}
Xu, G., Lin, L., Wei, P., and Pan, W. (2016).
\newblock An adaptive two-sample test for high-dimensional means.
\newblock {\em Biometrika}, 103(3):609--624.

\bibitem[\protect\astroncite{Yao et~al.}{2015}]{yao2015sample}
Yao, J., Zheng, S., and Bai, Z. (2015).
\newblock Sample covariance matrices and high-dimensional data analysis.
\newblock {\em Cambridge UP, New York}.

\bibitem[\protect\astroncite{Yu et~al.}{2022}]{yu2022jasa}
Yu, X., Li, D., and Xue, L. (2022).
\newblock Fisher's combined probability test for high-dimensional covariance
  matrices.
\newblock {\em Journal of the American Statistical Association}, Online
  publised(1-36).

\bibitem[\protect\astroncite{Yu et~al.}{2023}]{Yu2023PE}
Yu, X., Yao, J., and Xue, L. (2023).
\newblock Power enhancement for testing multi-factor asset pricing models via
  fisher’s method.
\newblock {\em Journal of Econometrics}, In press.

\bibitem[\protect\astroncite{Zhang and Cheng}{2018}]{zhang2018gaussian}
Zhang, X. and Cheng, G. (2018).
\newblock Gaussian approximation for high dimensional vector under physical
  dependence.
\newblock {\em Bernoulli}, 24(4A):2640--2675.

\bibitem[\protect\astroncite{Zhao et~al.}{2023}]{zhao2023spatial}
Zhao, P., Chen, D., and Wang, Z. (2023).
\newblock Spatial-sign based high dimensional white noises test.
\newblock {\em arXiv preprint arXiv:2303.10641}.

\bibitem[\protect\astroncite{Zou et~al.}{2014}]{zou2014multivariate}
Zou, C., Peng, L., Feng, L., and Wang, Z. (2014).
\newblock Multivariate sign-based high-dimensional tests for sphericity.
\newblock {\em Biometrika}, 101(1):229--236.

\end{thebibliography}
\end{document}